\documentclass[draftcls, onecolumn]{IEEEtran}
\usepackage{graphicx}

\usepackage{amsmath}
\usepackage{amssymb}
\usepackage{amsfonts}
\usepackage{epsfig}

\usepackage{amsmath,mathtools}    
\usepackage{graphicx}   
\usepackage{color}      
\usepackage{hyperref}
\usepackage{comment}
\usepackage{subfigure}
\usepackage{times}
\usepackage{hyphenat}
\usepackage{epsfig}
\usepackage{latexsym}
\usepackage{bbold,bbm}
\usepackage{setspace}
\usepackage{amssymb,amsmath}
\usepackage{mathrsfs}
\usepackage{amsgen,amsfonts,amsbsy,amsthm}
\usepackage{algorithmic,algorithm}
\usepackage{array}
\usepackage{booktabs}
\usepackage{amsfonts}
\usepackage{amssymb}
\usepackage{multirow}
\usepackage{longtable}
\usepackage{epstopdf}
\usepackage{wrapfig}
\usepackage[font=small]{caption}
\usepackage{tikz}
\usetikzlibrary{shapes,arrows}
\usepackage{enumerate}

\newtheorem{example}{Example}
\newtheorem{defn}{Definition}
\newtheorem{thm}{Theorem}

\newtheorem{lem}[thm]{Lemma}
\newtheorem{cor}[thm]{Corollary}
\newtheorem{note}{Remark}

\newtheorem{conj}{Conjecture}
\newtheorem{const}{Construction}

\newcommand{\bit}{\begin{itemize}}
	\newcommand{\eit}{\end{itemize}}
\newcommand{\bcor}{\begin{cor}}
	\newcommand{\ecor}{\end{cor}}
\newcommand{\beq}{\begin{equation}}
	\newcommand{\eeq}{\end{equation}}
\newcommand{\beqn}{\begin{equation*}}
	\newcommand{\eeqn}{\end{equation*}}
\newcommand{\bea}{\begin{eqnarray}}
	\newcommand{\eea}{\end{eqnarray}}
\newcommand{\bean}{\begin{eqnarray*}}
	\newcommand{\eean}{\end{eqnarray*}}
\newcommand{\ben}{\begin{enumerate}}
	\newcommand{\een}{\end{enumerate}}
\newcommand{\bdefn}{\begin{defn}}
	
	\renewcommand\footnotemark{}

	\newcommand{\vzero}{\mbox{$V_0$}}
	\newcommand{\vone}{\mbox{$V_1$}}
	\newcommand{\vtwo}{\mbox{$V_2$}}
	\newcommand{\vi}{\mbox{$V_i$}}
	\newcommand{\vs}{\mbox{$V_{s}$}}
	\newcommand{\vsmone}{\mbox{$V_{s-1}$}}
	\newcommand{\vinfty}{\mbox{$V_{\infty}$}}
	
	\newcommand{\eone}{\mbox{$E_1$}}
	\newcommand{\etwo}{\mbox{$E_2$}}
	
	\newcommand{\ej}{\mbox{$E_j$}}
	\newcommand{\es}{\mbox{$E_s$}}
	\newcommand{\espone}{\mbox{$E_{s+1}$}}
	\newcommand{\esmone}{\mbox{$E_{s-1}$}}

	\newcommand{\jzero}{\mbox{${\cal J}_0$}}

	\newcommand{\nbase}{\mbox{$N_{\text{base}}$}}
	\newcommand{\naux}{\mbox{$N_{\text{aux}}$}}

\newcommand{\aux}{\mbox{${\cal A}$}}
	\newcommand{\auxi}{\mbox{${\cal A}_i$}}
	
\newcommand{\hinf}{\mbox{$H_{\infty}$}}

		\newcommand{\hzero}{\mbox{$H_0$}}
		
	\newcommand{\calg}{\mbox{$\mathcal{G}$}}
		\newcommand{\gone}{\mbox{${\cal G}_1$}}
		\newcommand{\gzero}{\mbox{${\cal G}_0$}}
	\newcommand{\gtwo}{\mbox{${\cal G}_2$}}
	\newcommand{\gi}{\mbox{${\cal G}_i$}}
	\newcommand{\gs}{\mbox{${\cal G}_s$}}
		\newcommand{\gsmone}{\mbox{${\cal G}_{s-1}$}}
	\newcommand{\gsub}{\mbox{${\cal G}_{\text{sub}}$}}
       \newcommand{\ginf}{\mbox{${\cal G}_{\infty}$}}

	\newcommand{\gbase}{\mbox{${\cal G}_{\text{base}}$}}

	\newcommand{\calc}{\ensuremath{\mathcal{C}}}

		\newcommand{\nkrt}{\ensuremath{(n,k,r,t)_{\text{seq}}}}

	\newcommand\scalemath[2]{\scalebox{#1}{\mbox{\ensuremath{\displaystyle #2}}}}
	
\begin{document}
	\sloppy
	\title{A Tight Rate Bound and Matching Construction for Locally Recoverable Codes with Sequential Recovery From Any Number of Multiple Erasures}
	\author{S. B. Balaji, Ganesh R. Kini and P. Vijay Kumar
 \thanks{S. B. Balaji and P. Vijay Kumar are with the Department of ECE, Indian Institute of Science, Bangalore: 560 012, India (email: \{balaji.profess, pvk1729\}@gmail.com).}
\thanks{Ganesh R. Kini is with Maxlinear Inc, Bangalore (email: kiniganesh94@gmail.com).}
\thanks{P. Vijay Kumar is also a Visiting Professor at the University of Southern California. This research is supported in part by the National Science Foundation under Grant 1421848 and in part by an India-Israel UGC-ISF joint research program grant. The work of S. B. Balaji is supported by a TCS Research Scholarship. } 
\thanks{This paper was presented in part at the 2017 IEEE International Symposium on Information Theory \cite{BalKinKum_ISIT} and in part in 2018 National Conference on Communications \cite{BalKinKum_NCC}.}}

	\maketitle
	\newpage
	

\begin{abstract}
By a locally recoverable code (LRC), we will in this paper, mean a linear code in which a given code symbol can be recovered by taking a linear combination of at most $r$ other code symbols. The parameter $r$ is typically, significantly smaller than the dimension $k$ of the code, hence a code with locality enables recovery by contacting a smaller number of helper nodes.   A natural extension is to the local recovery of a set of $t$ erased symbols. There have been several approaches proposed for the handling of multiple erasures.  The approach considered here, is one of sequential recovery meaning that the $t$ erased symbols are recovered in succession, each time contacting at most $r$ other symbols for assistance in recovery.   Under the constraint that each erased symbol be recoverable by contacting at most $r$ other code symbols, this approach is the most general and hence offers maximum possible code rate. We characterize the maximum possible rate of an LRC with sequential recovery for any $r \geq 3$ and $t$. We do this by first deriving an upper bound on code rate and then going on to construct a {\em binary} code that achieves this optimal rate.   The upper bound derived here proves a conjecture made earlier relating to the structure (but not the exact form) of the rate bound.   Our approach also permits us to deduce the structure of the parity-check matrix of a rate-optimal LRC with sequential recovery. 

The parity-check matrix in turn, leads to a graphical description of the code.  The construction of a binary code having rate achieving the upper bound derived here makes use of this description.  Interestingly, it turns out that a subclass of binary codes that are both rate and block-length optimal, correspond to graphs known as Moore graphs that are regular graphs having the smallest number of vertices for a given girth.   A connection with Tornado codes is also made in the paper. 
\end{abstract}


	\section{Introduction}\label{sec:introduciton}

	Large-scale data centers such as those operated by Google, Amazon and Microsoft  are examples of distributed storage systems (DSS) that have become commonplace in the current information era and which play an important role in our everyday computational and data-retrieval tasks.    Apart from the need to store data in reliable fashion, a data center also seeks to minimize the storage overhead arising from the use of redundancy to achieve reliability.     The industry is increasingly turning towards the use of erasure codes for reducing this storage overhead while maintaining reliability given the current explosion in the amount of data to be stored and the cost of storing this data reliably.  The employment of erasure coding in Hadoop 3.0 in the form of the Hadoop Distributed File System - Erasure Coding (HDFS-EC) is an indication of this trend.   Maximum Distance Separable (MDS) codes such as Reed-Solomon (RS) codes are commonly employed since MDS codes minimize storage overhead for a given level of reliability. 
	
Yet another challenge faced by a data center is the relatively-frequent occurrence of individual node or storage-unit failure.   The conventional repair of an RS code is inefficient in terms of using resources when it comes to node repair.  Two approaches to coding have been proposed, to enable more efficient node repair in the case of single-node failures.  These are regenerating codes by Dimakis et. al. ~\cite{Dimakis_Regen} and codes with locality (known more commonly as locally repairable codes (LRC) by Gopalan et. al. \cite{GopHuaSimYek}.   Regenerating codes attempt to minimize the {\em amount of data download} needed to carry out node repair while codes with locality (also known as locally repairable codes or LRCs) aim to minimize the {\em number of nodes accessed} during node repair. The present paper deals with LRCs.  While the initial focus on LRC was on the repair under single-node failures, there is interest in multiple-node failure as well. This is because (i) simultaneous node failures can and do take place due to the increasing trend towards replacing expensive servers with low-cost commodity servers, (ii) some nodes in the system can be temporary unavailable either because they are down for maintenance or else are busy serving other demands placed on the data stored in these nodes.  In the present paper, we will focus on LRC for the repair of multiple node failures, i.e., LRC designed for recovery from multiple erasures.   


\subsection{Background on Single-Erasure LRC} \label{sec:LRC_single}

The notion of codes with locality was introduced in \cite{GopHuaSimYek} (See also \cite{HanMon,HuaChenLi,OggDat}) to design codes such that the number of nodes accessed to repair a failed node is much smaller than the dimension $k$ of the code . 
Let $\mathcal{C}$ be an $[n, k, d_{\min}]$ linear code over $\mathbb{F}_q$ having block length $n$, dimension $k$ and minimum distance $d_{\min}$. The $i^{\text{th}}$ code-symbol $c_i$, $1 \leq i \leq n$, of $\mathcal{C}$ is said to have locality $r$ if there exists $c_{i_1},\ldots,c_{i_{\ell}}$, $\ell \leq r$ with $i_1,\ldots,i_{\ell}$ distinct from $i$ such that $c_i = \sum_{j=1}^{\ell} a_j c_{i_j}$, $a_j \in \mathbb{F}_q$. If the set of $k$ message symbols in a systematic code have locality $r$ then $\mathcal{C}$ is said to have  information-symbol locality $r$. A code having information-symbol locality $r$ has minimum distance $d_{\min}$ upper bounded~\cite{GopHuaSimYek} by
\begin{equation} \label{eq:gopalan_bound}
d_{\min} \leq   n- k -\left\lceil \frac{k}{r} \right\rceil + 2.
\end{equation}
The pyramid-code construction in~\cite{HuaChenLi} yields optimal codes with information-symbol locality with field size $O(n)$ for all $\{n,k,r\}$.   
Codes with all-symbol locality (also called LRC) in which all code symbols, not just the message symbols, have locality $r$ is also studied in \cite{GopHuaSimYek}. For the case when $(r+1) \mid n$, construction of codes with all-symbol locality achieving the bound in \eqref{eq:gopalan_bound} can be found in \cite{GopHuaSimYek}. However these codes have field size exponential in $n$. In an earlier publication \cite{Reliab}, the authors provided a construction of codes with all-symbol locality having field size of $O(n)$. This construction is based on splitting the rows of the parity-check matrix of an MDS code and yields codes with all-symbol locality having field size $O(n)$ when $n = \lceil \frac{k}{r} \rceil (r+1)$. Optimality of this construction w.r.t \eqref{eq:gopalan_bound} is shown in \cite{PraKamLalKum}. A general construction of codes with all-symbol locality having field size of $O(n)$  achieving the bound \eqref{eq:gopalan_bound} for the case when $(r+1) | n$ can be found in \cite{TamBar_LRC}. This construction can be viewed as starting with a Reed-Solomon (RS) code and then restricting attention to a subcode that has all-symbol locality. Also contained in \cite{TamBar_LRC} is a construction of codes with all-symbol locality having field size of $O(n)$ whose minimum distance differs by at most $1$ from the bound given in \eqref{eq:gopalan_bound} when $r \nmid k$, $n \neq 1 \pmod{r+1}$. In \cite{ErnWesHol}, the authors show that codes with all-symbol locality whose minimum distance differs by atmost $1$ from the bound given in \eqref{eq:gopalan_bound} can be constructed for any $n,k,r$. However the construction provided in \cite{ErnWesHol} has field size exponential in $n$.

 When $(r+1) \nmid n$, the bound \eqref{eq:gopalan_bound} is not achievable in general. The bound in \eqref{eq:gopalan_bound} has been improved upon and tighter upper bounds can be found in \cite{PraLalKum,WanZhaI,ZhaWanGe,MehArd}. Constructions of LRC with field size exponential in $n$ achieving the tighter upper bound on minimum distance derived in \cite{WanZhaI} for the case of $n_1 > n_2$ where $n_1 = \lceil \frac{n}{r+1} \rceil$, $n_2=n_1(r+1)-n$ can also be found in \cite{WanZhaI}. Similarly when $n \mod (r+1) \geq k+\lceil \frac{k}{r}  \rceil \mod (r+1)$, a tighter upper bound on minimum distance of LRC and LRC achieving the upper bound can be found in \cite{MehArd}. In  \cite{PapDim} and \cite{ForYek}, the authors study codes with locality in the setting of nonlinear codes. In \cite{ForYek}, authors show that same upper bound as in \eqref{eq:gopalan_bound} continues to hold for non-linear codes. The implementation and performance evaluation of codes with locality in the context of Windows Azure storage can be found in \cite{Hua_etal}. The implementation and performance evaluation of codes with locality is carried out in the Hadoop Distributed File System in \cite{SatAst}.

The notion of maximal recoverable codes (MRC) in the context of codes with locality was introduced in \cite{CheMinHuaChenLIJin}. Roughly speaking, an MRC is an LRC which is as MDS as possible. More precisely, given a code with locality, let us define an admissible pattern as a set of code symbols which are missing at least one code symbol from each local code.   An MRC can then be defined as an LRC having the property that the code restricted to the co-ordinates of any admissible pattern, is an MDS code. Some constructions of MRC can be found in \cite{GopalanHJY13,BlaHafHet,HuY16,GokOzan,BalKUmMax}.   The problem of constructing MRCs with $O(n)$ field size remains open. Based on the above discussion, it can be seen that the problem of constructing optimal LRC for recovery from single erasures is largely settled and the focus has shifted to LRC that can handle multiple erasures.

\subsection{Background on LRC for Multiple Erasures} \label{sec:LRC_multiple} 

There are multiple approaches to designing LRC that provide deterministic local recovery from multiple erasures. The principal approaches that  have been studied in the literature are discussed below.  See Fig.~\ref{fig:LRCClassification} for a pictorial representation of the different approaches.

\noindent {\em $(r,\delta)$ Codes}

An early notion was that of codes with $(r, \delta)$ locality, introduced in \cite{PraKamLalKum} as a generalization of LRC for single erasures.    Here each local code is designed to be able to be able to recover from a larger number of erasures.  
An $[n,k]$ code \calc\  is said to have $(r, \delta)$ locality if for every $i \in [n]$, there exists a subset $S_i \subseteq [n]$ containing $i$ satisfying: 
\bea
\label{eq:rdeltaloc}
\dim(\mathcal{C} |_{S_i} ) \le r, \ \ \ \ \  d_{\min}(\mathcal{C} |_{S_i}) \ge \delta,
\eea
where $\mathcal{C} |_{S}$ (referred to as local code) is the code obtained by restricting $\mathcal{C}$ to the co-ordinates in $S$. Thus codes with $(r,\delta)$ locality replaces local codes which were single parity check codes with codes of minimum distance at least $\delta$ and dimension at most $r$. It can be seen that in codes with $(r,\delta)$ locality, any $\delta-1$ erasures in the local code $\mathcal{C}|_{S_i}$ can be recovered by contacting at most $r$ other code symbols in $\mathcal{C}|_{S_i}$. This follows as any $n-(\delta-1)$ columns of the generator matrix of the code $\mathcal{C}|_{S_i}$ has rank equal to $\dim(\mathcal{C} |_{S_i} )$ ($\leq r$). Upper bounds on minimum distance and optimal constructions having field size of $O(n)$ for codes with $(r,\delta)$ locality can be found in  \cite{PraKamLalKum,SonDauYueLi,TamBar_LRC,BinShuJieFan}. Codes with $(r,\delta)$ locality can guarantee local recovery from any $t$ erasures, if one chooses $\delta \geq t+1$. However this is inefficient in terms of rate. But recovery from a set of $t$ erased code symbols $E$ using at most $r$ code symbols (local recovery) for recovering each erased symbol can be achieved when $|E \cap S_i| \leq \delta-1$ for each $i$. 
Hence an alternative view is that codes with $(r,\delta)$ locality can be regarded as offering probabilistic guarantees of local recovery from $\leq t$ erasures. A generalization of codes with $(r,\delta)$ locality, termed as codes with hierarchical locality is introduced in \cite{SasAgaKumH}. In general for recovering from $t$ erasures, the average number of symbols contacted per erasure is smaller in codes with hierarchical locality compared to codes with $(r, \delta)$ locality. 

\noindent {\em LRC with Sequential-Recovery:}  An $(n,k,r,t)$ LRC with sequential-recovery (abbreviated as seq-LRC and denoted by $(n,k,r,t)_{\text{seq}}$) is an $[n,k]$ linear code $\calc$ having the following property: There is a permutation $(c_{\ell_1}, c_{\ell_2}, \cdots, c_{\ell_s})$ of any given set of $s \leq t$ erased symbols such that for every $j \in [s]$, there exists a subset $R_{j} \subseteq [n]$ satisfying (i) $|R_{j}| \le r$ , (ii) $\ R_{j} \cap \{\ell_{j}, \ell_{j+1}, \cdots, \ell_s \} = \phi$,  and 
\begin{eqnarray} \label{eq:locality}
\text{(iii)  \  }  \ c_{\ell_{j}} & = & \sum \limits_{i \in R_{j}} a_{i} c_{i}, \ a_{i} \in \mathbb{F}_q  .
\end{eqnarray} 
The definition guarantees that any set of $s \leq t$ erased code symbols $c_{\ell_1}, c_{\ell_2}, \cdots, c_{\ell_s}$, for $1 \leq s \leq t$ can be recovered by an $(n,k,r,t)_{\text{seq}}$ code by using \eqref{eq:locality} to recover the symbols $c_{\ell_{j}}, \ j=1,2,\cdots,s$, in succession. A little thought will show that LRC with sequential recovery, form the most general class of codes which can recover from a set of $t$ erased symbols by contacting small number of symbols for the recovery of each erased symbol. For this reason, LRC with sequential-recovery have the maximum possible rate and minimum distance among the class of LRC for multiple erasures.

\begin{figure}[ht]
	\centering
	\includegraphics[width=4in]{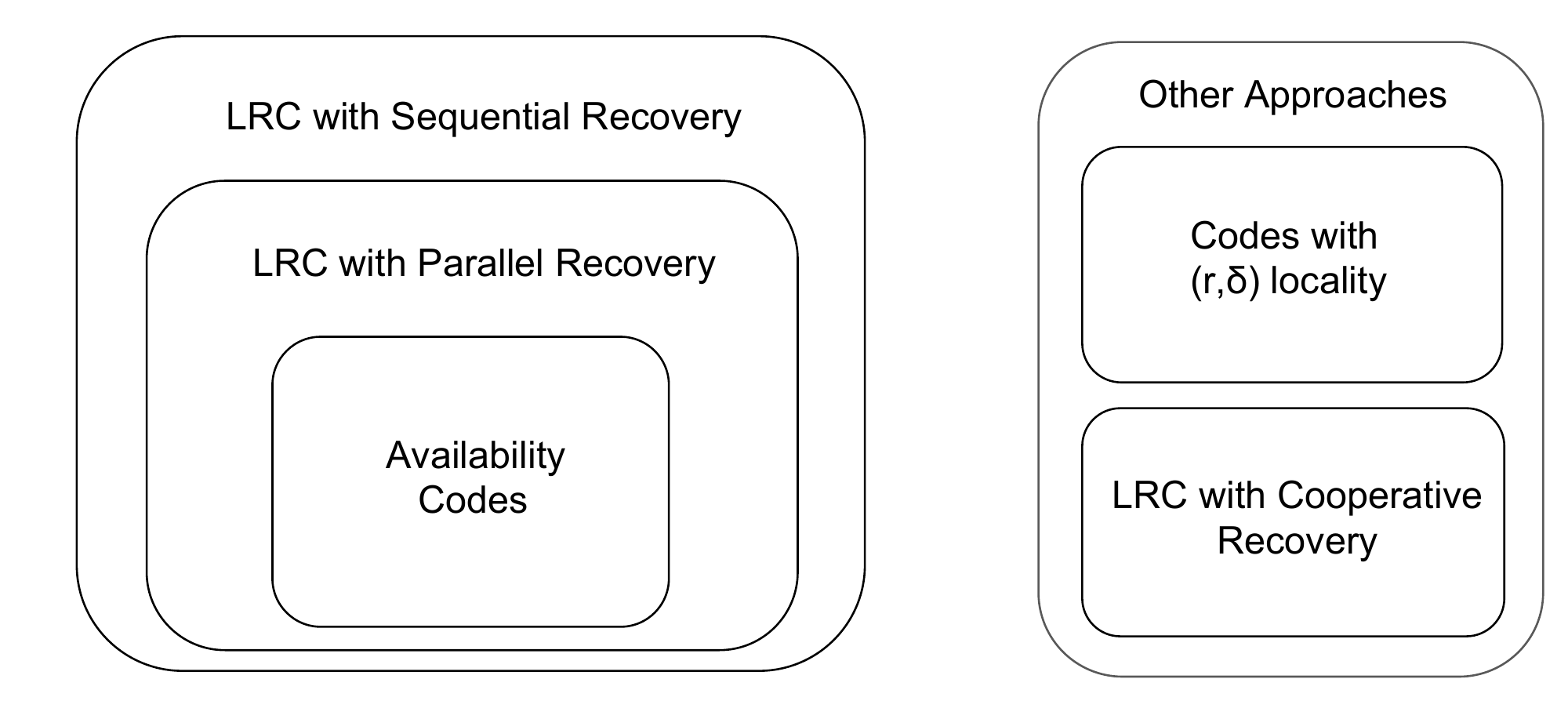}
	\caption{The various code classes corresponding to different approaches to recovery from multiple erasures.}
	\label{fig:LRCClassification}
\end{figure}

\noindent {\em LRC with Parallel-Recovery:} If we replace the condition (ii) in \eqref{eq:locality} by the more constrained requirement $R_{j} \cap \{\ell_1, \ell_2, \cdots, \ell_s \} = \phi$ in the definition of the $(n,k,r,t)_{\text{seq}}$ code, then the LRC will be called as a LRC with \noindent {\em parallel} recovery, abbreviated as par-LRC and denoted as $(n,k,r,t)_{\text{par}}$. Clearly the class of $(n,k,r,t)_{\text{seq}}$ codes contains the class of par-LRC.  From an implementation perspective, par-LRC may be preferred as the erased code symbols can be recovered in parallel. However, the advantage of parallel recovery comes with a rate penalty in comparison with an $(n,k,r,t)_{\text{seq}}$.   This could be a significant downside when the amount of data stored is large and there is a premium placed on minimizing storage overhead.  We note that depending upon the specific code, parallel recovery may need the same code symbol to be used in the recovery of more than one erased code symbol $c_{\ell_{j}}$.  Within industry, there is increased interest in a subclass of par-LRC called availability codes.
\linebreak
\linebreak
\noindent {\em Availability Codes:} An $(n,k,r,t)$ \noindent {\em availability} code (denoted by $(n,k,r,t)_{\text{avl}}$), is an $[n,k]$ linear code such that for each code symbol $c_{\ell}$, $\ell \in [n]$, there exist $t$ repair sets $\{R^{\ell}_j\}_{j=1}^t$ which are pairwise disjoint and of cardinality $|R^{\ell}_j| \leq r$ with $R^{\ell}_j \subseteq [n] \setminus \{\ell\}$ such that for every $j, 1 \leq j \leq t$, $c_{\ell}$ can be written in the form:
\bean
c_{\ell} = \sum\limits_{i \in R^{\ell}_j} a_{i} c_{i}, \ \ a_{i} \in \mathbb{F}_q.
\eean
An availability code can generate $t$ copies of a desired code symbol, i.e., make $t$ copies available, by calling upon $t$ disjoint recovery sets.   The following upper bound on the rate of an $(n,k,r,t)_{\text{avl}}$ code was given in \cite{TamBarFro}. 
	\begin{thm} \cite{TamBarFro} If \calc\ is an $(n,k,r,t)_{\text{avl}}$ code, then: 
		\bea
		\frac{k}{n} \leq \frac{1}{\prod_{j = 1}^{t}(1+\frac{1}{jr})}. \label{TamoBargRate}
		\eea
	\end{thm}
Fig~\ref{fig:availability} compares the upper bound on rate of an $(n,k,r, t=10)_{\text{avl}}$ code given in \eqref{TamoBargRate} with an achievable upper bound on rate of an $(n,k,r, t=10)_{\text{seq}}$ code given in \eqref{Thm1} in this paper.  As can be seen LRC with sequential recovery can be designed with significantly higher rate in comparison with availability codes.  This is not surprising since $(n,k,r,t)$ availability codes are a subclass of par-LRC.   This last statement follows because in an availability code, for any given set of $t$ erased symbols, there will be at least one repair set for every erased code symbol with all the symbols in the repair set unerased.   	
	\begin{figure}[h!]
			\begin{center}
				\includegraphics[width=5in]{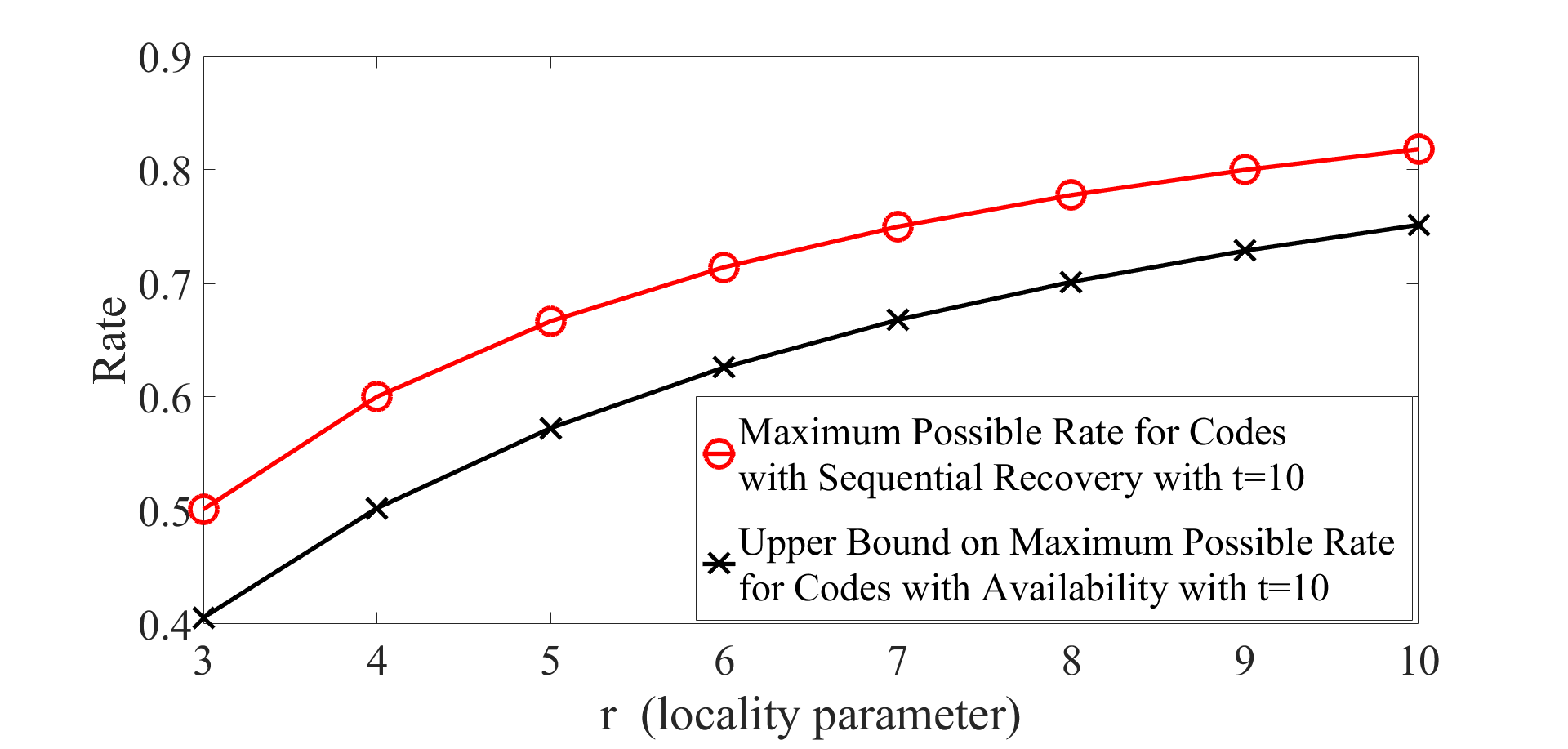} 
				\caption{Comparison of maximum possible rate of LRC with sequential recovery given in \eqref{Thm1} in this paper and an upper bound on maximum possible rate of availability codes given in \eqref{TamoBargRate} for $t=10$.}
				\label{fig:availability}
			\end{center}
		\end{figure}
\linebreak
\linebreak
\noindent {\em LRC with Cooperative Recovery:} An $(n,k,\rho,t)$ LRC with \noindent {\em cooperative} recovery is an $[n,k]$ linear code such that for each set $\{c_{j_1}, c_{j_2}, \cdots, c_{j_s}\}$, $1 \leq s \le t$ of erased code symbols, there exists a set $\{c_{i_1}, c_{i_2}, \cdots, c_{i_{\rho}}\}$ of $\rho$ unerased code symbols 
(i.e., $j_a \ne i_b$ for any $a, b$) such that for each $\ell \in [s]$:
\bean
c_{j_{\ell}} = \sum\limits_{b=1}^{\rho} \gamma_{\ell,b} \ c_{i_b}, \ \ \gamma_{\ell,b} \in \mathbb{F}_q.  
\eean
Thus LRC with cooperative recovery codes may be viewed as codes that aim to minimize the total number of symbols $\rho$ accessed for the recovery of a collection of $t$ erased code symbols, as opposed to minimizing the number of symbols accessed for the recovery of each individual erased code symbol.

	\subsection{Overview of Results}
	
  In this paper (Section \ref{sec:tight_bound}), we derive an upper bound on the maximum possible rate of a seq-LRC \calc\ for any $r \geq 3$, $t$.   We then make the observation that the parity-check matrix of a seq-LRC that achieves this upper bound on rate must necessarily possess a certain sparse, staircase-like form.   The form of the parity-check matrix is not sufficient however, to guarantee that the resultant code will be able to recover sequentially from $t$ erasures.  The structure of the p-c (parity check) matrix leads to a graphical description of the code \calc\ (Section \ref{sec:graphical_rep}).   The structure of parity-check matrix does not however,   fully specify the graph \calg.  A scale factor $a_0$ that determines the total number of vertices in the graph remains unspecified, as are certain edge connections.   Thus the task in code construction is to identify a suitable $a_0$ and nail down the edge connections, in such a way that the resultant code can recover sequentially from $t$ erasures.  It turns out that if the parameter $a_0$ and the edge connections are chosen so as to ensure that the graph \calg\ has girth $\geq (t+1)$, then the code is guaranteed to always be able to recover sequentially from $t$ erasures  (Section \ref{sec:girth_req}).  Moreover, under this girth condition, the associated rate-optimal code can be chosen to be a binary code.    We show how to construct graphs \calg\ having the desired form and of girth $\geq t+1$ (Section \ref{sec:complete_code}).  This shows the rate bound derived here to be tight and moreover achievable using binary codes.   These results prove a conjecture appearing in \cite{SongCaiYue_L} relating to an upper bound on the rate of a seq-LRC. It turns out that there are certain numerical values of $r,t$ for which the associated graph \calg\ can be chosen to be a Moore graph (Section  \ref{sec:Moore_ch4}).  Moore graphs are regular graphs having the smallest possible number of vertices for a given girth.  Whenever the associated graph is a Moore graph, it turns out that the resultant binary seq-LRC is optimal not only in terms of rate, it also has the smallest possible block length. In Section \ref{sec:high_dimension}, we give some hand-crafted examples of seq-LRC having maximum possible dimension for certain block lengths for $t=4$, $r=3$. We also make a connection with Tornado codes by noting certain structural similarities between the graphs associated with the two classes of codes (Section \ref{sec:Tornado}). 

  
\section{Background on Seq-LRC} \label{sec:background} 
The sequential approach to recovery from multiple erasures was introduced by Prakash et al. \cite{PraLalKum} and has been further investigated in  \cite{RawMazVis,SonYue,SonYue_Prdr,BalPraKum,BalPraKum4Era,SongCaiYue_L,BalKinKum_ISIT,BalKinKum_NCC} as discussed below.

\paragraph{Two Erasures} Seq-LRC with $t=2$ are considered in \cite{PraLalKum} (see also \cite{SonYue}) where a tight upper bound on the rate and a matching construction achieving the upper bound on rate for $t=2$ is provided.  A lower bound on block length and a construction achieving the lower bound on block length for $t=2$ is provided in \cite{SonYue}. 

\paragraph{Three Erasures} Seq-LRC with $t=3$ can be found discussed in \cite{SonYue,SongCaiYue_L,BalPraKum}.  A lower bound on block length as well as a construction achieving the lower bound on block length for $t=3$ appears in \cite{SonYue}. For some values of $k,r$, the lower bound given in \cite{SonYue} is not achievable. For these ranges of $k,r$, a tighter lower bound was derived in \cite{BalPraKum} and a construction which comes close to this lower bound for an infinite set of values of $k,r$ can also be found in \cite{BalPraKum}. A tight upper bound on rate of a seq-LRC with $t=3$, can be found in \cite{SongCaiYue_L}.

\paragraph{More Than $3$ Erasures}  The following conjecture on the maximum achievable rate of an $(n,k,r,t)_{\text{seq}}$ code appeared in \cite{SongCaiYue_L}. 
  	\begin{conj} \cite{SongCaiYue_L} [Conjecture] \label{conj:Song}
  		Let $\cal{C}$ denote an $(n,k,r,t)_{\text{seq}}$ code over a finite field $\mathbb{F}_q$. Let $ m=\lceil \log_r(k) \rceil$. Then an achievable upper bound on $\frac{k}{n}$ is given by:
  		\bean
  		\frac{k}{n} \leq \frac{1}{1+\sum_{i=1}^{m}\frac{a_i}{r^i}}, \text{ where } a_i \geq 0, \ \ a_i \in \mathbb{Z}, \ \ \sum_{i=1}^{m} a_i = t.
  		\eean
  	\end{conj}
  	As will be seen, the upper bound on rate of an $(n,k,r,t)_{\text{seq}}$ code derived in the present paper not only proves the above conjecture, it also (a) identifies the precise value of the coefficients $a_i$ appearing in the conjecture and (b) provides matching constructions of {\em binary} codes that achieve the upper bound on code rate for any $r,t$ with $r \geq 3$.

  Constructions of seq-LRC for any $r,t$ appears in \cite{SonYue_Prdr,SongCaiYue_L}. Although these constructions are interesting, based on the tight upper bound on code rate presented here, it can be seen that the constructions provided in \cite{SonYue_Prdr,SongCaiYue_L} do not achieve the maximum possible rate of a seq-LRC. In \cite{RawMazVis}, the authors provide a construction of seq-LRC for any $r,t$ with rate $\geq \frac{r-1}{r+1}$.  
We show in this paper (Section  \ref{sec:Moore_ch4}), that the rate of the construction given in \cite{RawMazVis} is actually $\frac{r-1}{r+1}+\frac{1}{n}$ and further, that the construction in \cite{RawMazVis} can be generalized by replacing the regular bipartite graph appearing in the construction by a regular graph.  It turns out that there are regular graphs for which the attained code rate of $\frac{r-1}{r+1}+\frac{1}{n}$ is optimal for certain parameters $(r,t)$.  These turn out to precisely the parameter sets for which a class of regular graphs, known as Moore graphs exist.  In the context of the present paper, Moore graphs are regular graphs of degree $(r+1)$ having girth $\geq t+1$ and the smallest possible number of vertices. Unfortunately, Moore graphs exist (Theorem \ref{thm:MooreExis}) only for a very sparse set of $(r,t)$ parameters.

Throughout the paper, by weight, we will mean Hamming weight. The notation $V(G)$ for a graph $G$ refers to the set of all vertices in the graph $G$. We use the abbreviation p-c for parity-check. We use the term girth of a graph $G$ to refer to the length (number of edges) of the cycle in the graph $G$ having shortest length. 

\section{Illustrative Examples of Rate-Optimal seq-LRC} \label{sec:examples}

A rate-optimal seq-LRC refers to a seq-LRC having maximum possible rate for a given $r,t$ where the maximization is over all possible field sizes and over all possible block lengths $n$.
This section provides two examples that illustrate the main result of this paper corresponding to binary, and which correspond to rate-optimal seq-LRC having parameters $t=4,r=6$ and $t=3,r=3$ respectively.

\subsection{Illustrative Example for $t$ Even}
Let \calc\ be a binary, rate-optimal seq-LRC with parameters $t=4$ and $r=6$.
For this choice of parameters, it will be shown in Corollary \ref{cor:equality_conditions} that the p-c matrix of the code takes on the form: 
\bea
	H = \left[
	\begin{array}{c|c|c}
		D_0 & A_1 & 0  \\
		\cline{1-3}
		0 & D_1 & C 
	\end{array} 
	\right],  
	\eea
	where
	\ben
	\item $D_0$ is an $(a_0 \times a_0)$ diagonal matrix,
	\item $D_1$ is an $(a_0r \times a_0r)$ diagonal matrix,
	\item $A_1$ is an $(a_0 \times a_0r)$ matrix with each row of weight $r$ and each column of weight $1$,
	\item $C$ is an $(a_0r \times \frac{a_0 r^2}{2})$ matrix with each row of weight $r$ and each column of weight $2$.
	\een
The integer parameter $a_0$ determines the block length $n=a_0(1+r+\frac{r^2}{2})$ of the code.   It turns out that that the rank of the above p-c matrix is equal to the number of rows.  Thus this seq-LRC has code parameters $(n,k,r,t)$ given by $(a_0(1+r+\frac{r^2}{2}),a_0 \frac{r^2}{2},r,t)$.  Clearly, one would be interested in having the length $n$ and hence $a_0$ as small as possible. We address this aspect of minimizing $a_0$ in Section \ref{sec:Moore_ch4}.

The matrices $D_0,D_1,A_1,C$ corresponding to our example code are all binary $\{0,1\}$ matrices and thus lead to a binary code.  This is the case with all of the graph-based code constructions that we provide in the paper.  Thus all of the seq-LRC codes constructed here are rate-optimal and binary. This is not however, necessary.   The matrices matrices $D_0,D_1,A_1,C$ could in general, be nonbinary and could potentially lead to a rate-optimal nonbinary seq-LRC having shorter block length.  

Next, we modify $H$ slightly as this will lead us to a convenient graphical interpretation of the code.  Let us form the matrix \hinf\ obtained by adding at the very top, a row whose entries are the sums of the entries in the remaining rows.   Thus \hinf\ takes on the form:  
\bea
	\hinf\ = \left[
	\begin{array}{c|c|c}
		\underline{1} & 0 & 0 \\
		\cline{1-3}
		D_0 & A_1 & 0  \\
		\cline{1-3}
		0 & D_1 & C 
	\end{array} 
	\right],  
	\eea
	\ben
	\item $\underline{1}$ is an $(1 \times a_0)$ vector with each coordinate equal to $1$.
	\item the matrices $D_0,D_1,A_1,C$ remain as before.
	\een	
Clearly, \hinf\ is also a valid p-c matrix for the code ${\cal C}$.  Each column in the p-c matrix \hinf\ above has Hamming weight $2$. This column-weight property of \hinf\ facilitates a graphical representation of the code.  The corresponding graph \ginf\ with node-edge incidence matrix \hinf\ is shown in Fig~\ref{fig:Moore}, corresponding to the value $a_0=7$ for a certain choice of the matrices $D_0,D_1,A_1,C$ such that girth of \ginf\ is $\geq t+1=5$.  As will be seen in Theorem \ref{LowerBoundOna0}, it turns out that the parameter $a_0$ in the case of the current example, cannot be any smaller. 

Each edge in \ginf\ represents a distinct code symbol while each vertex represents a parity check on the code symbols represented by edges incident on the vertex.  Thus each vertex is associated to a row in the p-c matrix \hinf\ and each edge to a column of the p-c matrix.  Each column of the p-c matrix \hinf\ has Hamming weight $2$ and the location of the two $1$s within the column indicates the vertices to which the edge is connected. In Fig~\ref{fig:Moore}, the edges at the very top, which are colored in blue, correspond to the first $a_0$ columns of \hinf. The edges which are colored in red and green, correspond respectively, to the columns of \hinf\ corresponding to the sub-matrices  
\bean
\left[ \begin{array}{c} 0 \\ A_1 \\ D_1 \end{array} \right] & \text{  and  } & \left[ \begin{array}{c} 0 \\ 0 \\ C \end{array} \right].
\eean
 In the example, we have $a_0=r+1=7$ and hence \ginf\ is a regular graph. In general, we can only assert that $a_0 \geq (r+1)$ (Theorem \ref{LowerBoundOna0}) and hence \ginf\ will not in general, be regular.  The sequential recovery property of this binary code derives from the girth of \ginf.  The girth of \ginf\ in our example, can be observed to be $5$. Hence if there are any $\leq 4$ erased symbols and if in \ginf\ only edges corresponding to erased symbols are retained, there will be at least one vertex or parity check with degree $1$ and hence the erased symbols can be recovered one by one.  A decoder that proceeds to decode in this fashion, is called a ``peeling decoder".
	\begin{figure}[ht]
		\centering
		\includegraphics[scale = 0.2]{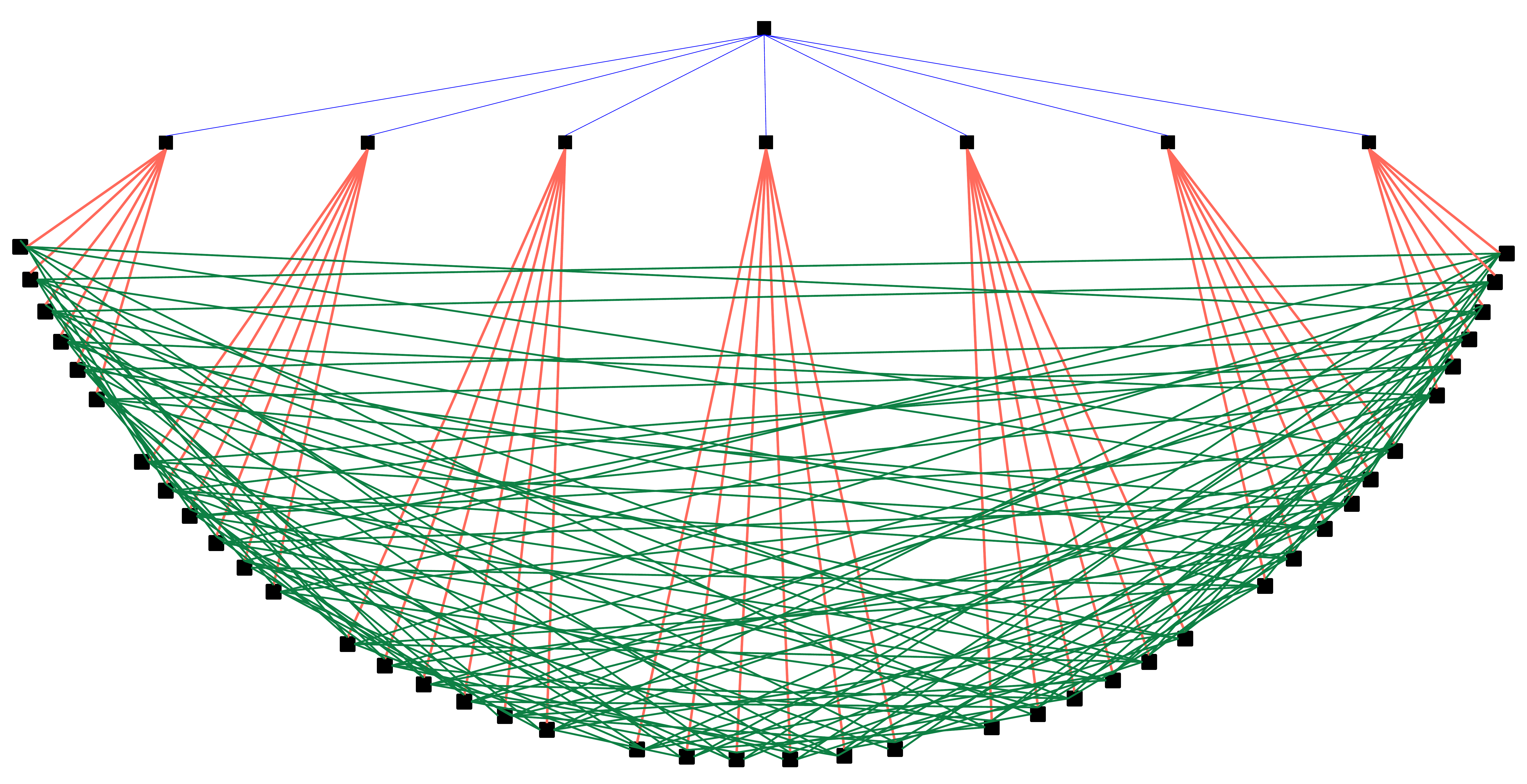}
		\caption{The figure shows a graphical interpretation of a binary, rate-optimal seq-LRC $\mathcal{C}$ having parameter set $(n,k,r,t)=(175,126,6,4)$.   Each of the $175$ edges of the graph represents a distinct code symbol and each of the $50$ vertices represents a parity check of the code symbols represented by edges incident on it. This is a regular graph with a total of $50$ vertices, each of degree $r+1=7$ and is an example of a Moore graph called the Hoffman-Singleton graph.   This graph has girth $5$, which is a necessity for the associated binary code to be able to recover from $t=4$ erasures.    The code has  redundancy $49$ and not $50$ since it turns out that the overall p-c at the very top is redundant. }
		\label{fig:Moore}
	\end{figure}
	
\begin{note}
We remark that even in the general case, when the constituent matrices  $D_0,D_1,A_1,C$ are not binary, a graphical interpretation of a rate-optimal $(n,k,r,4)_{\text{seq}}$ code is still possible, by introducing a fictitious p-c which plays the role of the vertex appearing at the very top of the graph in Fig.~\ref{fig:Moore}. 
\end{note}
	
	\subsection{Illustrative Example for $t=5,r=3$ ($t$ Odd Case)}
	
	Let \calc\ be a binary, rate-optimal seq-LRC with parameters $t=5,r=3$.
	For this choice of parameters, it will be shown in Corollary \ref{cor:equality_conditions} that the p-c matrix of the code takes on the form: 
	\bea
	H = \left[
	\begin{array}{c|c|c}
		D_0 & A_1 & 0 \\
		\cline{1-3}
		0 & D_1 & A_2 \\
		\cline{1-3}
		0 & 0 & D_2
		 
	\end{array} 
	\right],  
	\eea
	where
	\ben
	\item $D_0$ is an $(a_0 \times a_0)$ diagonal matrix,
	\item $A_1$ is an $(a_0 \times a_0r)$ matrix with each row of weight $r$ and each column of weight $1$,
	\item $D_1$ is an $(a_0r \times a_0r)$ diagonal matrix,
	\item $A_2$ is an $(a_0r \times a_0r^2)$ matrix with each row of weight $r$ and each column of weight $1$,
	\item $D_2$ is an $(\frac{a_0 r^2}{r+1} \times a_0r^2)$ matrix with each row of weight $r+1$ and each column of weight $1$,
	\een
	The integer parameter $a_0$ determines the block length $n=a_0(1+r+r^2)$ of the code.   It turns out that that the rank of the above p-c matrix is equal to the number of rows.  Thus this seq-LRC has code parameters $(n,k,r,t)$ given by $(a_0(1+r+r^2),a_0r^2-\frac{a_0r^2}{r+1},r,t)$.  Clearly, one would be interested in having the length $n$ and hence $a_0$ as small as possible.   This aspect of minimizing $a_0$ is discussed in Section \ref{sec:Moore_ch4}.
	
	As in the earlier case of $t$ even, $t=4$, the matrices $D_0,A_1,D_1,A_2,D_2$ corresponding to our example code are all binary $\{0,1\}$ matrices and thus lead to a binary code.  
	
	As in the case $t$ even, we modify $H$ slightly and form the matrix \hinf\ obtained by adding at the very top, a row whose entries are the sums of the entries in the remaining rows.   Thus \hinf\ takes on the form:  
	\bea
	\hinf\ = \left[
	\begin{array}{c|c|c}
		\underline{1} & 0 & 0 \\
		\cline{1-3}
		D_0 & A_1 & 0 \\
		\cline{1-3}
		0 & D_1 & A_2 \\
		\cline{1-3}
		0 & 0 & D_2
		
	\end{array} 
	\right],  
	\eea
	\ben
	\item $\underline{1}$ is an $(1 \times a_0)$ vector with each coordinate equal to $1$.
	\item the matrices $D_0,A_1,D_1,A_2,D_2$ remain as before.
	\een	
	Clearly, \hinf\ is also a valid p-c matrix for the code ${\cal C}$.  Each column in the parity-check matrix \hinf\ above has Hamming weight $2$. This column-weight property of \hinf\ facilitates a graphical representation of the code.  The corresponding graph \ginf\ with \hinf\ as node-edge incidence matrix is shown in Fig~\ref{fig:Moore2}, corresponding to the value $a_0=4$ for a certain choice of matrices $D_0,A_1,D_1,A_2,D_2$ such that girth of \ginf\ is $\geq t+1=6$.  As will be seen in Theorem \ref{LowerBoundOna0}, it turns out that the parameter $a_0$ in the case of the current example, cannot be any smaller.
	
	As in the example case above of $t$ even, $t=4$, each edge in \ginf\ represents a distinct code symbol while each vertex represents a parity check on the code symbols represented by edges incident on the vertex.  Thus each vertex is associated to a row in the p-c matrix \hinf\ and each edge to a column of the p-c matrix.  Each column of the p-c matrix \hinf\ has Hamming weight $2$ and the location of the two $1$s within the column indicates the vertices to which the edge is connected. In Fig~\ref{fig:Moore2}, the edges at the very top, which are colored in blue, correspond to the first $a_0$ columns of \hinf. The edges which are colored in red and green, correspond respectively, to the columns of \hinf\ corresponding to the sub-matrices  
	\bean
	\left[ \begin{array}{c} 0 \\ A_1 \\ D_1 \\ 0 \end{array} \right] & \text{  and  } & \left[ \begin{array}{c} 0 \\ 0 \\ A_2 \\ D_2 \end{array} \right].
	\eean
	In the example, we have $a_0=r+1=4$ and hence \ginf\ is a regular graph. In general, we can only assert that $a_0 \geq (r+1)$ (Theorem \ref{LowerBoundOna0}) and hence \ginf\ will not in general, be regular.  The sequential recovery property of this binary code derives from the girth of \ginf. The girth of \ginf\ in our example, can be observed to be $6$. Hence if there are any $\leq 5$ erased symbols and if in \ginf\ only edges corresponding to erased symbols are retained, there will be at least one vertex or parity check with degree $1$ and hence the erased symbols can be recovered one by one.  A decoder that proceeds to decode in this fashion, is called a ``peeling decoder".
	\begin{figure}[ht]
		\centering
		\includegraphics[scale = 0.5]{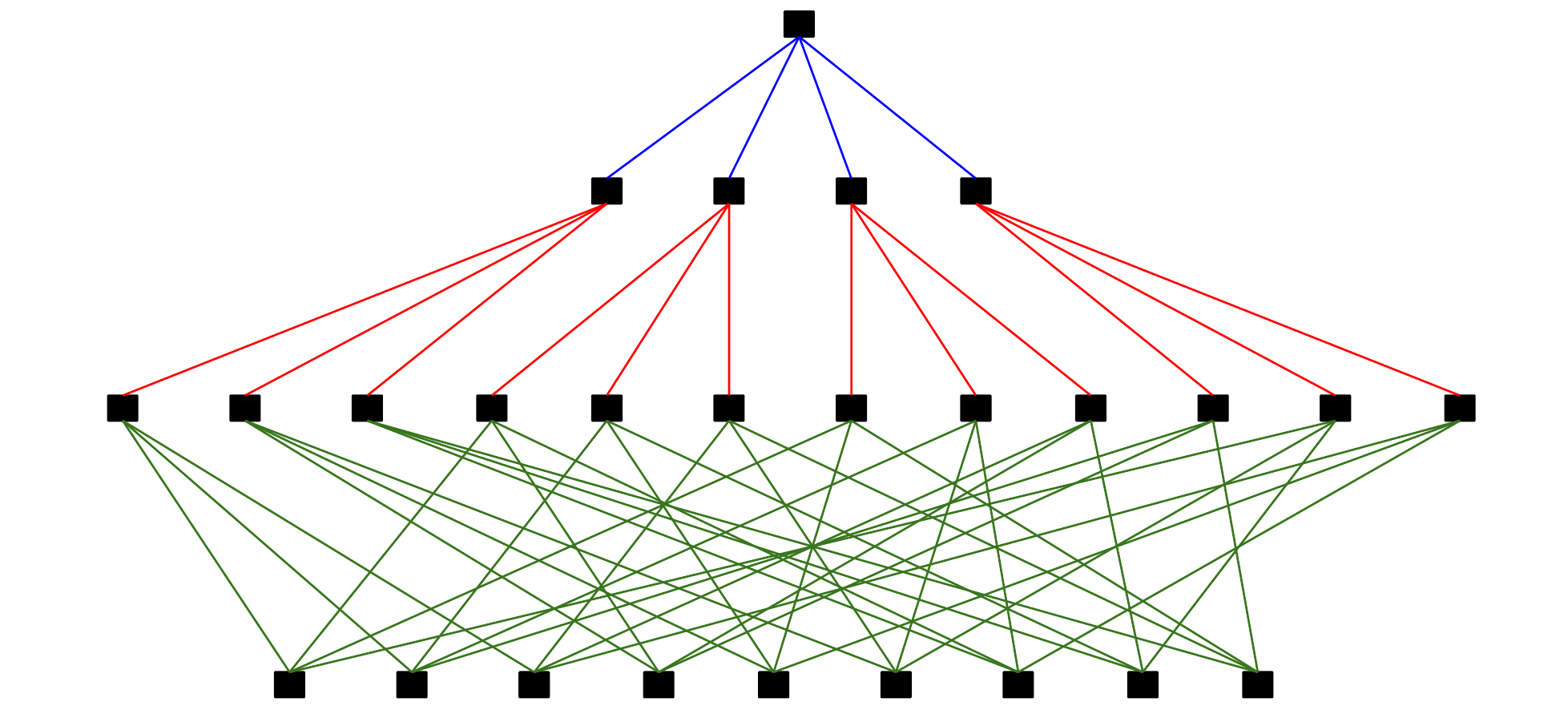}
	\caption{The figure shows a graphical interpretation of a binary, rate-optimal seq-LRC $\mathcal{C}$ having parameter set $(n,k,r,t)=(52,27,3,5)$.  Each of the $52$ edges of the graph represents a distinct code symbol and each of the $26$ vertices represents a parity check of the code symbols represented by edges incident on it.   This is a regular graph with a total of $26$ vertices, each of degree $r+1=4$ and is an example Moore graph for $t=5,r=3$ corresponding to projective plane of order $r=3$.  This graph has girth $6$, which is a necessity for the associated binary code to be able to recover from $t=5$ erasures. The code has  redundancy $25$ and not $26$ since it turns out that the overall p-c at the very top is redundant. }
		\label{fig:Moore2}
	\end{figure}
	
	\begin{note}
		As in the case of $t$ even, in the general case of $t$ odd, where the constituent matrices  $D_0,A_1,D_1,A_2,D_2$ are not binary, a graphical interpretation of a rate-optimal $(n,k,r,5)_{\text{seq}}$ code is still possible by introducing a fictitious p-c which plays the role of the vertex appearing at the very top of the graph in Fig.~\ref{fig:Moore2}. 
	\end{note}


	 \section{A Parity-Check-Matrix-Based Tight Upper Bound on the Rate of a seq-LRC} \label{sec:tight_bound}
	 
	 In this section, an upper bound on the rate of an $(n,k,r,t)_{\text{seq}}$ code for any $r \geq 3$ and any $t$ is derived. The cases of even $t$ and odd $t$ are considered separately.  The proof proceeds by deducing the structure of parity-check matrix of a seq-LRC. Constructions of binary codes achieving this upper bound for any $(r,t)$ with $r \geq 3$ are provided in Section \ref{sec:complete_code}.   These matching constructions establish that the upper bound on rate derived here is tight for all $(r,t)$ with $r \geq 3$.  The upper bound also proves Conjecture \ref{conj:Song} due to Song et al. 
	 	 \begin{thm} \label{thm:rate_both} \textbf{Rate Bound}: 
	 	Let $\mathcal{C}$ denote an $(n,k,r,t)_{\text{seq}}$ code over a finite field $\mathbb{F}_q$. Let $r \geq 3$. Then 
	 	\begin{align}
	 	\frac{k}{n} & \leq  \frac{r^{s+1}}{r^{s+1} + 2 \sum_{i=0}^{s} r^i}\hspace{0.5cm}&\text{for $t$ even,} \label{Thm1}\\
	 	\frac{k}{n} & \leq  \frac{r^{s+1}}{r^{s+1} + 2 \sum_{i=1}^{s} r^i + 1}\hspace{0.5cm}&\text{for $t$ odd,}\label{Thm2}
	 	\end{align}
	 	where $s = \lfloor \frac{t-1}{2} \rfloor$.
	 \end{thm}
	 \begin{proof}  We introduce some notations used in the proof and provide a sketch of the proof here, further details can be found in Appendix \ref{Appendix_rate_bounds}. An alternative method of proof, using the technique of linear programming, is also provided in Appendix \ref{app:linprog}.
Let $\mathcal{C}^{\perp}$ denote the dual code and $H$ be a p-c matrix of the code $\mathcal{C}$.   
We begin by setting $\mathcal{B}_0 = \text{span}(\{\underline{c} \in \mathcal{C}^{\perp} : w_H(\underline{c}) \leq r+1 \})$ where $w_H(\underline{c})$ denotes the Hamming weight of the vector $\underline{c}$.  Let $m$ be the dimension of $\mathcal{B}_0$. Let ${\underline{c}_1,\underline{c}_2,\ldots,\underline{c}_m}$ be a basis of $\mathcal{B}_0$ chosen such that $w_H(\underline{c}_i) \leq r+1$, $\forall i \in [m]$.
	 	Let $H_0={[{\underline{c}_1} \  {\underline{c}_2} \ldots {\underline{c}_m}]}^T$.
	 	It follows that $H_0$ is a p-c matrix of an $(n,n-m,r,t)_{\text{seq}}$ code as its row space contains every codeword of Hamming weight at most $r+1$ which is present in $\mathcal{C}^\perp$. Since the null space of $H_0$ contains the code $\mathcal{C}$, 
	 	\bean
	 	\frac{k}{n} \leq 1 - \frac{m}{n}.
	 	\eean
	 	As we are interested in characterizing rate-optimal seq-LRC, we can assume w.l.o.g that $H_0$ is the p-c matrix of the code $\mathcal{C}$, i.e., $H=H_0$ and $k=n-m$.  Thus the parameter $m$ has interpretation as the redundancy $m=(n-k)$ of the code $\mathcal{C}$.  The idea behind the next few arguments in the proof is the following.  Seq-LRCs with higher rate will have a larger value of $n$ for a fixed value of redundancy $m$.  On the other hand, the Hamming weight of the matrix $H$ (i.e., the number of non-zero entries in the matrix) is bounded above by $m(r+1)$. It follows that to make $n$ large, the columns of $H$ must be chosen to have as small a weight as possible. It is therefore quite logical to start building $H$ by picking many columns of weight $1$, then columns of weight $2$ and so on.  As one proceeds by following this approach, it turns out that the matrix $H$ is forced to have a certain sparse, block-diagonal, staircase form and an understanding of this structure is used to derive the upper bound on code rate.    The cases of $t$ odd and $t$ even are treated separately.  Further details can be found in Appendix \ref{Appendix_rate_bounds}  	 \end{proof}
	 
	 \begin{cor}{Conditions for equality in \eqref{Thm1},\eqref{Thm2}} \label{cor:equality_conditions} 
	 	
 	 	As shown in Appendix \ref{Appendix_rate_bounds} (at the end of the proof for $t$ even and at the end of the proof for $t$ odd), if \calc\ is an $(n,k,r,t)_{\text{seq}}$ code having rate achieving the upper bound given in Theorem~\ref{thm:rate_both}, then there exists a p-c matrix $H$ for \calc\ such that 
	 	\ben
	 	 \item each row of $H$ has Hamming weight equal to $(r+1)$ and 
	 	\item each column of $H$ has Hamming weight equal to either $1$ or $2$. 
	 	\een
	 	Furthermore, for $t$ even, $t=2s+2$ the p-c matrix can be put in the form:
	 	\bea
	 	H = \left[
	 	\begin{array}{c|c|c|c|c|c|c|c}
	 		D_0 & A_1 & 0 & 0 & \hdots & 0 & 0 & 0  \\
	 		\cline{1-8}
	 		0 & D_1 & A_2 & 0 & \hdots & 0 & 0 & 0 \\
	 		\cline{1-8}
	 		0 & 0 & D_2 & A_3 & \hdots & 0 & 0 & 0  \\
	 		\cline{1-8}
	 		0 & 0 & 0 & D_3 & \hdots & 0 & 0 & 0 \\
	 		\cline{1-8}
	 		\vdots & \vdots & \vdots & \vdots & \hdots & \vdots & \vdots & \vdots \\
	 		\cline{1-8}
	 		0 & 0 & 0 & 0 & \hdots & A_{s-1} & 0 & 0  \\
	 		\cline{1-8}
	 		0 & 0 & 0 & 0 & \hdots & D_{s-1} & A_{s} & 0 \\
	 		\cline{1-8}
	 		0 & 0 & 0 & 0 & \hdots & 0 & D_{s} & C \\
	 	\end{array} 
	 	\right],  \label{eq:Hmatrixteven_ch3}
	 	\eea
	 	where 
	 	\bit
	 	\item $D_0$ is an $(a_0 \times a_0)$ diagonal matrix for some integer $a_0$,
	 	\item $D_i$ is an $(a_0 r^i \times a_0 r^i)$ diagonal matrix, $\forall 1 \leq i \leq s$.
	 	\item $A_i$ is an $(a_0 r^{i-1} \times a_0 r^i)$ matrix with each column of weight $1$ and each row of weight $r$, $\forall 1 \leq i \leq s$,
	 	\item $C$ is an $(a_0 r^{s} \times a_0 \frac{r^{s+1}}{2})$ matrix with each column of weight $2$ and each row of weight $r$, 
	 	\item $n=\sum_{i=0}^{s}a_0 r^i + a_0 \frac{r^{s+1}}{2}$, $k = a_0 \frac{r^{s+1}}{2}$. 
	 	\eit
	 	For $t$ odd, $t=2s+1$ the p-c matrix can be put in the form:
	 	\bea
	 	H & = & \left[
	 	\begin{array}{c|c|c|c|c|c|c}
	 		D_0 & A_1 & 0 & 0 & \hdots & 0 & 0   \\
	 		\cline{1-7}
	 		0 & D_1 & A_2 & 0 & \hdots & 0 & 0  \\
	 		\cline{1-7}
	 		0 & 0 & D_2 & A_3 & \hdots & 0 & 0   \\
	 		\cline{1-7}
	 		0 & 0 & 0 & D_3 & \hdots & 0 & 0   \\
	 		\cline{1-7}
	 		\vdots & \vdots & \vdots & \vdots & \ddots & \vdots & \vdots  \\
	 		\cline{1-7}
	 		0 & 0 & 0 & 0 & \hdots & A_{s-1} & 0   \\
	 		\cline{1-7}
	 		0 & 0 & 0 & 0 & \hdots & D_{s-1} & A_{s}   \\
	 		\cline{1-7}
	 		0 & 0 & 0 & 0 & \hdots & 0 & D_{s} 
	 	\end{array} \right] \label{eq:Hmatrixtodd_ch3},
	 	\eea
	 	where 
	 	\bit
	 	\item $D_0$ is an $(a_0 \times a_0)$ diagonal matrix for an integer $a_0$, 
	 	\item $D_i$ is an $(a_0 r^i \times a_0 r^i)$ diagonal matrix, $\forall 1 \leq i \leq s-1$, 
	 	\item $A_i$ is an $(a_0 r^{i-1} \times a_0 r^i)$ matrix with each column of weight $1$ and each row of weight $r$, $\forall 1 \leq i \leq s$, 
	 	\item $D_{s}$ is an $(a_0 r^{s-1} \frac{r}{r+1} \times a_0 r^{s})$ matrix with each column of weight $1$ and each row of weight $r+1$,  
	   \item $n=\sum_{i=0}^{s}a_0 r^i $, $k =a_0r^s- a_0 r^{s-1} \frac{r}{r+1}$. 
	 	\eit 
	 \end{cor}
	 
These properties will be made use of in Section~\ref{sec:complete_code} where {\em binary} codes achieving the rate bound are constructed. 

 \begin{note} [Block length] \label{blk_length_rem}
 	Since the dimension $k$ of a code is an integer, and the numerator and denominator of the right hand side in \eqref{Thm2} are relatively prime, it follows that for $t$ odd, in a code achieving the upper bound on code rate in \eqref{Thm2}, one must have that $n$ is an integer multiple of $r^{s+1}+2\sum_{i = 1}^{s}r^i+1$. When $t$ is even, the corresponding requirement from \eqref{Thm1}, is that $2n$ be an integer multiple of $r^{s+1} + 2 \sum_{i=0}^{s} r^i$.
 \end{note}
 
\begin{note} [Proof of the Conjecture \ref{conj:Song}]
	It can be seen that the upper bound on rate given in Theorem \ref{thm:rate_both} is of the form given in Conjecture \ref{conj:Song}.  We prove the conjecture in full here i.e., we will prove in Section~\ref{sec:complete_code} that the upper bound in Theorem \ref{thm:rate_both} is also achievable by constructing binary codes that achieve the upper bound on code rate for any $r \geq 3$ and any $t$. 
The upper bound on rate given in Theorem \ref{thm:rate_both}, for $t=2,3$, coincides with the upper bound given in \cite{PraLalKum} and \cite{SongCaiYue_L} respectively.  For $t \geq 4$, the upper bound on rate given in Theorem \ref{thm:rate_both} is new.  
\end{note}
	
Throughout the remainder of the paper, an \nkrt\ (seq-LRC) code achieving the upper bound in either \eqref{Thm1} or \eqref{Thm2} will be referred to as a {\em rate-optimal} seq-LRC.
	 
	 From the proof of Theorem \ref{thm:rate_both} in Appendix \ref{Appendix_rate_bounds}, it is apparent that the upper bound on the rate of an $(n,k,r,t)_{\text{seq}}$ code given in Theorem \ref{thm:rate_both} can also be viewed as an upper bound on rate of an $[n,k]$ linear code having minimum distance $\geq (t+1)$, and a p-c matrix whose rows have Hamming weight $\leq (r+1)$.    We refer the reader to the papers \cite{BurKriLitMil,BenLit,IceSam,Fro} in which an upper bound is derived on the rate of an $[n,k]$ code having p-c matrix whose rows have Hamming weight $\leq (r+1)$ and the much larger minimum distance $d_{\min}=n \delta >> (t+1)$.  The tightness or otherwise of the bounds derived in \cite{BurKriLitMil,BenLit,IceSam,Fro} is currently unknown.
	 
We note from Remark \ref{blk_length_rem} that it is not possible to construct codes which achieve the upper bound on rate given in Theorem \ref{thm:rate_both} for all values of block length $n$. This motivates the introduction in the Corollary below, of the notion of dimension optimality.
\begin{cor} \label{cor:dimension_bound}
	Let $\mathcal{C}$ denote an $(n,k,r,t)_{\text{seq}}$ code over a finite field $\mathbb{F}_q$. Let $r \geq 3$. Then 
		\bea 
		k \ & \leq  \ \left \lfloor \frac{n r^{s+1}}{r^{s+1} + 2 \sum_{i=0}^{s} r^i} \right \rfloor \hspace{0.5cm}  & \text{for $t$ even,} \label{Thm3} \label{eq:dim_1} \\
		k \ & \leq  \ \left \lfloor \frac{n r^{s+1}}{r^{s+1} + 2 \sum_{i=1}^{s} r^i + 1} \right \rfloor \hspace{0.5cm} &  \text{for $t$ odd,} \label{Thm4} \label{eq:dim_2} 
	\eea
		where $s = \lfloor \frac{t-1}{2} \rfloor$.
\end{cor}
\begin{proof}
	Directly follows from Theorem \ref{thm:rate_both}.
\end{proof}
We will refer to codes achieving the bounds in either \eqref{eq:dim_1} or \eqref{eq:dim_2} as {\em dimension-optimal codes}.  A few constructions of dimension-optimal seq-LRC are provided in Section~\ref{sec:high_dimension}.

\subsection{Sufficient Condition for code to be a seq-LRC}

Corollary~\ref{cor:equality_conditions} shows that for a seq-LRC to be rate-optimal, the p-c matrix must necessarily have the form described in the corollary.  The form of the parity-check matrix is not sufficient however, to guarantee that the resultant code will be able to recover sequentially from $t$ erasures.  As was seen in the example codes of Section~\ref{sec:examples} and as will be seen in general in Section \ref{sec:graphical_rep}, the structure of the p-c matrix leads to a graphical description \ginf\ of the code \calc. The general form of the parity-check matrix given in Corollary~\ref{cor:equality_conditions} does not however, completely specify the graph \ginf\ as the matrix $C$ in case of $t$ even and the matrix $D_s$ in case of $t$ odd is not completely specified in Corollary~\ref{cor:equality_conditions}.  As a result, a scale factor $a_0$ that determines the total number of vertices in the graph remains unspecified, as are certain edge connections.    It turns out that if the parameter $a_0$ and the edge connections are chosen so as to ensure that the graph \ginf\ has girth $\geq (t+1)$, then the code is guaranteed to always be able to recover sequentially from $t$ erasures.   Moreover, for the values of $a_0$ chosen in the present paper to ensure girth $\geq (t+1)$, the associated rate-optimal code can be chosen to be a binary code.

 	
	\section{A Graphical Representation for the Rate-Optimal seq-LRC}  \label{sec:graphical_rep} 
	
In the last section, we saw that the p-c matrix of a rate-optimal seq-LRC can be assumed without loss of generality, to have the staircase form appearing in equations \eqref{eq:Hmatrixteven_ch3} and \eqref{eq:Hmatrixtodd_ch3}.   It will be shown in the present section, just as was done in the case of the examples presented in Section~\ref{sec:examples}, that this form of p-c matrix leads to a graphical representation of the code.  The construction of rate-optimal seq-LRC presented in Section \ref{sec:complete_code} is based on this graphical representation and yields rate-optimal, binary seq-LRCs. 
As was the case with the examples presented in Section~\ref{sec:examples}, the graphical representation is slightly different for the cases of $t$ odd and $t$ even.  We will begin with the $t$-even case.  

	\subsection{$t$ Even Case} 
	
	In the case $t$ even, we recall from equation  \eqref{eq:Hmatrixteven_ch3}, that the p-c matrix of a rate-optimal code can be put into the form:
	\bea
	H = \left[
	\begin{array}{c|c|c|c|c|c|c|c}
		D_0 & A_1 & 0 & 0 & \hdots & 0 & 0 & 0  \\
		\cline{1-8}
		0 & D_1 & A_2 & 0 & \hdots & 0 & 0 & 0 \\
		\cline{1-8}
		0 & 0 & D_2 & A_3 & \hdots & 0 & 0 & 0  \\
		\cline{1-8}
		0 & 0 & 0 & D_3 & \hdots & 0 & 0 & 0 \\
		\cline{1-8}
		\vdots & \vdots & \vdots & \vdots & \hdots & \vdots & \vdots & \vdots \\
		\cline{1-8}
		0 & 0 & 0 & 0 & \hdots & A_{s-1} & 0 & 0  \\
		\cline{1-8}
		0 & 0 & 0 & 0 & \hdots & D_{s-1} & A_{s} & 0 \\
		\cline{1-8}
		0 & 0 & 0 & 0 & \hdots & 0 & D_{s} & C \\
	\end{array} 
	\right],  \label{eq:H_Eq_ch4}
	\eea
	where $s= \lfloor \frac{t-1}{2} \rfloor$, or equivalently, $t=2s+2$. 
We note first that, each column in $H$, with the exception of the columns associated to diagonal sub-matrix $D_0$ has Hamming weight $2$. To make this uniform, we add an additional row to $H$ at the top, which has all $1$s in the columns associated to $D_0$ and zeros elsewhere.  This leads to the augmented p-c matrix \hinf, shown in \eqref{eq:stair_H_even_aug_ch4}.  The added row may be regarded in general, as a fictitious parity check, which we will regard as the parity check ``at infinity'' associated to node $V_{\infty}$. 

	\bea
	\hinf\ & = & 
	\begin{array}{c} \ \\  V_{\infty} \\ V_0 \\ V_1 \\ V_2 \\ V_3 \\ \vdots \\ V_{s-2} \\ V_{s-1} \\ V_{s} \end{array} 
	\left[
	\begin{array}{c|c|c|c|c|c|c|c}
		E_0 & E_1 & E_2 & E_3 & \cdots & E_{s-1} & E_{s} & E_{s+1} \\ \hline \hline 
		\underline{1}^t & \underline{0}^t & \underline{0}^t & \underline{0}^t  & \hdots & \underline{0}^t & \underline{0}^t & \underline{0}^t  \\
		\cline{1-8}
		D_0 & A_1 & 0 & 0 & \hdots & 0 & 0 & 0  \\
		\cline{1-8}
		0 & D_1 & A_2 & 0 & \hdots & 0 & 0 & 0 \\
		\cline{1-8}
		0 & 0 & D_2 & A_3 & \hdots & 0 & 0 & 0  \\
		\cline{1-8}
		0 & 0 & 0 & D_3 & \hdots & 0 & 0 & 0 \\
		\cline{1-8}
		\vdots & \vdots & \vdots & \vdots & \ddots & \vdots & \vdots & \vdots \\
		\cline{1-8}
		0 & 0 & 0 & 0 & \hdots & A_{s-1} & 0 & 0  \\
		\cline{1-8}
		0 & 0 & 0 & 0 & \hdots & D_{s-1} & A_{s} & 0 \\
		\cline{1-8}
		0 & 0 & 0 & 0 & \hdots & 0 & D_{s} & C \\
	\end{array} 
	\right].  \label{eq:stair_H_even_aug_ch4}
	\eea
\begin{note} [Additional parity check]
 It turns out that in all of the code constructions presented in the current paper, and which arise from the graphical representation described below, all the entries in $H$ are binary, i.e., belong to the set $\{0,1\}$.  In these codes, the  additional row at the top of \hinf\ is simply the modulo $2$ sum of the remaining rows of \hinf\ and hence is associated with an actual, rather than fictitious, parity check.   In general, however, it is not necessary for the entries of $H$ to be binary.
\end{note}

	Since each column of \hinf\ has weight $2$, the matrix has a natural interpretation as the vertex-edge incidence matrix of a graph \ginf\ where the incidence matrix of \ginf\ is obtained by replacing each non-zero entry of \hinf\ with a $`1'$. We will interchangeably refer to a vertex as a node.  Hence the vertices of \ginf\ are in one-one correspondence with the rows of the matrix \hinf\ and the edges of \ginf\ are in one-one correspondence with the columns.  An edge in \ginf\ corresponding to a column containing non-zero entries in rows $i,j$ connects the nodes corresponding to these two rows. 


The nodes of \ginf\ corresponding to the rows of \hinf\ containing the rows of $D_i$, $ 0 \leq i \leq s$ will be denoted by \vi\ and similarly, the edges corresponding to the columns of \hinf\ containing the columns of $D_j$, $ 0 \leq j \leq s$ will be denoted by $E_j$. The edges associated with the columns of \hinf\ containing the columns of $C$ will be denoted by $E_{s+1}$.  As noted earlier, we use $V_{\infty}$ to denote the node associated with the row at the very top of \hinf\ (see \eqref{eq:stair_H_even_aug_ch4}), i.e., associated with the added parity-check. Each node except the node $V_{\infty}$ has degree $(r+1)$.  The vertex $V_{\infty}$ has degree $a_0$.

	\subsubsection{Canonical Graphical Representation of a Rate-Optimal seq-LRC  ($t$ Even)}
	
		We will now deduce a simple representation of the graph \ginf\ which we will refer to as the canonical representation.   We will begin with a description of the representation in the general case, followed by an example.  
		
Since each row of \hinf, apart from the top row, has weight $(r+1)$, it follows that in the resultant graph, every node except the node $V_{\infty}$ has degree $(r+1)$.  Node \vinfty\ has degree $a_0$. Since $D_0$ is a diagonal matrix, the $a_0$ edges originating from \vinfty\ are terminated in the $a_0$ nodes making up \vzero. We will use $E_{0}$ to denote this collection of edges. There are $r$ other edges that emanate from each node in \vzero, each of these edges is terminated at a distinct node in \vone. We use \eone\ to denote this collection of edges.  Each of the other $r$ edges that emanate from each node in \vone, terminate in a distinct node in \vtwo.  We use \etwo\ to denote this collection of edges.  We continue in this fashion, until we reach the nodes in \vs\ via edge-set \es.  Here the pattern is discontinued and the $r$ other edges outgoing from each node in \vs\ are terminated among themselves.  We use \espone\ to denote this last collection of edges.  From this it can be inferred that the graph has a tree-like structure, except for the edges (corresponding to edge-set \espone) linking the leaf nodes $V_s$ at the very bottom. 

\begin{example} Fig.~\ref{fig:tree_graph_t_even} shows the graph for the case $(r+1)=4, t=6$, $s= \lfloor \frac{t-1}{2} \rfloor \ = \ 2$ with $a_0=(r+1)=4$.  It can be verified from Corollary~\ref{cor:equality_conditions} that $n=106,k=54$.   Since $a_0=(r+1)=4$, node \vinfty\ has the same degree as all the other nodes, thus making this a regular graph of degree $4$.  The $4$ edges (corresponding to edge-set $E_0$) originating from \vinfty\ are terminated in the $4$ nodes making up \vzero.  The $3$ other edges (corresponding to edge-set \eone) that emanate from each node in \vzero, are each terminated at distinct nodes in \vone.   Each of the other $3$ edges (corresponding to edge-set $E_2$) that emanate from each node in \vone, terminate in a distinct node in \vtwo.   The $3$ other edges (corresponding to edge-set $E_3$) outgoing from each node in \vtwo\ are terminated among themselves.    Thus as can be seen in Fig.~\ref{fig:tree_graph_t_even}, the graph has a tree-like structure, except for the edges (corresponding to edge-set $E_3$) linking the leaf nodes $V_2$ at the very bottom. 
 \end{example}		

We will use \gi\ to denote the restriction of \ginf\ to node-set $V_i \cup  \cdots V_{s-1} \cup V_{s}$ i.e., \gi\ denotes the subgraph of \ginf\ induced by the nodes $V_i \cup  \cdots V_{s-1} \cup V_{s}$ for $0 \leq i \leq s$.   

Thus the graphs are nested:
	\bean
	\gs \subseteq {\cal G}_{s-1} \subseteq \cdots \subseteq \gtwo\ \subseteq \gone\ \subseteq \gzero \subseteq \ginf.
	\eean
The graphs $\gtwo\ \subseteq \gone\ \subseteq \gzero\ \subseteq {\mathcal{G}}_{\infty}$ for the case $t=6$ are identified in Fig.~\ref{fig:tree_graph_t_even}. 

\begin{note}  We note that the incidence matrix \hzero\ of the graph \gzero\ is obtained by deleting the top row of \hinf\ as well as the columns corresponding to the matrix $D_0$ as shown below: 
	\bea
\hzero & = & 
\begin{array}{c} \ \\ V_0 \\ V_1 \\ V_2 \\ V_3 \\ \vdots \\ V_{s-2} \\ V_{s-1} \\ V_{s} \end{array} 
\left[
\begin{array}{c|c|c|c|c|c|c}
	E_1 & E_2 & E_3 & \cdots & E_{s-1} & E_{s} & E_{s+1} \\ \hline \hline 
	A_1 & 0 & 0 & \hdots & 0 & 0 & 0  \\
	\cline{1-7}
	D_1 & A_2 & 0 & \hdots & 0 & 0 & 0 \\
	\cline{1-7}
	0 & D_2 & A_3 & \hdots & 0 & 0 & 0  \\
	\cline{1-7}
	0 & 0 & D_3 & \hdots & 0 & 0 & 0 \\
	\cline{1-7}
	\vdots & \vdots & \vdots & \ddots & \vdots & \vdots & \vdots \\
	\cline{1-7}
	0 & 0 & 0 & \hdots & A_{s-1} & 0 & 0  \\
	\cline{1-7}
	0 & 0 & 0 & \hdots & D_{s-1} & A_{s} & 0 \\
	\cline{1-7}
	0 & 0 & 0 & \hdots & 0 & D_{s} & C \\
\end{array} 
\right].  \label{eq:stair_H_even_pun_ch4}
\eea
As we will see later in Section \ref{sec:complete_code}, the construction of rate-optimal codes begins with the graph \gzero. 
\end{note}	
 \begin{note} [Girth requirement, $t$ even] 
 The structure of \ginf\ is to a large extent determined once we specify $a_0$, since the graph \ginf\ with the edges belonging to edge-set \espone\ deleted, is a tree with $V_{\infty}$ as the root node where every node apart from the root node, has degree $(r+1)$.  The root node, $V_{\infty}$, itself has degree $a_0$.    The only other freedom lies in selecting the pairs of nodes in node-set \vs\ that are linked by the edges belonging to edge-set \espone. The p-c matrix requires that these edges be selected such that each node in \vs\ is of degree $(r+1)$.   A key additional requirement that we will impose is that the girth of \ginf\ be $\geq t+1$.  While we do not claim this condition to be necessary, it does lead directly to the construction of a code that is guaranteed to recover form $t$ erasures sequentially (Theorem \ref{thm:girthreq}).  Moreover, this code can be chosen to be a binary code. As is shown in Theorem \ref{LowerBoundOna0}, for the graph \ginf\ to have girth $\geq t+1$, we must have that $a_0 \geq r+1$.
	\end{note}
	
	\begin{figure}[ht]
		\centering
		\includegraphics[scale = 0.6]{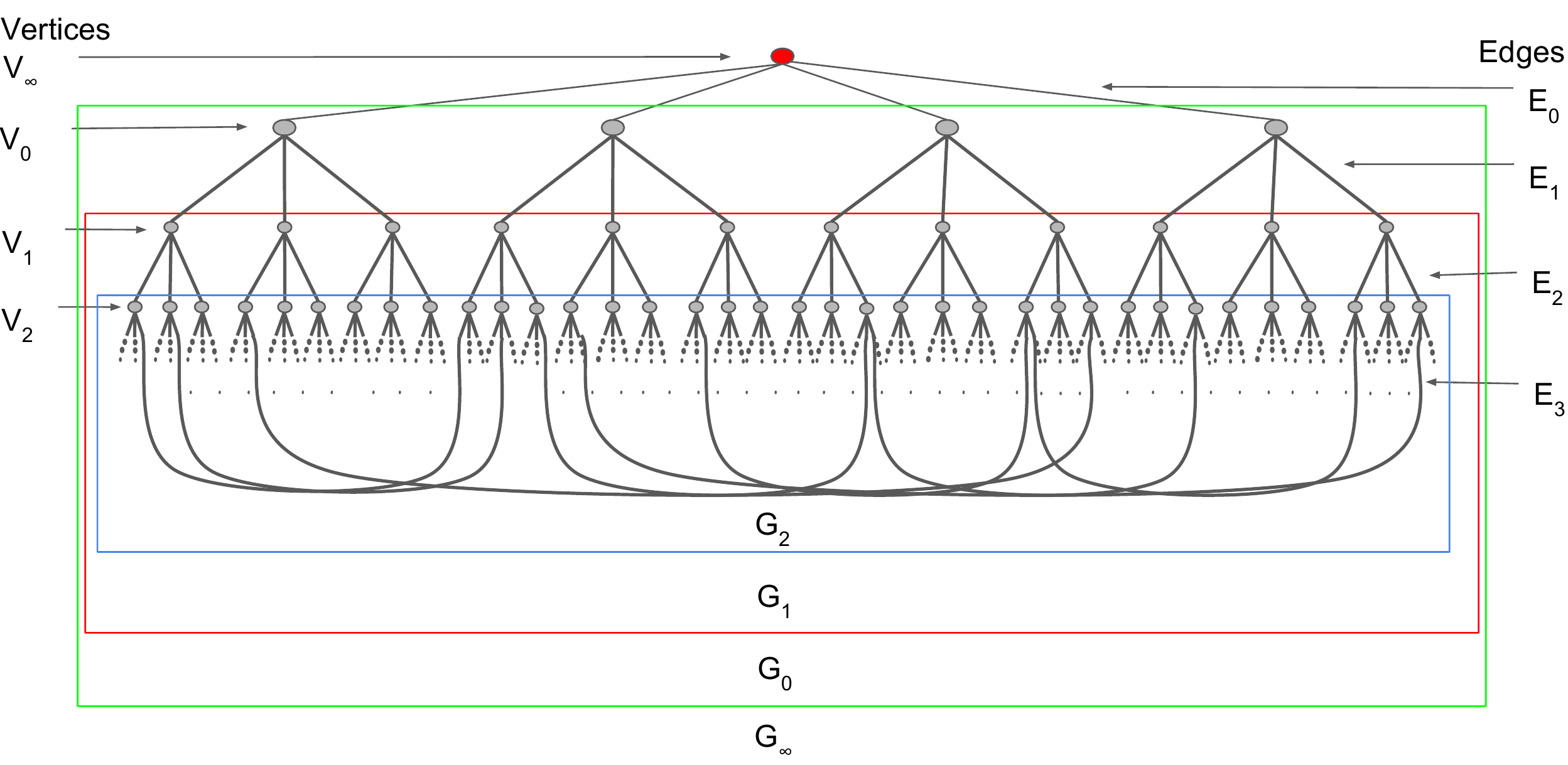}
		\caption{Case of $t$ even: Graphical representation of the code induced by the staircase nature of the p-c matrix appearing in \eqref{eq:stair_H_even_aug_ch4} shown for the case $(r+1)=4, \ t=6, \ s=2$ with $a_0=4$. Here $n=106$ and $k=54$. Apart from the edges connecting the nodes in the bottom layer, the graph is an $3$-ary tree with the exception of root node $V_{\infty}$ which has degree $4$.  In this particular example, we have chosen $a_0$ to have the minimum possible value $a_0=(r+1)$ (Theorem \ref{LowerBoundOna0}) which makes this graph a regular graph of degree $(r+1)=4$. In the general case, all the nodes would have degree $(r+1)$ with the possible exception of node \vinfty\ which would have degree $a_0$. }
		\label{fig:tree_graph_t_even}
	\end{figure}

	\subsection{$t$ Odd Case} \label{tOdd_Graph_Desc}
	
	In the case $t$ odd, the p-c matrix of a rate-optimal code can be put into the form (equation \eqref{eq:Hmatrixtodd_ch3}): 
	\bea
	H & = & \left[
	\begin{array}{c|c|c|c|c|c|c}
		D_0 & A_1 & 0 & 0 & \hdots & 0 & 0   \\
		\cline{1-7}
		0 & D_1 & A_2 & 0 & \hdots & 0 & 0  \\
		\cline{1-7}
		0 & 0 & D_2 & A_3 & \hdots & 0 & 0   \\
		\cline{1-7}
		0 & 0 & 0 & D_3 & \hdots & 0 & 0   \\
		\cline{1-7}
		\vdots & \vdots & \vdots & \vdots & \ddots & \vdots & \vdots  \\
		\cline{1-7}
		0 & 0 & 0 & 0 & \hdots & A_{s-1} & 0   \\
		\cline{1-7}
		0 & 0 & 0 & 0 & \hdots & D_{s-1} & A_{s}   \\
		\cline{1-7}
		0 & 0 & 0 & 0 & \hdots & 0 & D_{s} 
	\end{array} \right] \label{eq:stair_H_odd_ch4},
	\eea
	where $s= \lfloor \frac{t-1}{2} \rfloor$, or equivalently, $t=2s+1$. 
	Our goal once again, is to show that $H$ forces the code to have a certain graphical representation.  We add here as well, an additional row to $H$ at the top, which has all $1$s in the columns associated to $D_0$ and zeros elsewhere to obtain the matrix \hinf, shown in \eqref{eq:stair_H_odd_aug_ch4}.
	\bea
	\hinf = 
	\begin{array}{c} \ \\  V_{\infty} \\ V_0 \\ V_1 \\ V_2 \\ V_3 \\ \vdots \\ V_{s-2} \\ V_{s-1} \\ V_{s} \end{array} 
	\left[
	\begin{array}{c|c|c|c|c|c|c}
		E_0 & E_1 & E_2 & E_3 & \cdots & E_{s-1} & E_{s}  \\ \hline \hline 
		\underline{1}^t & \underline{0}^t & \underline{0}^t & \underline{0}^t  & \hdots & \underline{0}^t & \underline{0}^t  \\
		\cline{1-7}
		D_0 & A_1 & 0 & 0 & \hdots & 0 & 0  \\
		\cline{1-7}
		0 & D_1 & A_2 & 0 & \hdots & 0 & 0  \\
		\cline{1-7}
		0 & 0 & D_2 & A_3 & \hdots & 0 & 0   \\
		\cline{1-7}
		0 & 0 & 0 & D_3 & \hdots & 0 & 0  \\
		\cline{1-7}
		\vdots & \vdots & \vdots & \vdots & \ddots & \vdots & \vdots  \\
		\cline{1-7}
		0 & 0 & 0 & 0 & \hdots & A_{s-1} & 0   \\
		\cline{1-7}
		0 & 0 & 0 & 0 & \hdots & D_{s-1} & A_{s}  \\
		\cline{1-7}
		0 & 0 & 0 & 0 & \hdots & 0 & D_{s}  \\
	\end{array} 
	\right].  \label{eq:stair_H_odd_aug_ch4}
	\eea
	Since each column of \hinf\ also has weight $2$, the matrix again has an interpretation as the node-edge incidence matrix of a graph \ginf\ where the incidence matrix of the graph \ginf\ is obtained by replacing every non-zero entry of \hinf\ with $1$. We retain the earlier notation with regard to node sets \vi\ and node \vinfty\ and edge sets \ej (see \ref{eq:stair_H_odd_aug_ch4}).   We note that even here, apart from $V_{\infty}$, each node has degree $(r+1)$, with the root node $V_{\infty}$ itself having degree $a_0$.
	
	\subsubsection{Canonical Graphical Representation of a Rate-Optimal seq-LRC ($t$ Odd)}
	
	Next, we move on to a canonical representation of the graph exactly as in the case of $t$ even. Differences compared to the case $t$ even, appear only when we reach the nodes \vsmone\ via edge-set \esmone. Here, the $r$ other edges outgoing from each node in \vsmone\ are terminated in the node set \vs.  We use \es\ to denote this last collection of edges.  As can be seen, the graph has a tree-like structure, except for the edges (corresponding to edge-set \es) linking nodes in \vsmone\ and \vs.\   The restriction of the overall graph to $\vsmone\ \cup \ \vs$ i.e., the subgraph induced by the nodes $\vsmone\ \cup \vs$ can this time be seen to result in a {\em biregular, bipartite graph} \gsmone\ where each node in \vsmone\ has degree $r$ while each node in \vs\ has degree $r+1$.    A more detailed explanation for the appearance of such a bipartite graph is provided in Fig.~\ref{fig:bipartite}.
	
	\begin{figure}[ht!]
		\centering
		\begin{minipage}[c]{0.48\textwidth}
			\centering
		   \bean
		  \underbrace{ \left[ \begin{array}{ccc}
		   	D_{s-1} & A_s & 0 \\
		   	0 & D_s & C
		   \end{array} \right]}_{t \ \text{even}} 
		   \eean
		\end{minipage}
		\hspace{-0.1\textwidth}
		\begin{minipage}[c]{0.48\textwidth}
					\bean
					\underbrace{\left[ \begin{array}{cc}
						D_{s-1} & A_s\\
						0 & D_s
					\end{array} \right]}_{t \ \text{odd}}
					\eean
		\end{minipage}
				\caption{Explaining the appearance of a bi-regular, bipartite graph in the case of $t$ odd: Depicted on the left and right are the sub-matrices of the p-c matrix restricted to the last few rows and columns for the $t$ even and $t$ odd cases respectively. On the left, we see that since $C$ is a matrix where each column has weight $2$, the nodes corresponding to rows of $C$ have edges that originate and terminate among themselves, thereby forming an $r$-regular graph.  This is exactly the subgraph induced by nodes $V_s$.  On the right, we see that the nodes corresponding to the rows of $A_s$ are connected to nodes corresponding to the rows associated with the matrix $D_s$, as each column of either $A_s$ or $D_s$ has Hamming weight one.   It follows that the nodes corresponding to rows of $[\frac{A_s}{D_s}]$ form a bipartite graph. This graph is bi-regular as each row of $A_s$ has weight $r$ and each row of $D_s$ has weight $r+1$. This bipartite graph is precisely the subgraph induced by nodes $V_{s-1} \cup V_s$.}
				\label{fig:bipartite}
	\end{figure}
	
\begin{example} Fig.~\ref{fig:tree_graph_t_odd} shows the graph for the case $(r+1)=4,t=7,s=3$ with $a_0=(r+1)=4$.  It can be verified from Corollary~\ref{cor:equality_conditions} that $n=160,k=81$. Here $a_0=(r+1)=4$ and hence node \vinfty\ has the same degree as all other nodes making this a regular graph of degree $4$.  As described above, the subgraph induced by $V_2 \cup V_3$ is a biregular, bipartite graph with each node in $V_2$ having degree $3$ and each node in $V_3$ having degree $4$ in the induced bipartite graph.
\end{example}
	
	We use \ginf\ to denote the overall graph and use \gi\ to denote the restriction of \ginf\ to node-set $ V_i \cdots \cup V_{s-1} \cup V_{s}   $ i.e., \gi\ denotes the subgraph of \ginf\ induced by the nodes $V_i \cup  \cdots V_{s-1} \cup V_{s}$ for $0 \leq i \leq s-1$. Thus the graphs are nested:
	\bean
	{\cal G}_{s-1} \subseteq {\cal G}_{s-2} \subseteq \cdots \subseteq \gtwo\ \subseteq \gone\ \subseteq \gzero\ \subseteq \ginf.
	\eean
	Fig.~\ref{fig:tree_graph_t_odd} identifies the graphs $\gtwo\ \subseteq \gone\ \subseteq \gzero\ \subseteq {\mathcal{G}}_{\infty}$ for the case $t=7$. 
	
	\begin{note} [Condition on $a_0$]\label{note:bipartite_count} 
From the tree-like structure of the graph \ginf\ it follows that the number of nodes in $V_{s-1}$ and $V_s$ are respectively given by 
		\bea
		|V_{s-1}| & = & a_0 r^{s-1} \label{eq:Vsminus1_ch4} \\
		|V_s| & = & a_0 r^{s-1}\frac{r}{r+1}. \label{eq:Vs_ch4}
		\eea
		Since $r,r+1$ are co-prime, this forces $a_0$ to be a multiple of $(r+1)$.    This leads to the theorem below. 
	\end{note}
	\begin{thm}\label{thm:t_odd_azero}
		For $t$ odd, $a_0$ must be a multiple of $(r+1)$. 
	\end{thm} 
The incidence matrix of the graph \gzero\ i.e., the graph \ginf\ with the node $V_{\infty}$ removed is given by :

	\bea
	\hzero = 
	\begin{array}{c} \ \\ V_0 \\ V_1 \\ V_2 \\ V_3 \\ \vdots \\ V_{s-2} \\ V_{s-1} \\ V_{s} \end{array} 
	\left[
	\begin{array}{c|c|c|c|c|c}
		E_1 & E_2 & E_3 & \cdots & E_{s-1} & E_{s}  \\ \hline \hline 
		A_1 & 0 & 0 & \hdots & 0 & 0  \\
		\cline{1-6}
		D_1 & A_2 & 0 & \hdots & 0 & 0  \\
		\cline{1-6}
		0 & D_2 & A_3 & \hdots & 0 & 0   \\
		\cline{1-6}
		0 & 0 & D_3 & \hdots & 0 & 0  \\
		\cline{1-6}
		\vdots & \vdots & \vdots & \ddots & \vdots & \vdots  \\
		\cline{1-6}
		0 & 0 & 0 & \hdots & A_{s-1} & 0   \\
		\cline{1-6}
		0 & 0 & 0 & \hdots & D_{s-1} & A_{s}  \\
		\cline{1-6}
		0 & 0 & 0 & \hdots & 0 & D_{s}  \\
	\end{array} 
	\right].  \label{eq:stair_H_odd_pun_ch4}
	\eea
As we will see later in Section \ref{sec:complete_code}, the construction of rate-optimal codes begins with the graph \gzero. 
 \begin{note} [Girth requirement, $t$ odd] 
 As in the case $t$ even, the structure of \ginf\ is largely determined once we specify $a_0$ as \ginf\ with edges in \es\ removed is just a tree with $V_{\infty}$ as the root node. Hence the only freedom lies in selecting the edges that make up the bipartite graph \gsmone.  We once again impose the key additional requirement that the girth of \ginf\ be $\geq t+1$.  While we do not claim this condition to be necessary, it does lead directly to the construction of a code that is guaranteed to recover form $t$ erasures sequentially (Theorem \ref{thm:girthreq}).  Moreover, this code can be chosen to be a binary code.  
\end{note}

	\begin{figure}[ht]
		\centering
		\includegraphics[scale = 0.55]{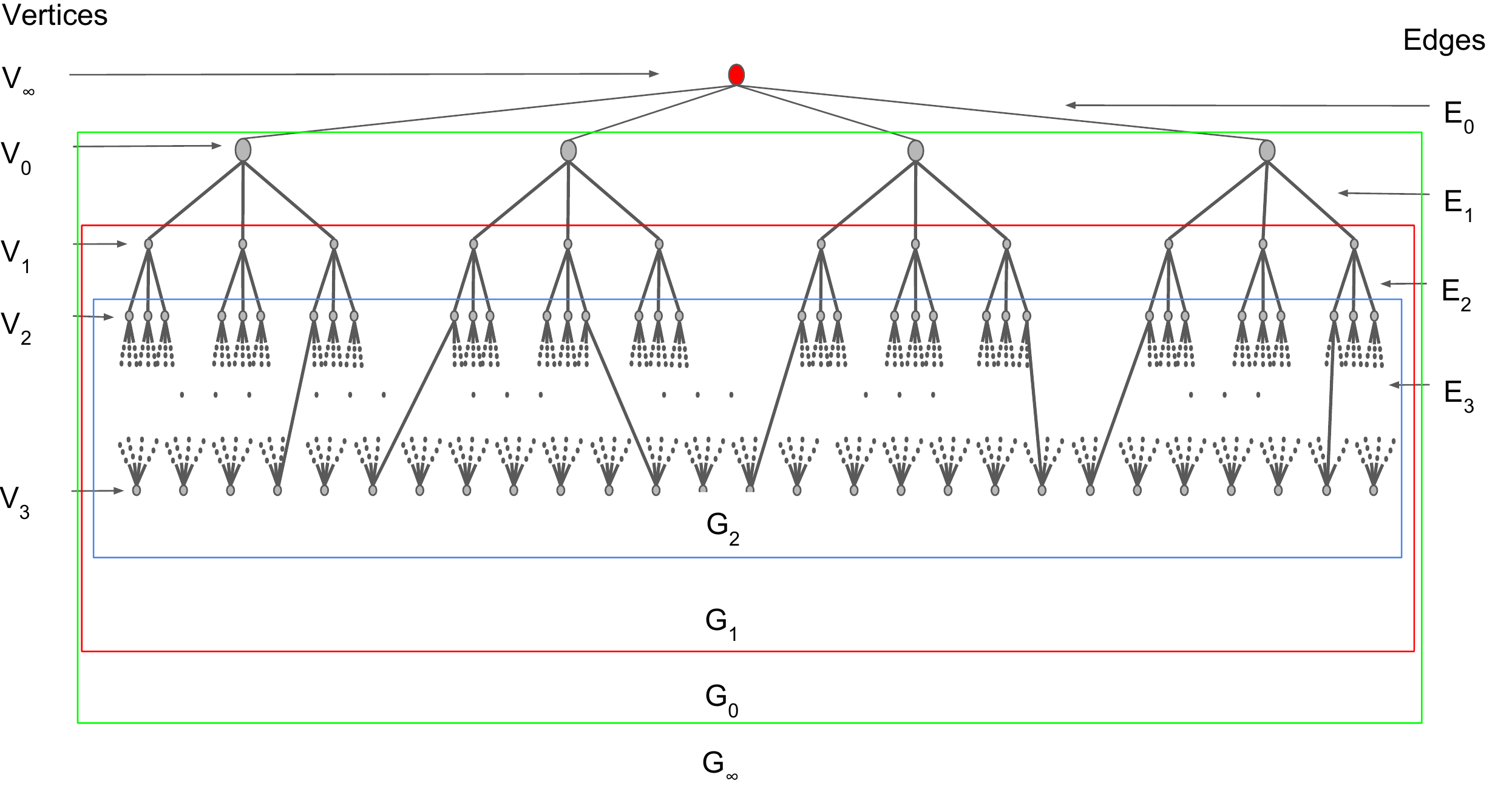}
		\caption{Case of $t$ odd: graphical representation of the code induced by the staircase nature of the p-c matrix appearing in \eqref{eq:stair_H_odd_aug_ch4} shown for the case $(r+1)=4, \ t=7, \ s=3$ with $a_0=4$. Here $n=160$ and $k=81$. Apart from the edges between $V_2$ and $V_3$, the graph is an $3$-ary tree with the exception of root node $V_{\infty}$ which has degree $4$.  The subgraph induced by $V_2 \cup V_3$ is a bi-regular, bipartite graph with each node in $V_2$ having degree $3$ and each node in $V_3$ having degree $4$ in the induced bipartite graph.  In this particular example, we have chosen $a_0$ to have the minimum possible value $a_0=(r+1)$ (Theorem \ref{LowerBoundOna0}) which makes this graph a regular graph of degree $(r+1)=4$. In the general case, all the nodes would have degree $(r+1)$ with the possible exception of node \vinfty\ which would have degree $a_0$. }
		\label{fig:tree_graph_t_odd}
	\end{figure}

	
	\section{Girth Requirement on the Underlying Graph} \label{sec:girth_req} 
	We begin by noting that the p-c matrix deduced in Corollary~\ref{cor:equality_conditions} was only specified to the extent of distinguishing between entries that are zero from those that are nonzero.  The graph \ginf\ replaces each nonzero entry of the p-c matrix with a $1$ and adds a row and then views the resultant matrix as the vertex-edge incidence matrix for the graph. We now show that irrespective of the finite field, if the graph \ginf\ has girth $\geq t+1$, then the code is guaranteed to recover from $t$ erasures sequentially. It is shown that this girth condition on \ginf\ is necessary in the case of binary codes but it may not be necessary for codes over other finite fields.    This is because it may be possible to have a p-c matrix as in \eqref{eq:Hmatrixteven_ch3}, \eqref{eq:Hmatrixtodd_ch3} whose entries belong to a nonbinary field and which leads to a graphical representation \ginf\ with girth $< t+1$ while still being able to recover from $t$ erasures sequentially.   Although not pursued in the present paper, this holds out the possibility that nonbinary seq-LRC can have shorter block length when compared against their binary counterparts.  Nonethless the block length of a rate-optimal seq LRC irrespective of the finite field has to satisfy the condition given in Remark \ref{blk_length_rem}.  We summarize our observations concerning the girth requirement in the theorem below.  
	\begin{thm} \label{thm:girthreq}
		For $t$ either even or odd, 
		\ben	
		\item the binary code associated to graph \ginf\ with each node representing a p-c over  $\mathbb{F}_2$ can recover sequentially from $t$ erasures iff \ginf\ has girth $\geq (t+1)$, 
		\item the code associated to graph \ginf\ with each node representing a p-c over $\mathbb{F}_q$ in the nonbinary $q>2$ case, can recover sequentially from $t$ erasures if \ginf\ has girth $\geq (t+1)$.  
		\een
	\end{thm}
	\begin{proof}
	   Let us assume that there is an erasure pattern involving $\ell \leq t$ erased code symbols and that it is not possible to recover from this erasure pattern sequentially and $\ell$ is the smallest number with this property.  These $\ell$ erasures correspond to $\ell$ distinct edges $\{e_i\}_{i=1}^{\ell}$ of the graph \ginf.  Let $J=\{e_i \mid 1 \leq i \leq \ell\}$ and let us restrict our attention to the subgraph \gsub\ of \ginf\ having edge-set $J$ and vertex set equal to the set ${\cal U}$ of nodes that the edges in $J$ are incident upon. We note that every node in ${\cal U} \setminus \{\vinfty\}$ in the graph \gsub\ must have degree $\geq 2$.  This is because, the presence in \gsub\ of a node in ${\cal U} \setminus \{\vinfty\}$ of degree one would imply that the corresponding row of the p-c matrix can be used to recover the erased code symbol corresponding to the edge incident on it in \gsub.  We now consider two cases separately.  
		\ben
		\item Suppose $\vinfty \not \in {\cal U}$. We start with edge $e_1$, this must be linked to a p-c node $U_1 \in {\cal U}$ which is linked to a second erased symbol (say) $e_2$ and so on, as degree of each node in \gsub\ is $\geq 2$.  In this way, we can create a path in \gsub\ with distinct edges.  But since there are only a finite number of nodes, this must eventually force us to revisit a previous node, thereby establishing that the graph \ginf\ has girth $\leq \ell \leq t$. 
		\item Next  suppose $\vinfty \in {\cal U}$.  In this case, we start at an edge incident upon node \vinfty\ corresponding to an erased symbol and move to the node at the other end of the edge.  Since that node has degree $\geq 2$, there must be an edge corresponding to a second erased symbol that is connected to the node and so on.  Again the finiteness of the graph will force us to revisit either \vinfty\ or else, a previously-visited node proving once again that an unrecoverable erasure pattern contains a cycle and hence the graph \ginf\ has girth $\leq \ell \leq t$. 
		\een
		We have thus established that having a girth $\geq (t+1)$ will guarantee sequential recovery from $\leq t$ erasures.  For $q=2$, it is easy to see that a girth of $\geq (t+1)$ is necessary since if the girth is $\leq t$, then the set of erasures with erased code symbols forming a cycle of length $\leq t$ is uncorrectable regardless of whether or not the nodes associated with this cycle includes \vinfty.   This is because the columns of the p-c matrix $H$ corresponding to the edges forming the cycle, sum to the all zero vector, and are hence linearly dependent.
	\end{proof}
	
	\begin{thm} \label{LowerBoundOna0}
		For the graph \ginf\ to have girth $\geq (t+1)$, the degree $a_0$ of \vinfty\ or equivalently, the number $a_0$ of nodes in $V_0$, is lower bounded as $a_0 \geq r+1$.
	\end{thm}
	\begin{proof}
		\underline{Case of $t$ odd:} \ As shown in Theorem~\ref{thm:t_odd_azero}, $a_0$ must be in fact be a multiple of $(r+1)$.  
		
		\underline{Case of $t$ even:}  Let $v \in V_0$. Let $N_v$ denote the set containing all nodes in $V_s$ which are at a distance at most $s$ from the node $v \in V_0$. Here, distance between vertices $w_1,w_2$ is measured as the number of edges in the shortest path between $w_1,w_2$. Note that $N_v \subseteq V_s$ and $N_v \cap N_w = \emptyset$, $v \neq w$, as the two sets of vertices are disjoint.
		\begin{itemize}
			\item {\em In \ginf, no two nodes in $N_v$ can be connected by an edge in \espone\ :} every node in $N_v$ has a path of length (measured by the number of edges) $s$ leading to $v$, not involving an  edge from \espone. Now, if two nodes in $N_v$ were to be connected by an edge in \espone\, then there would be a cycle of length at most $2s+1 < t+1$, which violates the condition for $t$ erasure-correction.
			\item {\em A node in $N_v$ cannot connect to two nodes in $N_w$ via edges in \espone, $\forall v \neq w, w \in V_0$:} It would result in a cycle of length at most $2s+2 < t+1$.
		\end{itemize}
	  Hence, in conclusion, each node in $N_v$ must connect via edges in \espone\ to $r$ nodes, with $i$th node belonging to $N_{w_i}$, $1 \leq i \leq r$ respectively, for some set of $r$ distinct nodes $\{w_1,...,w_r\} \subseteq V_0-\{v\}$. It follows that there must be at least $r+1$ distinct nodes in $V_0$. In other words, $a_0 \geq r+1$. 
	\end{proof}
	  \begin{note} [Lower bound on block length]
	  	The inequality $a_0 \geq r+1$, imposes a lower bound on the block-length of the binary rate-optimal codes. This will be elaborated upon in the next section.
	  \end{note}


		\section{Seq-LRC that are Optimal with Respect to Both Rate and Block Length} \label{sec:Moore_ch4}
		
		In this section, we begin by presenting a construction for seq-LRC given in \cite{RawMazVis}. We improve upon the lower bound  on the code rate of this construction provided in \cite{RawMazVis} and also present a slight generalization that replaces the regular bipartite graph appearing in \cite{RawMazVis} with just a regular graph. This improved lower bound on code rate also applies to this slight generalization of the construction appearing in \cite{RawMazVis}. It is shown here that for a given $(r,t)$, there is a unique block length for which the improved lower bound on code rate is equal to the right hand side of the upper bound on rate derived in Theorem \ref{thm:rate_both}. We will see in this section that this unique block length corresponds to having $a_0=r+1$ and hence the resultant codes are not only rate optimal, they also have least possible block length (Theorem \ref{LowerBoundOna0}) for a binary rate-optimal seq-LRC.   For a given $(r,t)$, codes based on this construction with this unique block length correspond to codes based on a type of regular graphs known as a Moore graph.  In the present context, Moore graphs can be viewed as regular graphs, where every node has degree $r+1$ and which have the smallest number of vertices possible under the added requirement that the graphs have girth $\geq t+1$.   Unfortunately, Moore graphs exist (Theorem \ref{thm:MooreExis}) only for a very sparse set of $(r,t)$ parameters.   We begin with the construction given in \cite{RawMazVis}.
		\begin{const} \label{con:Raw_bipartite} (\cite{RawMazVis})
			Consider an $(r+1)$-regular bipartite graph $G$ having girth $\geq t+1$. Let the number of nodes in $G$ be $N$. Let $H$ be the $(N \times \frac{N(r+1)}{2})$ node-edge incidence matrix of the graph $G$. The binary code $\mathcal{C}$ with p-c matrix $H$ thus defined is a seq-LRC with parameters $(n = \frac{N(r+1)}{2},k \geq n-N,r,t)$ over $\mathbb{F}_2$. The construction takes $G$ as input and constructs code $\mathcal{C}$ as output.
		\end{const}
		\begin{proof}
			(sketch of proof)  
			Let $\mathcal{C}$ be the code obtained as the output of Construction \ref{con:Raw_bipartite} having as input, an $(r+1)$-regular bipartite graph $G$ with girth $\geq t+1$. The code $\mathcal{C}$ is a seq-LRC with parameters $(r,t)$, simply because a set of erased symbols with least cardinality which cannot be recovered through sequential recovery, must correspond to a set of linearly dependent columns in $H$ with least cardinality and hence corresponds to a set of edges forming a cycle in $G$.  Since $G$ has girth $\geq t+1$, the number of edges in this cycle must be $>t$ and hence the number of erased symbols is $>t$.  The code parameters follow from a simple calculation. 
		\end{proof}
		The graph $G$ described in Construction \ref{con:Raw_bipartite} need not have the tree-like structure of the graph \ginf.   Let $\mathcal{C}$ be the code obtained as the output of Construction \ref{con:Raw_bipartite} having as input, an $(r+1)$-regular bipartite graph $G$ having girth $\geq t+1$.  Since the graph $G$ need not have the tree-like structure of \ginf, it may not in general, be possible for the code $\mathcal{C}$ to have a p-c matrix similar to \eqref{eq:stair_H_even_aug_ch4} for $t$ even and \eqref{eq:stair_H_odd_aug_ch4} for $t$ odd and hence will not be rate optimal in general. We will now see that the code $\mathcal{C}$ is rate-optimal  if and only if $G$ is a Moore graph. It follows from Construction \ref{con:Raw_bipartite}, as was observed in \cite{RawMazVis}, that the rate of the code $\mathcal{C}$ is $\geq \frac{r-1}{r+1}$.   We will shortly provide a precise value for the rate of this code. 
	
		\begin{defn} (\textbf{Connected Component})
			Let $G$ be a graph. Then a connected component of $G$ is an induced subgraph $G_1$ such that $G_1$ is connected as a graph and moreover, there is no edge in $G$, connecting a vertex in $V(G_1)$ to a vertex in $V(G)\setminus V(G_1)$. 
		\end{defn}
		Clearly, if $G$ is a connected graph then there is just a single connected component, namely the graph $G$ itself. 

		\begin{thm} \label{thm:rate}
			Let $G$ be a connected, $(r+1)$-regular bipartite graph with $N$ vertices and having girth $\geq t+1$. The code $\mathcal{C}$ obtained as the output of Construction \ref{con:Raw_bipartite} with the graph $G$ as input
			is a seq-LRC with parameters $(n = \frac{N(r+1)}{2},k = n-N+1,r,t)$ over $\mathbb{F}_2$ and hence having rate given by:
			\bea
			\frac{r-1}{r+1}+\frac{1}{n}. \label{RateofRawatConst}
			\eea
		\end{thm}

		\begin{proof}
			Let $H$ be the node-edge incidence matrix of the graph $G$.
			From the description of Construction \ref{con:Raw_bipartite}, the matrix $H$ is a p-c matrix of the code $\mathcal{C}$.
			The p-c matrix $H$, has each row of Hamming weight $(r+1)$ and each column of weight $2$.  It follows that the sum of all the rows of $H$ is the all-zero vector.  Thus the rank of $H$ is $\leq N-1$.  

			Next, let $\ell$ be the smallest integer such that a set of $\ell$ rows of $H$ add up to the all-zero vector. Let $M$ be the set of nodes in $G$ corresponding to a set of $\ell$ rows $\underline{r}_1,\ldots,\underline{r}_{\ell}$ in $H$ such that $\sum_{i=1}^{\ell} \underline{r}_i=\underline{0}$.  We note that any edge $(u,v)$ in $G$ with $u \in M$ will be such that $v \in M$ and similarly if $v \in M$ then $u \in M$. Let $S=\cup_{i=1}^{\ell}\text{supp}(\underline{r}_i)$, it follows that the subgraph of $G$ with vertex set equal to $M$ and the edge set equal to the edges associated to columns of $H$ indexed by $S$ form a connected component of the graph $G$.  But since $G$ is connected, $\ell=|M|=N$ and hence $S=[n]$.  It follows that any set of $N-1$ rows of $H$ is linearly independent.  Hence the rank of $H$ equals $N-1$.   The parameters of $\mathcal{C}$ are thus given by:
			\bean
			\text{block length} \ n & = & \frac{N(r+1)}{2} \\
			\text{dimension} \ k & = & \frac{N(r+1)}{2} - (N-1) \\
			\text{rate} \ R & = & 1 - \frac{2(N-1)}{N(r+1)} =  1 - \frac{2}{r+1} + \frac{2}{N(r+1)}  =  \frac{r-1}{r+1}+\frac{1}{n}.
			\eean
		\end{proof}
		We note here that while the Construction~\ref{con:Raw_bipartite} made use of regular bipartite graphs, the bipartite requirement is not a requirement as in the argument above, we only used the fact that the graph $G$ is regular.   We collect together the above observations concerning rate and sufficiency of the regular-graph requirement into a (slightly) modified construction. 
		
		\begin{const} \label{con:modified} (modified version of the construction in \cite{RawMazVis})
			Let $G$ be a connected, regular graph of degree $(r+1)$ and of girth $\geq t+1$ having exactly $N$ vertices. Let $H$ be the $(N \times \frac{N(r+1)}{2})$ node-edge incidence matrix of the graph $G$ with each row representing a distinct node and each column representing a distinct edge. The code $\mathcal{C}$ with p-c matrix $H$ is a seq-LRC having parameters $(n = \frac{N(r+1)}{2},k=n-(N-1),r,t)$ over $\mathbb{F}_2$. The construction takes $G$ as input and constructs code $\mathcal{C}$ as output.
		\end{const}
		
		For the rest of this section: let $r,t$ be arbitrary but fixed positive integers. Let $G$ be a connected, regular graph of degree $(r+1)$ and of girth $\geq t+1$ having exactly $N$ vertices. Let $\mathcal{C}$ be the seq-LRC having parameters $(n = \frac{N(r+1)}{2},k=n-(N-1),r,t)$ over $\mathbb{F}_2$ obtained as the output of the Construction \ref{con:modified} with the graph $G$ as input.

		Clearly, the rate of the code $\mathcal{C}$ is maximized by minimizing the block length $n=\frac{N(r+1)}{2}$ of the code, or equivalently, by minimizing the number of vertices $N$ in $G$. Thus there is interest in regular graphs of degree $r+1$, having girth $\geq t+1$ with the least possible number of vertices.   This leads us to the Moore bound and Moore graphs.

		\begin{thm} \label{thm:Moore} \textbf{(Moore Bound)} (\cite{DynCageSur})
			The number of vertices $L$ in a regular graph of degree $r+1$ and girth $\geq t+1$ satisfies the lower bound :
			\bean
			L \ \geq N_{r,t} & := &  1+\sum_{i=0}^{s}(r+1)r^i,  \ \  \ \ \text{ for $t=2s+2$ even }, \label{eq:Moore:Even} \\
			L \ \geq N_{r,t} & := & 2\sum_{i=0}^{s}r^i,  \ \ \ \ \ \ \ \ \ \ \ \ \ \ \  \text{ for $t=2s+1$ odd }.  \label{eq:Moore:Odd}
			\eean
		\end{thm}
		\begin{defn} \label{def:Mooregraph} \textbf{(Moore graphs)}
			A regular graph with degree $r+1$ with girth at least $t+1$ with number of vertices $L$ satisfying $L = N_{r,t}$ is called a Moore graph.
		\end{defn}
		\begin{thm} \label{thm:MooreExis} \textbf{(Existence of Moore graphs)} (\cite{DynCageSur})
			There exists a Moore graph of degree $r+1$ and girth $t+1$ if and only if
			\begin{enumerate}[(a)]
			\item $r = 1$ and $t \geq 2$, (cycles);
			\item $t = 2$ and $r \geq 1$, (complete graphs);
			\item $t = 3$ and $r \geq 1$, (complete bipartite graphs);
			\item $t = 4$ and \bit 
			\item $r = 1$, (the 5-cycle), 
			\item $r = 2$, (the Petersen graph), 
			\item $r = 6$, (the Hoffman-Singleton graph), 
			\item and possibly $r = 56$; \eit 
			\item $t = 5, 7, or 11$, and there exists a symmetric generalized $n$-gon of order $r$.
			\end{enumerate} 
	    \end{thm}
		\begin{lem} \label{MooreOptimality}
			The rate $R$ of the code $\mathcal{C}$ with block length $n=\frac{N(r+1)}{2}$ satisfies:
			\bea
			R = \frac{r-1}{r+1}+\frac{2}{N(r+1)} & \leq & \frac{r-1}{r+1}+\frac{2}{N_{r,t}(r+1)} \label{eq:rate_bound_Moore}
			\eea
			
			The inequality \eqref{eq:rate_bound_Moore} will become equality iff $G$ is a Moore graph.
		\end{lem}
		\begin{proof}
			The Lemma follows from Theorem \ref{thm:rate}, Construction \ref{con:modified} and the Moore bound given in Theorem \eqref{thm:Moore}.
		\end{proof}
		It turns out interestingly, that the upper bound on rate of the code $\mathcal{C}$ given by the expression $\frac{r-1}{r+1}+\frac{2}{N_{r,t}(r+1)}$ (Lemma \ref{MooreOptimality}) is numerically, precisely equal to the right hand side of inequality \eqref{Thm1} (for $t$ even), \eqref{Thm2} (for $t$ odd). As the inequality \eqref{Thm1} (for $t$ even), \eqref{Thm2} (for $t$ odd) gives an upper bound on rate of a seq-LRC, we have the following Corollary.
		\begin{cor}
			Let $r \geq 3$. The seq-LRC $\mathcal{C}$ having block length $n=\frac{N(r+1)}{2}$ is a rate-optimal code iff $G$ is a Moore graph.
		\end{cor}
		\begin{cor}
			Let $r \geq 3$. Let $(r,t)$ be such that a Moore graph exists. If $G$ is the Moore graph then the code $\mathcal{C}$ is not only rate optimal, it also has the smallest block length possible for a binary rate-optimal seq-LRC for the given parameters $r,t$. 
		\end{cor}
		\begin{proof}
			From the discussions in Section \ref{sec:graphical_rep}, a rate-optimal code must have the graphical representation \ginf. From Theorem \ref{thm:girthreq}, \ginf\ must have girth $\geq t+1$, if the rate-optimal code is over $\mathbb{F}_2$. Hence from Theorem \ref{LowerBoundOna0}, the number $a_0$ of vertices in $V_0$ satisfies the lower bound 
			\bea
			a_0 & \geq & r+1. \label{eq:a0}
			\eea
			Since the value of $a_0,r$ through numerical computation  gives the number of edges in \ginf, it can be seen that \ginf\ with girth $\geq t+1$ and $a_0=(r+1)$ leads to a binary rate-optimal code having block length equal to $\frac{N_{r,t} (r+1)}{2}$ by Corollary \ref{cor:equality_conditions}. The code $\mathcal{C}$ also has block length equal to $\frac{N_{r,t} (r+1)}{2}$ when $G$ is a Moore graph. Hence the corollary follows from \eqref{eq:a0} and noting that the number of edges in the graph \ginf\ grows strictly monotonically with $a_0$ and hence $a_0=r+1$ correspond to least possible block length.
		\end{proof}

		%
	\begin{example}	Set $t=4,r=6$. An example Moore graph with $t=4,r=6$, known as the Hoffman-Singleton graph is shown in the Fig.~\ref{fig:Moore}. If  $G$ is this Moore graph then the code \calc\ is a rate-optimal code with least possible block-length with $n=175$. \end{example} 
		Since Moore graphs exist (Theorem \ref{thm:MooreExis}) only for a very sparse set of $(r,t)$ parameters, we now present a different and general construction of rate-optimal seq-LRC for any $r \geq 3$, and any $t$.


   \section{Rate-Optimal Code Construction by Meeting Girth Requirement} \label{sec:complete_code}

Throughout this section we assume $r \geq 3$. As noted in Section~\ref{sec:girth_req}, to complete the construction of a rate-optimal code over $\mathbb{F}_2$ for either the $t$ even or $t$ odd case, we need to ensure that the graph \ginf\ has girth $\geq t+1$. Since the code is binary, the node-edge incidence matrix of a graph \ginf\ with girth $\geq t+1$ will directly yield the p-c matrix of a rate-optimal code.   Hence in this section, we focus only on the construction of a graph \ginf\ having girth $\geq t+1$.  
The construction of \ginf\ that we present here has a significantly smaller number of nodes and correspondingly smaller block length in comparison to the construction previously presented by us in \cite{BalKinKum_ISIT}.   The outline of a construction that differs from the one presented here in terms of the manner in which a certain base graph appearing in the construction is colored, appears in our paper \cite{BalKinKum_NCC}.  
We describe the construction below. Throughout the construction, whenever we speak of coloring the edges of a graph with $b$ colors or of edge coloring of a graph with $b$ colors, we will mean an assignment of a set of $b$ colors to the edges of the graph such that each edge is assigned a color and no two adjacent edges, i.e., no two edges incident on the same vertex, have the same color.  We begin by introducing the notion of  a perfect matching. 

\begin{defn} \label{defn:perfect_matching}
A {\em perfect matching} of a graph  \aux\ is a subset ${\cal E}$ of edges of  \aux, such that each vertex of the graph is incident on precisely one edge belonging to ${\cal E}$.  A graph is said to have $r+1$ pairwise disjoint perfect matchings if there are $r+1$ pairwise disjoint subsets of edges, each of which is a perfect matching of the graph.
\end{defn}
  
  In the following we construct the graph \ginf\ by first constructing its sub-graph \gzero\ such that its girth $\geq t+1$.  This then concludes the construction because if \gzero\ has girth $\geq t+1$ then we can add a node $V_{\infty}$ and connect it to each node in $V_0$.  The resulting graph will have the same node-edge incidence structure as does $\ginf$ and will have girth $\geq t+1$.  

We next present an overview of the construction.  The construction proceeds in $3$ steps. 

\begin{enumerate}[(i)]
\item Recall that the graph \gzero\ is the graph \ginf\ restricted to the node set $V_0 \cup \cdots \cup V_s$.  Equivalently, \gzero\ is the graph \ginf\ obtained by deleting node $V_{\infty}$ and the edges incident on node $V_{\infty}$.  As the first step, we construct a graph \gbase\ having the same structure as the graph \gzero, but with arbitrary girth. Hence this graph \gbase\ has the same incidence matrix as the corresponding matrix \hzero\ for $t$ even or odd  i.e., \eqref{eq:stair_H_even_pun_ch4} for $t$ even and \eqref{eq:stair_H_odd_pun_ch4} for $t$ odd. We will also impose an additional condition, namely that the graph \gbase\ permit an edge coloring with $r+1$ colors.  Turns out that in the case $t$ even, the requirement of an edge coloring with $(r+1)$ colors requires a special construction of \gbase, whereas, in the case of $t$ odd, any graph \gbase\ has an edge coloring using $r+1$ colors. 
\item Next, pick a graph \aux\ with $m$ nodes having both girth $\geq t+1$ as well as a set of $r+1$ pairwise disjoint perfect matchings.  
\item In the third step, create $m$ replicas of the graph \gbase.   If there is an edge between nodes $a,b$ in \gbase\ of color $i$, then replace that edge with the $i^{th}$ perfect matching of \aux\ between the replicated copies of the nodes $a,b$. This process is illustrated in Fig.~\ref{fig:matching} below.  This results in a graph having the structure of \gbase\ but will turn out to have girth $\geq t+1$. This concludes the outline of the construction. 
\een

	\begin{figure}[ht]
		\centering
		\includegraphics[scale = 0.6]{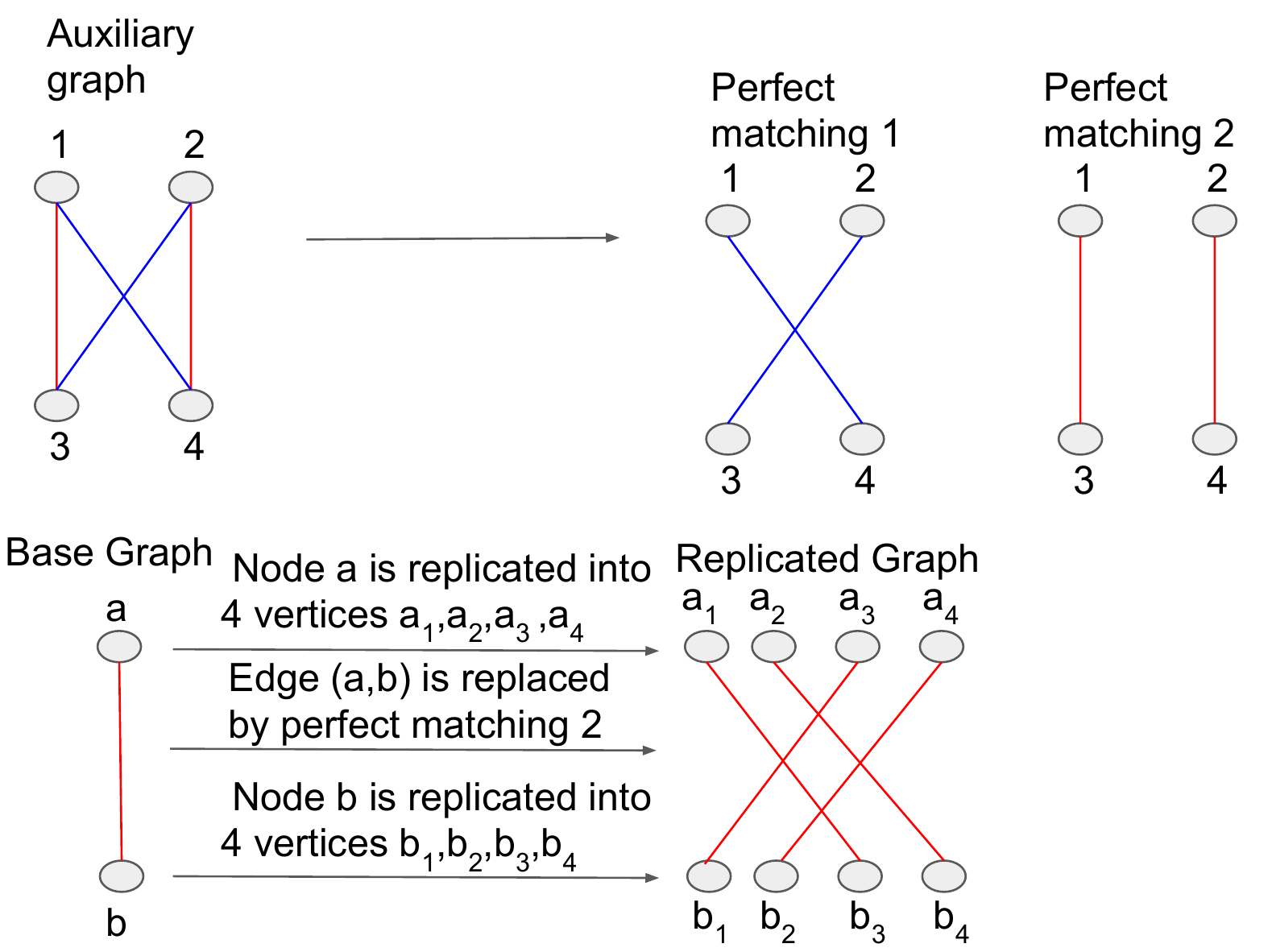} 
		\caption{Explaining how matching is employed in the construction. }
		\label{fig:matching}
	\end{figure}

The total number of nodes in the resultant graph is given by
\bean
|V(\gbase)| \times |V(\aux)| = a_0 \left(1+r+\cdots +r^s\right) \times O(r^{t+1}), & & \text{for $t$ even and} \\
|V(\gbase)| \times |V(\aux)| =  a_0 \left(1+r+\cdots +r^{s-1}+r^{s-1} \frac{r}{r+1} \right) \times O(r^{t+1}), & & \text{for $t$ odd}. 
\eean
This follows from Corollary \ref{cor:equality_conditions} together with the additional observation that the graph \aux\ can be constructed with $O(r^{t+1})$ nodes as will subsequently be shown. We will also see that we can choose $a_0 \leq r+1$ in the above expression, since the construction of \gbase\ in the first step does not require girth $\geq t+1$.    Each step involved in the construction is described in greater detail below. 

\subsection*{Step $1$ : Construction and Edge Coloring of the Base Graph} 
	
	\ben
	\item Select a graph \ginf\ as described in Section \ref{sec:graphical_rep}, with an arbitrary regular graph as the subgraph induced by nodes in $V_s$ in the case of $t$ even and an arbitrary biregular bipartite graph \gsmone\ as the subgraph induced by nodes in $V_{s-1} \cup V_s$ in the case $t$ odd. We begin by deleting the edges in the graph \ginf\ connecting $V_{\infty}$ to the nodes in $V_0$.  One is then left with the graph \gzero\ where each of the nodes in $V_0$ has degree $r$ and all the remaining nodes have degree $(r+1)$.  Thus in particular, every node has degree $\leq (r+1)$. We shall call the resultant graph the {\em base graph} \ \gbase. Hence this graph \gbase\ has the same incidence matrix as the corresponding matrix \hzero\ for $t$ even or odd i.e., \eqref{eq:stair_H_even_pun_ch4} for $t$ even and \eqref{eq:stair_H_odd_pun_ch4} for $t$ odd.  If the graph \gbase\ has girth $\geq t+1$, the construction ends here. If not, as noted above, we will modify \gbase\ to construct a second graph ${\cal J}_{\infty}$ with girth $\geq t+1$, having an incidence matrix that has the form appearing in \eqref{eq:stair_H_even_aug_ch4} for $t$ even, and in \eqref{eq:stair_H_odd_aug_ch4} for $t$ odd. 

	\item Since every node in \gbase\ has degree $\leq r+1$, it follows from Vizing's theorem \cite{Viz}, that the edges of \gbase\ can be colored using $\ell \leq (r+2)$ colors. However, as will be seen below, it is possible to color the edges of \gbase\ using $\ell = (r+1)$ colors.  We discuss separately, the cases of $t$ even and $t$ odd.  
	\bit
	\item {\bf Case $t$ even:}   For $t=2$ and hence, $s=0$, we can choose \gbase\ to be a complete graph of degree $r+1$ and the complete construction of rate-optimal code ends here as the graph \gbase\ has girth $t+1=3$.  For $t=4$ and hence $s=1$, we can construct \gbase\ with girth $\geq t+1=5$ and the complete construction of rate-optimal code ends here. For the sake of brevity, we skip this part ($t=4$) of the proof and refer the reader instead to our arXiv publicaiton~\cite{BalPraKum4Era}. The construction for $t=4$ in \cite{BalPraKum4Era} has $a_0=O(r^2)$. In the following we give a construction of rate-optimal codes with $t=4$ with $a_0=2r$ for some selected values of $r$ ($r$ such that a Moore graph of degree $r+1$ and girth $6$ exists). Take a Moore graph $G$ (if it exists) of degree $r+1$ and girth $6$. Take an edge $(u,v)$ in $G$. Now merge the nodes $u,v$ into a single node $V_{\infty}$ and connect all the neighbours of $u,v$ to $V_{\infty}$ to form the graph $G_{u,v}$. By expanding the neighbourhood structure of $V_{\infty}$ in the form of a tree i.e., by placing $V_{\infty}$ as root node and by placing the neighbours of $V_{\infty}$ at depth $1$ and placing neighbours of neighbours of $V_{\infty}$ at depth $2$ and by noting that all the nodes in the graph $G_{u,v}$ has appeared exactly once in this expansion (because $G$ is a Moore graph (Definition \ref{def:Mooregraph}) of degree $r+1$ and girth $6$) we can see that the graph $G_{u,v}$ has the same structure as $\mathcal{G}_{\infty}$ with $a_0=2r$ and it also has girth $\geq 5$.
	
	For the general $t \geq 6$, $s \geq 2$ case, one can set  $a_0=4$ and through careful selection of the edges connecting nodes in $V_s$, ensure that the edges of the graph \gbase\ can be colored using $(r+1)$ colors but with \gbase\ having arbitrary girth, details are provided in Appendix \ref{Gbasecolouring_ch4}.   Thus the base graph will in this case have 
	\bea
	\nbase\ \ = \ a_0(1+r+\cdots +r^s) & = & 4\left( \frac{r^{s+1}-1}{r-1} \right) \label{eq:tevenblklength}
	\eea
	vertices.

	\item {\bf Case $t$ odd:} \ In the $t$ odd case, by selecting $a_0=(r+1)$, one can ensure that the edges of the graph \gbase\ can be colored using $(r+1)$ colors but with \gbase\ having arbitrary girth, see Appendix \ref{Gbasecolouring_ch4} for details.   Thus the base graph will have in this case, a total of 
	\bean
	\nbase\ \ = \ a_0(1+r+\cdots +r^{s-1}+r^{s-1} \frac{r}{r+1}) & = & (r+1)\left( \frac{r^{s}-1}{r-1}\right) +r^s \label{eq:toddblklength}
	\eean
vertices. 
	
	\eit
	The coloring is illustrated in Fig.~\ref{fig:base_graph_coloring} for the case $t=5$ and $r=3$. 
	\een
	\begin{figure}[ht]
		\centering
		\includegraphics[scale = 0.35]{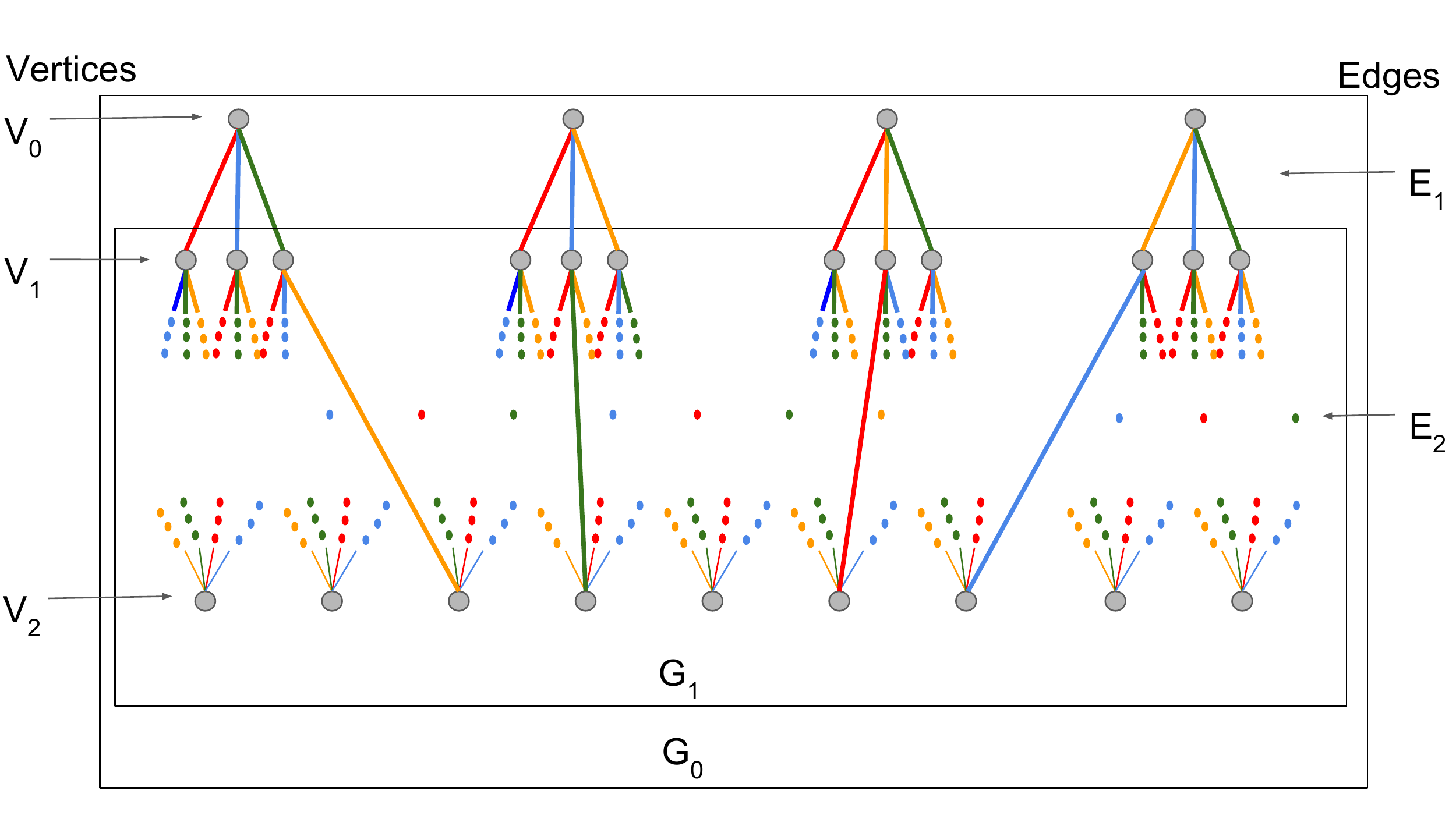}
		\caption{ An example base graph \gbase\ with associated coloring of the edges using $(r+1)$ colors.  Here $t=5$, $r=3$, $a_0=r+1=4$ with the edges of \gbase\ colored with $r+1=4$ colors.}
		\label{fig:base_graph_coloring}
	\end{figure}
	In summary, the base graph \gbase\ corresponds to the graph \gzero\  having incidence matrix as the corresponding matrix \hzero\ for $t$ even or odd.  All the nodes in \gbase\ are of degree $ \leq (r+1)$ and the edges of the graph can be colored using $(r+1)$ colors.   These are the only properties of the base graph \gbase\ that are carried forward to the next steps of the construction.   We will number the colors $1$ through $r+1$ and speak of color $i$ as the $i$th color. The steps that follow are the same for either $t$ even or $t$ odd.  
	
	\subsection*{Step $2$ : Construction and Coloring of the Auxiliary Graph} 
	
	 Next, pick a graph \aux\ referred to as {\em auxiliary graph} that has both girth $\geq t+1$ as well as a set of $r+1$ pairwise disjoint perfect matchings.  Let ${\cal E}_1,\ldots,{\cal E}_{r+1}$ be the subsets of edges corresponding to the $r+1$ pairwise disjoint perfect matchings.   Let the edges in ${\cal E}_i$ be colored using color $i$. Let us denote by  \auxi, the subgraph of \aux\ with edge set exactly equal to ${\cal E}_i$. The coloring of an example auxiliary graph is shown in Fig.~\ref{fig:aux_graph_coloring}.  Note that the edges of \aux\ are colored using the same set of colors used to color the edges of \gbase\ in Step 1 above.
	\begin{figure}[h!]
		\centering
		\includegraphics[angle=270,scale = 0.45]{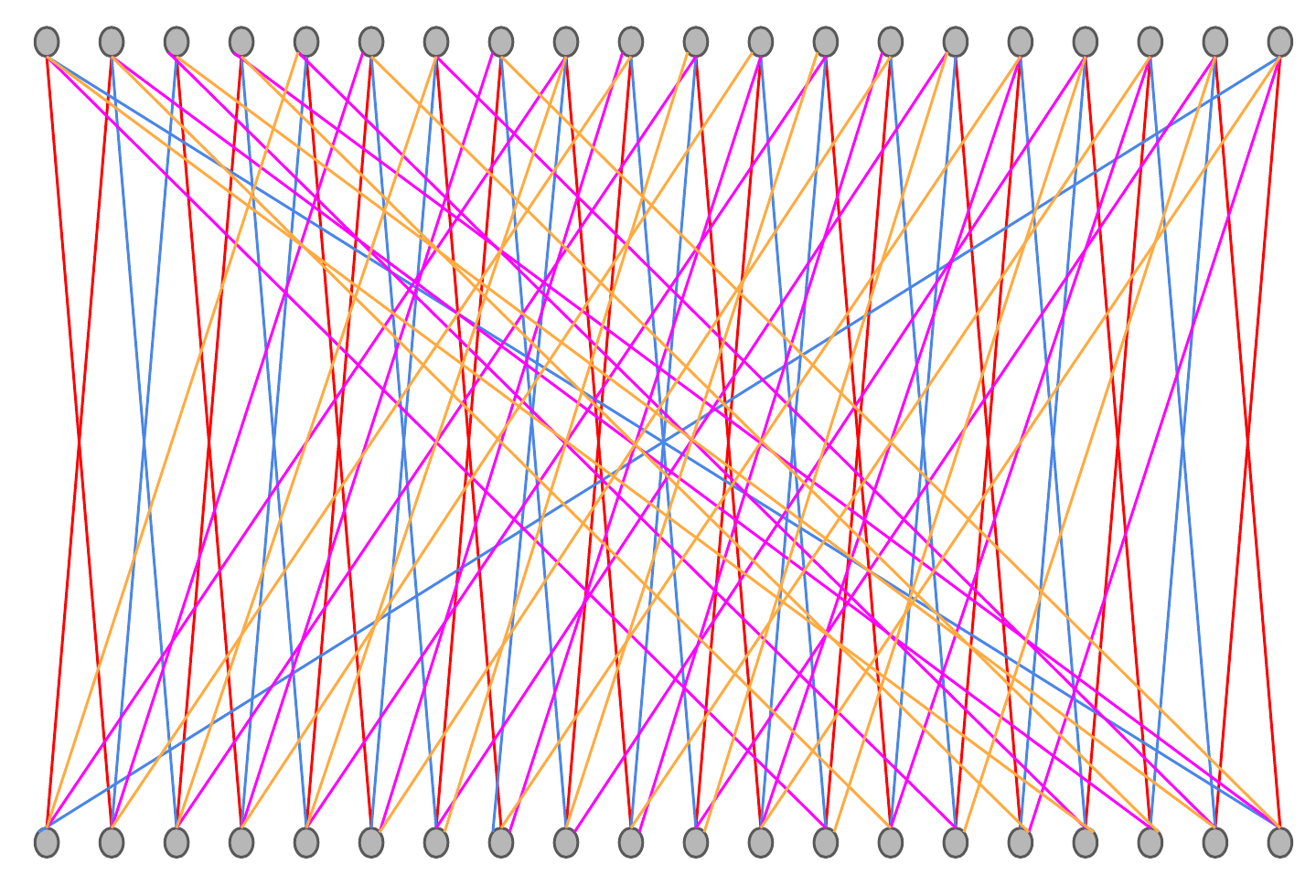}
		\caption{An example auxiliary graph \aux\ whose edges are colored using $(r+1)$ colors.  Here $r=3$, so there are $r+1=4$ colors.  This graph is a regular bipartite graph of degree $r+1=4$ with $\naux=40$ vertices and having girth $\geq 6$.  This graph can be seen to have a set of $4$ pairwise disjoint perfect matchings.  }
		\label{fig:aux_graph_coloring}
	\end{figure}
	It is shown in Theorem \ref{bipartiteColouring} of Appendix \ref{Gbasecolouring_ch4}, that every $(r+1)$-regular bipartite graph $G$ has a set of $(r+1)$ pairwise disjoint perfect matchings.    It follows that an $(r+1)$-regular bipartite graph $G$ of girth $\geq (t+1)$ meets the requirements placed on the auxiliary graph \aux.  
We can even create the auxiliary graph by starting with a general regular graph ${\cal P}$ of degree $(r+1)$ and girth $\geq t+1$.    To do this, one first replicates the vertices of ${\cal P}$ to produce two replicas ${\cal P}_1$ and ${\cal P}_2$ of the vertices of ${\cal P}$.   Next one places edges between the two replicas as follows. If there is an edge between nodes $u,v$ in ${\cal P}$, we draw edges $(u_1,v_2)$ and $(u_2,v_1)$ where $(u_i,v_i)$ are the vertices in ${\cal P}_i$ respectively, associated with $u,v$.    Let ${\cal Q}$ denote the resultant graph.  Then it is easily seen that ${\cal Q}$ is a regular bipartite graph of degree $(r+1)$ and has girth $\geq t+1$.  
In this way, we have created a bipartite regular graph ${\cal Q}$ by starting from a  regular graph ${\cal P}$ of the same degree while maintaining the girth requirement.  
This construction is of interest since an $(r+1)$-regular graph having girth $\geq t+1$ can be constructed having a relatively small number, $O(r^{t+1})$ of nodes by drawing from the results in \cite{Lubotzky1988,X_Dahan,Mor,DavSarVal,LazUstWol}.  By Theorem \ref{thm:Moore}, this is close to the smallest possible.

	\subsection*{Step $3$ : Using the Auxiliary Graph to Expand the Base Graph} 
   In this final step, we use the auxiliary graph \aux\ to expand the graph \gbase, creating in the process, a new graph \jzero. The graph \jzero\ will be our desired graph \gzero\ with girth $\geq t+1$.  
	\begin{enumerate}[(i)]
	\item Let the auxiliary graph \aux\ have $m$ nodes and let these nodes be identified with the integers in $\{1,2,\ldots,m\}$.
	\item For every node $a$ in \gbase\, we create $m$ replicas $a_1,\ldots,a_m$. Let the resulting $m$-fold replicated graph be denoted by ${\cal J}_0$.
    \item Next, we describe the edge set of the graph ${\cal J}_0$. If \gbase\ has an edge $(a,b)$ of color $i$ then we connect the corresponding nodes $a_1,\ldots,a_m$ and $b_1,\ldots,b_m$ in ${\cal J}_0$ as follows: we connect $a_{\ell_1}$ to $b_{\ell_2}$ and $a_{\ell_2}$ to $b_{\ell_1}$ for every edge $(\ell_1,\ell_2)$ in \auxi. This is illustrated in Figure \ref{fig:matching}.
	\een
	
	\begin{thm}
		${\cal J}_0$ has girth $\geq t+1$.
	\end{thm}
	\begin{proof}
		This follows simply because corresponding to every path traversed through ${\cal J}_0$ with successive edges in the path with some sequence of colors, there is a corresponding path in \aux\ with successive edges in the path corresponding to the same sequence of colors. The edge coloring of the base graph ensures that we never retrace our steps in the auxiliary graph i.e., any two successive edges in the path in auxiliary graph are not the same.  It follows that since \aux\ has girth $\geq t+1$, the same must hold for ${\cal J}_0$. 
	\end{proof} 

  A little thought will now show that the graph \jzero\ has the same structure as \gbase.   
 \bit
 \item In the $t$ even case, both graphs \jzero, \gbase\ can be regarded as the union of $r$-ary trees of the same depth followed by a pairing of the leaf nodes and the drawing of an edge between pairs of leaf nodes is done in such a way that the graph restricted to only leaf nodes is an $r$-regular graph.  The two differences are that \gbase\ is the union of $a_0$ trees and the girth of \gbase\ is unknown, while \jzero\ can be regarded as the union of $a_0 |V(\aux)|=a_0 m$ trees and has girth $\geq t+1$. 
 \item In the $t$ odd case, the graph \gbase\ can be regarded as containing the union of $a_0$ $r$-ary trees of depth $s-1$ starting from depth $0$; additionally, the leaf nodes are the upper nodes of a bi-regular, bipartite graph of degrees $r,r+1$, so the overall graph is the union of the $a_0$ trees together with the vertex set $V_s$ making up the lower nodes of the bipartite graph.  The graph \jzero\ has the same structure except once again that while \gbase\ contains the union of $a_0$ trees and the girth of \gbase\ is unknown, \jzero\ contains the union of $a_0 | V(\aux) |=a_0 m$ trees and has girth $\geq t+1$.  In both \gbase\ and \jzero\ there is a bi-regular, bipartite graph of degrees $r,r+1$ at the very bottom, whose upper nodes are the leaf nodes of the different trees.  
\eit

%
Let $U_0 \subseteq V({\cal J}_0)$ be the $m$-fold replication of the vertices in $V_0$ of \gbase. To complete the construction, we add a node $U_{\infty}$ to ${\cal J}_0$ and connect it to each node in $U_0$ through an edge. Let us call the resulting graph ${\cal J}_{\infty}$.  Hence ${\cal J}_{\infty}$ has the same node-edge incidence matrix as  \eqref{eq:stair_H_even_aug_ch4} over $\mathbb{F}_2$  for $t$ even, \eqref{eq:stair_H_odd_aug_ch4} over $\mathbb{F}_2$ for $t$ odd and also has girth $\geq t+1$. This node-edge incidence matrix gives the required p-c matrix for a rate-optimal code over $\mathbb{F}_2$ with parameters $(r,t)$. Hence we have constructed the required p-c matrix for a rate-optimal code over $\mathbb{F}_2$ for both $t$ even and $t$ odd, for any $r \geq 3$.
This concludes the construction of rate-optimal codes for both $t$ even and $t$ odd.
	
\begin{note}
Following the conference presentations of portions of this paper, we discovered that similar graph-expansion constructions using perfect matchings have previously been used to construct large-girth LDPC codes \cite{PraSubTha}. 
\end{note}

\section{Seq-LRC Having Largest Dimension for Given Block length} \label{sec:high_dimension}

 In this section, we focus on constructing dimension-optimal codes. Recall that dimension-optimal codes are codes achieving the bound in Corollary \ref{cor:dimension_bound} with equality. A dimensional optimal seq-LRC need not be a rate-optimal seq-LRC since the upper  bound on rate provided in Theorem \ref{thm:rate_both} is not achievable if the block length is not of the form given by Remark \ref{blk_length_rem}.  In the cases when bound in Theorem \ref{thm:rate_both} is not achievable focus shifts to achieving the bound provided in Corollary \ref{cor:dimension_bound}. While it is not clear whether the bound in Corollary \ref{cor:dimension_bound} is achievable for all block lengths, in the following we illustrate an approach to construct dimension-optimal codes when the bound in Corollary \ref{cor:dimension_bound} is achievable. Linear inequalities involving the parameters $a_0,a_1,\ldots,a_s,a_{s+1},p$ described in the derivation of the rate bound in Theorem \ref{thm:rate_both} when used in linear programming formulation (Appendix \ref{app:linprog}), lead to an upper bound on dimension.   (The parameters $a_0,a_1,\ldots,a_s,a_{s+1},p$ are introduced in the proof of Theorem \ref{thm:rate_both} in the Appendix \ref{Appendix_rate_bounds}).  Hence the linear programming formulation given in  Appendix \ref{app:linprog} can be used to obtain an upper bound on dimension for a given block length. Note that in the linear programming formulation we can constrain all the values of $a_0,a_1,\ldots,a_s,a_{s+1},p,k$ to be integers. While from Appendix \ref{app:linprog}, it is clear that linear programming approach yields the same bound as Corollary \ref{cor:dimension_bound}, it is not clear that they yield the same bound after constraining all the values of $a_0,a_1,\ldots,a_s,a_{s+1},p,k$ to be integers in the linear programming formulation. In all of our simulation results for $t$ even, Corollary \ref{cor:dimension_bound} and the integer linear programming have always yielded the same upper bound on dimension. 

Our approach: we use the values of $a_0,a_1,\ldots,a_s,a_{s+1},p$ obtained via the integer linear programming formulation and which yield the maximum-possible dimension to derive the graphical structure of the code's graphical representation.   By graphical representation, we mean the tree-like structure of the graphical representation of a seq-LRC of given block length that is dimension optimal. Note that as long as $\sum_{i=0}^{s+1} a_i=n$, the tree-like structure holds even for dimension-optimal codes except for the fact that the degrees of nodes is only contrained to be $\leq r+1$ and not $=(r+1)$ and the graph at the bottom of the tree could be different as $p$ could be non-zero and in case of $t$ odd, $a_{s+1}$ could be non-zero. This is because the tree-like structure is applicable as long as the code have minimum distance $\geq (t+1)$ (as opposed to requiring rate optimality) as the tree-like structure comes from Lemma \ref{claim1} for $t$ even and Lemma \ref{claim2} for $t$ odd in Appendix \ref{Appendix_rate_bounds}. 

We start with a rate-optimal code having parameters close to the parameter set of the code that we are trying to construct and modify it according to the values of $a_0,a_1,\ldots,a_s,a_{s+1},p$ to derive dimension-optimal codes for the target parameters. We do not have an explicit algorithm to modify the rate-optimal code.  Rather, we illustrate this approach to constructing dimension-optimal codes by providing a few hand-crafted examples. All of the example constructions shown correspond to the parameter set: $t = 4$, $r = 3$. The values of $a_0,a_1,a_2,p$ which yield the maximum dimension $k$ for a given value of $n$ are obtained by solving the integer linear programming problem using MATLAB. Table \ref{table:dim} lists the values of $a_0$, $a_1$, $a_2$ and $p$ which yield the maximum value for dimension $k$ say $k_{max}$, for a number of values of $n$ and $t = 4$, $r = 3$. While complying with the structure dictated by the parameters $a_0,a_1,a_2,p,n,k_{max},r,t$, we try to connect the nodes with edges, maintaining the right degrees. 

In all of the graphs described below, we assume as always that each edge represents a distinct code symbol and that each node represents the parity check of the code symbols represented by edges incident on it.
Note that some rows in Table \ref{table:dim} have non-zero values for $p$, whereas, as the derivation using linear programming shows, the value of $p$ for rate-optimal codes is zero.  A few example graphs are shown next based on our approach.  The graphs below are derivatives of a Moore graph of degree $4$ and girth $6$ with $26$ vertices. Note that the code corresponding to the Moore graph is both rate and block-length optimal for $t=5$.  As explained in Section \ref{sec:complete_code}, we use this Moore graph to get a {\em rate-optimal} code having parameters $t=4$, $r=3$, $a_0=2r=6$, $n = 51$ and $k = 27$. This is shown in the Figure \ref{fig:olsgraph}. Further, modification of this graph turns out to yield examples of codes that are dimension optimal. 

\begin{table}
	\begin{center}
		\begin{tabular}{||c|c|c|c|c|c||}
			\hline
			$n$ & $k_{\text{max}}$ & $a_0$ & $a_1$ & $a_2$ & $p$\\
			\hline\hline
			$\mathbf{43}$ & $\mathbf{22}$ & $\mathbf{6}$ & $\mathbf{15}$ & $\mathbf{22}$ & $\mathbf{0}$\\
			\hline
			$44$ & $23$ & $5$ & $15$ & $24$ & $1$\\
			\hline
			$45$ & $23$ & $6$ & $16$ & $23$ & $0$\\
			\hline
			$46$ & $24$ & $6$ & $16$ & $24$ & $0$\\
			\hline
			$\mathbf{47}$ & $\mathbf{24}$ & $\mathbf{7}$ & $\mathbf{16}$ & $\mathbf{24}$ & $\mathbf{0}$\\
			\hline
			$48$ & $25$ & $6$ & $17$ & $25$ & $0$\\
			\hline
			$49$ & $25$ & $7$ & $17$ & $25$ & $0$\\
			\hline
			$\mathbf{50}$ & $\mathbf{26}$ & $\mathbf{6}$ & $\mathbf{18}$ & $\mathbf{26}$ & $\mathbf{0}$\\
			\hline
			$\mathbf{51}$ & $\mathbf{27}$ & $\mathbf{6}$ & $\mathbf{18}$ & $\mathbf{27}$ & $\mathbf{0}$\\
			\hline
			$52$ & $27$ & $7$ & $18$ & $27$ & $0$\\
			\hline
			$53$ & $28$ & $6$ & $18$ & $29$ & $1$\\
			\hline
			$54$ & $28$ & $7$ & $19$ & $28$ & $0$\\
			\hline
			$55$ & $29$ & $6$ & $18$ & $31$ & $2$\\
			\hline
			$\mathbf{56}$ & $\mathbf{29}$ & $\mathbf{7}$ & $\mathbf{20}$ & $\mathbf{29}$ & $\mathbf{0}$\\
			\hline
		\end{tabular}
	\end{center}
	\caption{Table showing the upper bound on dimension $k_{\text{max}}$, and the optimal values of $a_i,p$ (values for which the upper bound on dimension is achievable) for different values of block length for $t=4$ and $r=3$. The rows with highlighted text represent the parameters for which we were able to construct dimension-optimal codes.}
	\label{table:dim}
\end{table}
\begin{figure}[ht]
	\centering
	\includegraphics[width=3in]{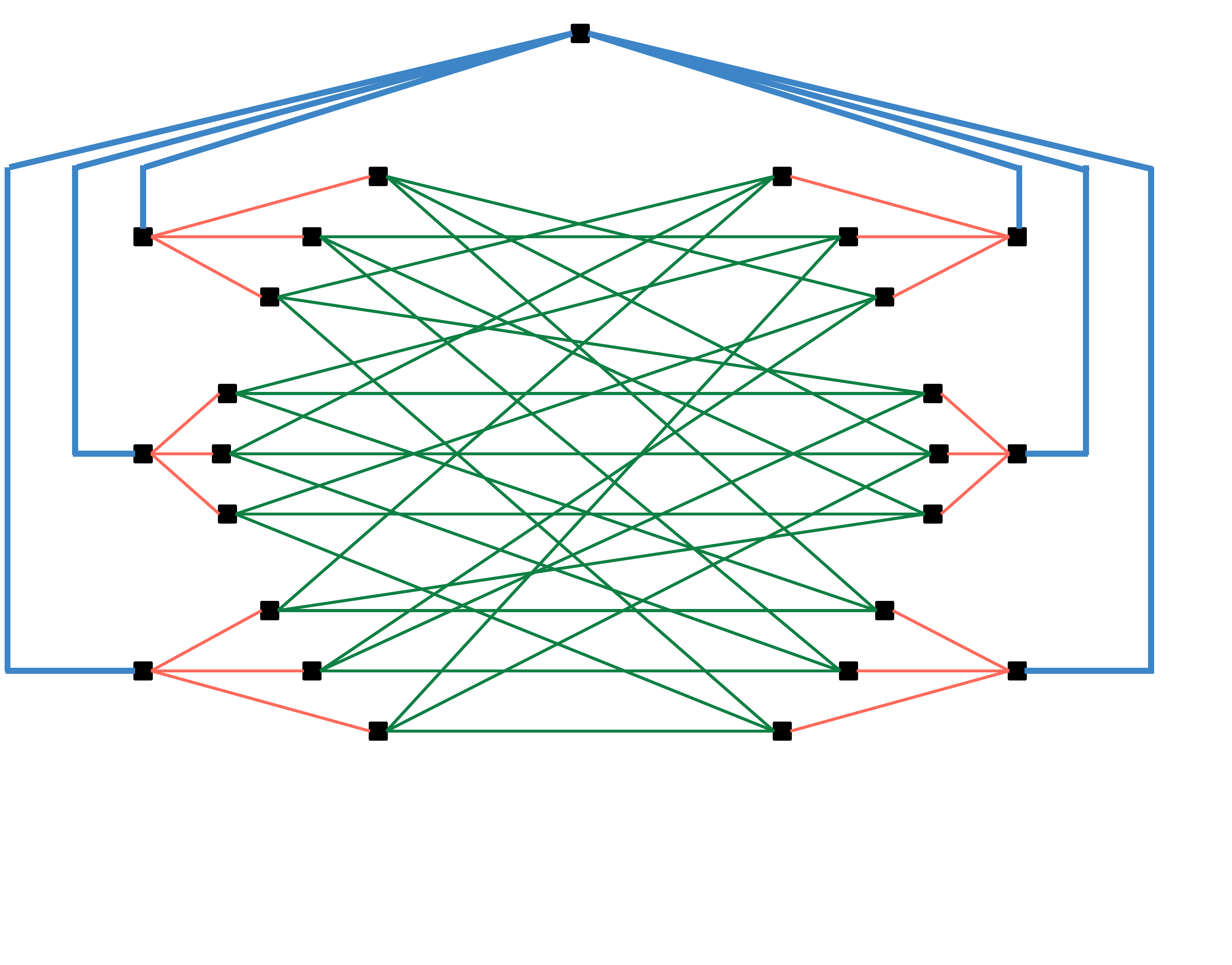}
	\caption{Graphical representation of a dimension-optimal code with $r=3,t=4$ having largest possible dimension $(k=27$) for block length $n=51$. It is also a rate-optimal code.}
	\label{fig:olsgraph}
\end{figure}

Two other example constructions of dimension-optimal codes are shown in the Figures \ref{fig:blkop1} and \ref{fig:blkop4}. Note that these examples use the values of $a_0,a_1,a_2,p$ from the Table \ref{table:dim}. The graphs corresponding to these 2 examples are derived from the graph shown in Figure \ref{fig:olsgraph}. 
\begin{figure}[h!]
	\begin{minipage}[t]{0.45\linewidth}
		\includegraphics[width=\linewidth]{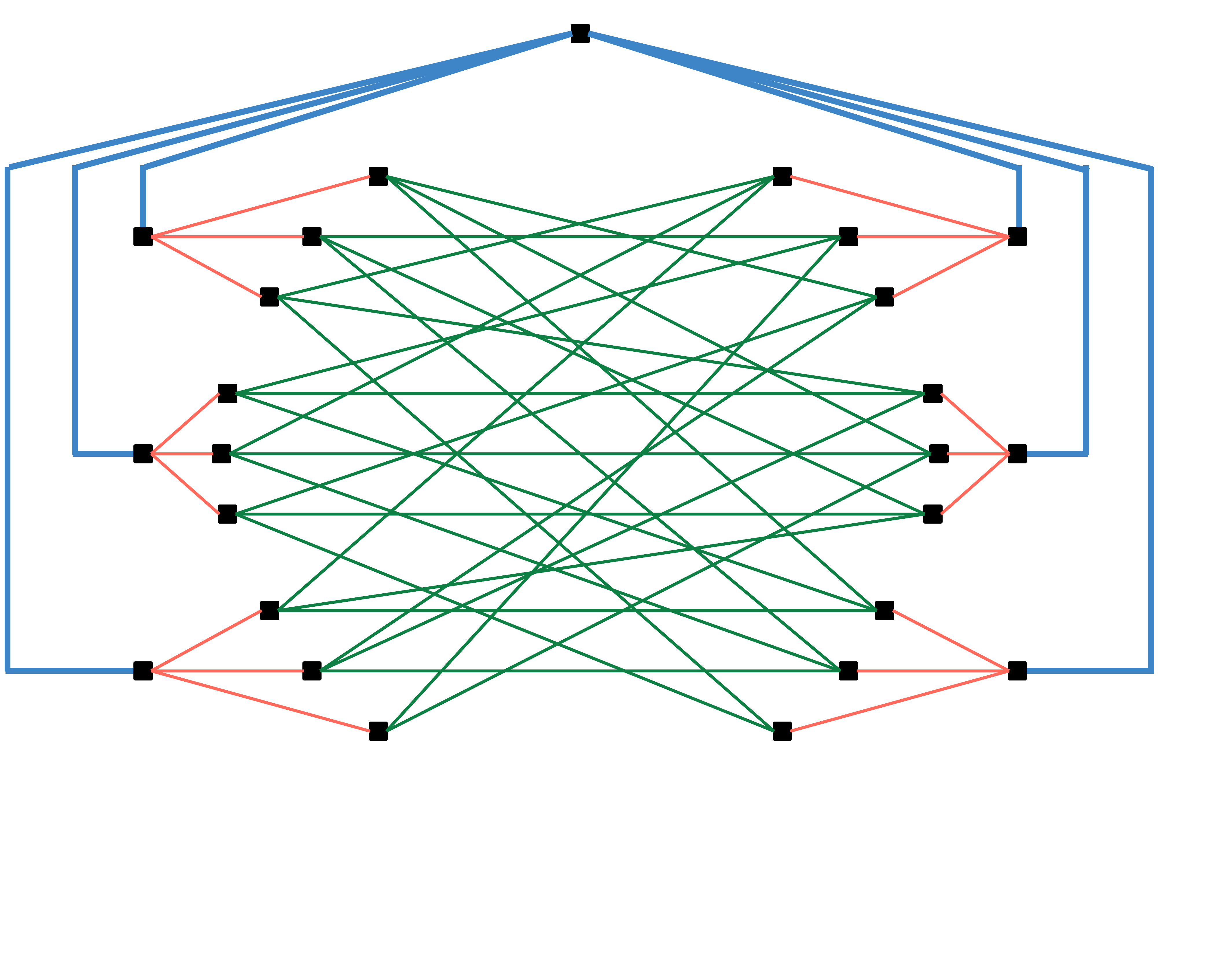}
		\caption{Graphical representation of a dimension-optimal code with $r=3,t=4$ having largest possible dimension $(k=26$) for block length $n=50$.}
		\label{fig:blkop1}
	\end{minipage}%
	\hfill%
	\begin{minipage}[t]{0.45\linewidth}
		\includegraphics[width=\linewidth]{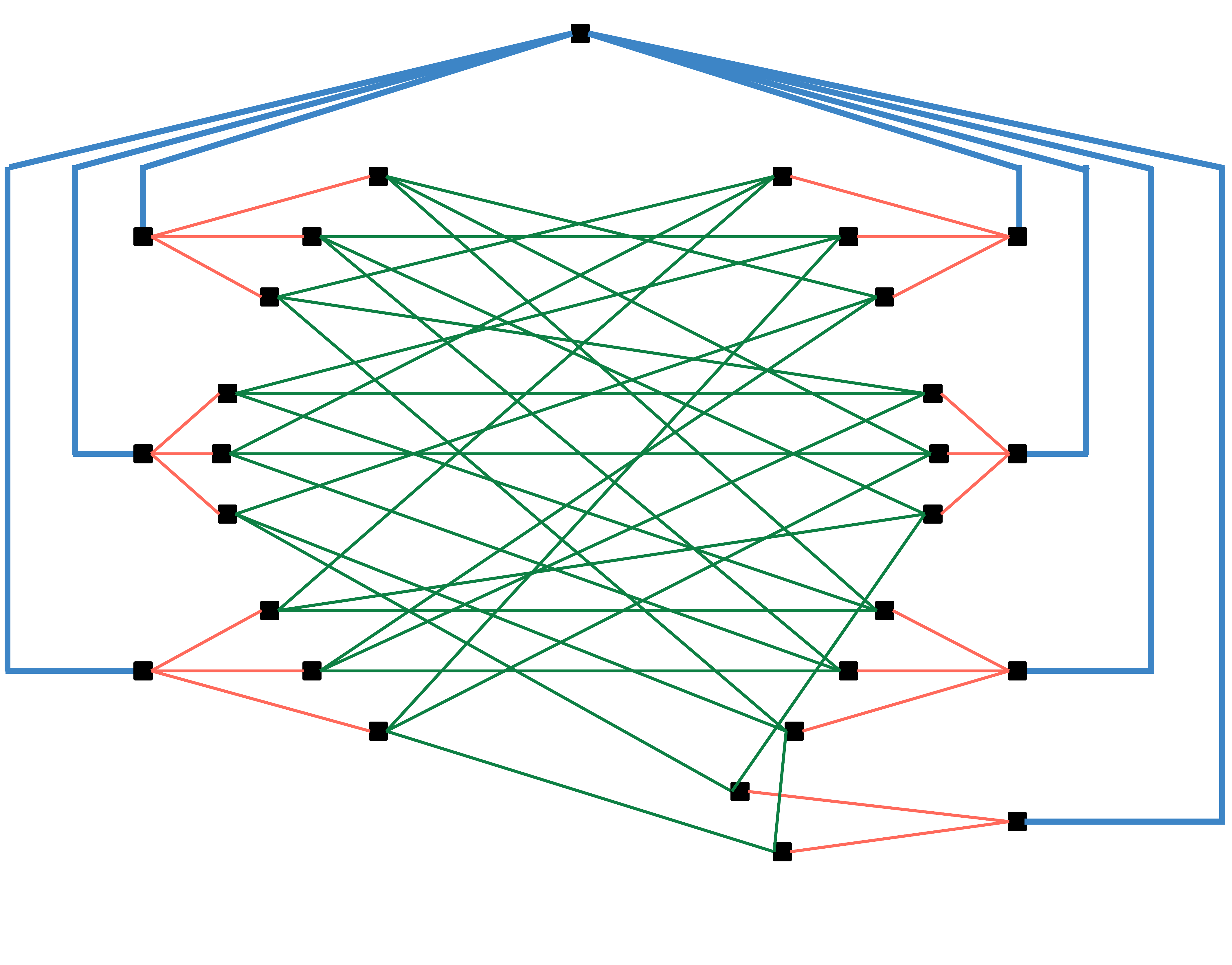}
		\caption{Graphical representation of a dimension-optimal code with $r=3,t=4$ having largest possible dimension $(k=29$) for block length $n=56$.}
		\label{fig:blkop4}
	\end{minipage} 
\end{figure}	 	 
	
	 The girth of the graphs shown in Figure \ref{fig:olsgraph}, \ref{fig:blkop1} and \ref{fig:blkop4} were verified to be at least $5$ (exactly $5$ in some cases) using a computer program that executes a simple algorithm described below.\\
	 \textbf{Simple Algorithm for Calculating Girth of an Undirected Graph}\\	
	 Let $G=(V,E)$ be the graph under consideration. Let the graph not contain any self-loops. Let $E=\{e_1,...,e_M\}$.\\
	 \textbf{Algorithm:}
	 \begin{itemize}
	 	\item \textit{girth}$=\infty$
	 	\item \textit{for} $(i=1,i\leq M,i=i+1)$\\
	 	\ \ \ \ \textit{Remove edge }$e_i$\textit{ from }$E$, \textit{ girth }$=$\textit{ min}\textit{(girth, $1+$ length of the shortest path between nodes previously connected by $e_i $)}
	 \end{itemize}
	 The above algorithm was carried out by representing a graph through its incidence matrix. The shortest path between pairs of nodes was computed by the application of Dijkstra's algorithm \cite{Dijkstra:1959:NTP:2722880.2722945}.
	
		\section{A Connection with Tornado Codes} \label{sec:Tornado}
		 	Tornado codes, introduced in \cite{tornado}, offer an efficient choice for erasure correction as they permit a large number of erasures to be corrected with high probability, while maintaining the complexity of encoding and decoding to a low level. Sometime after we came up with the graphical construction of rate-optimal seq-LRC, we realized that there were structural similarities with the layered-parity-check nature of the graph representing a Tornado code.  We provide below, a brief description of these similarities.

The graphical representation of a Tornado code over a finite field $\mathbb{F}_q$, involves a collection of $m+1$ consecutive bipartite graphs $B_0,\ldots,B_m$ as well as an auxiliary erasure code ${\cal C}_{\text{T}}$\footnote{Not to be confused with the earlier notion of  an auxiliary graph ${\cal A}$.}.  The $i$th graph, $B_i$, has $\beta^i k$ left nodes and $\beta^{i+1}k$ right nodes (with edges drawn between left and right nodes) for $\forall$ $0\leq i \leq m+1$ and for some $0 < \beta < 1$.   
The $\beta^{i+1}k$ right nodes of the graph $B_i$ are identified with the $\beta^{i+1}k$ left nodes of the graph $B_{i+1}$.    The left nodes of $B_0$ are associated to message symbols, the right nodes with parity symbols. When we speak of a parity-check on a collection of symbols $x_i$, we mean a linear constraint $\sum_i a_i x_i=0$ with $a_i \in \mathbb{F}_q$.  The right nodes of $B_1$ are associated to parity symbols representing parity checks imposed on the parity symbols associated to graph $B_0$ etc.  Thus the code symbols can be partitioned into message symbols, parity symbols, parity-upon-parity, parity-upon-parity-upon-parity and so on, as described below: 
\bean
\text{left nodes of $B_0$} & \Rightarrow & \text{message symbols $m_i$} \\ 
\text{left nodes of $B_1$ (or equivalently, right nodes of $B_0$)} &  \Rightarrow & \text{parity symbols $p_i$ on message symbols $m_i$} \\ 
\text{left nodes of $B_2$ (or equivalently, right nodes of $B_1$)} &  \Rightarrow & \text{parity symbols $q_i$ on parity symbols $p_i$}, 
\eean
etc.   
 Hence in any graph $B_i$, a node in the right represents a parity symbol storing the parity of the symbols represented by nodes in the left it is connected with via edges.
In addition, there is a final, set of parity symbols that are derived as follows.  Let ${\cal C}_{\text{T}}$ be an erasure code of rate $(1-\beta)$.  The $k \beta^{m+1} $ parity symbols corresponding to the right nodes of the bipartite graph $B_m$ are fed as message symbols to the erasure code ${\cal C}_{\text{T}}$. This code then generates a further $\frac{\beta^{m+2}k}{1-\beta}$ parity symbols. The overall code can be verified to have rate $(1-\beta)$. A graphical depiction of the above description of a Tornado code is given in Fig. \ref{fig:tornado1}

We show below that rate-optimal seq-LRC also permit a graphical description in terms of  sequence of bipartite graphs.  
		 	
\subsection{Bipartite-Graph-Based Graphical Description of a Rate-Optimal seq-LRC}
		 	Rate-optimal seq-LRC for $t$ even can be viewed as being constructed in a manner similar to Tornado codes.    To see this, we have to modify our graphical representation.  This is because in our graphical representation of the seq-LRC code, nodes represent parity checks and edges represent message symbols, whereas, in the description of a Tornado code, vertices represent code symbols, i.e., either message or parity symbols.    The edges in the graphical representation of a Tornado code, serve only to identify the linear-dependence relation between symbols in a  vertical layer with symbols in the immediately prior layer to the left.    
			
To make the connection with Tornado codes, we therefore provide a revised graphical description of the rate-optimal seq-LRC in terms of a sequence of bipartite graphs $B_i$.  
The description is slightly different for the odd and even cases.  We begin with the case of $t$ even, $t=2s+2$.  
\begin{enumerate}[(i)]
\item The left nodes of the bipartite graph $B_0$ are in $1-1$ correspondence with the edges $E_{s+1}$, the right nodes in $1-1$ correspondence with the edges $E_{s}$ .  We connect a node on the left representing an edge $e_1$ in $E_{s+1}$ with a node on the right corresponding to an edge $e_2$ in $E_{s}$ iff the edges are incident on the same node belonging to vertex set $V_s$\footnote{The restriction to node set $V_s$ etc, causes the resultant graph to be different from the line graph of \ginf.}.
\item Next, the left nodes of the bipartite graph $B_1$ are in $1-1$ correspondence with the edges $E_{s}$, the right nodes in $1-1$ correspondence with the edges $E_{s-1}$ .  We connect a node on the left representing an edge $e_1$ in $E_{s}$ with a node on the right corresponding to an edge $e_2$ in $E_{s-1}$ iff the edges are incident on the same node in $V_{s-1}$.  
\item We continue in this fashion in this way ending up with the rightmost bi-partite graph $B_s$ linking edges $E_1$ on the left with edges $E_0$ on the right iff the edges are incident on the same node in $V_{0}$.  
\end{enumerate} 
This is illustrated in the figure on the left in Fig.~\ref{fig:tornado3}. We note that the left nodes of $B_0$ have degree $2$, the right nodes have degree $r$.  In all the other bipartite graphs, $B_i, 1 \leq i \leq s$, the left nodes have degree $1$, the right nodes continue to have degree $r$.  The left nodes of $B_0$ can be regarded as message symbols. 

In the case $t$ odd, the graph is constructed in a similar (but not identical) fashion:
\begin{enumerate}[(i)]
\item The left nodes of the bipartite graph $B_0$ are in $1-1$ correspondence with the edges $E_{s}$, the right nodes in $1-1$ correspondence with the edges $E_{s-1}$ .  We connect a node on the left representing an edge $e_1$ in $E_{s}$ with a node on the right corresponding to an edge $e_2$ in $E_{s-1}$ iff the edges are incident on the same node in $V_{s-1}$.  
\item Next, the left nodes of the bipartite graph $B_1$ are in $1-1$ correspondence with the edges $E_{s-1}$, the right nodes in $1-1$ correspondence with the edges $E_{s-2}$ .  We connect a node on the left representing an edge $e_1$ in $E_{s-1}$ with a node on the right corresponding to an edge $e_2$ in $E_{s-2}$ iff the edges are incident on the same node in $V_{s-2}$.  
\item We continue in this fashion in this way ending up with the rightmost bi-partite graph $B_{s-1}$ linking edges $E_1$ on the left with edges $E_0$ on the right iff the edges are incident on the same node in $V_{0}$.   
\end{enumerate} 
This is illustrated in the figure on the right in Fig.~\ref{fig:tornado3}.   Here, all the left nodes of the bipartite graphs, $B_i, 0 \leq i \leq s-1$, have degree $1$, the right nodes continue to have degree $r$.  The left nodes of $B_0$, cannot be regarded as message symbols in the odd case however, as they are constrained to satisfy the parity checks imposed by the nodes $V_s$.  Thus we may regard the left nodes of $B_0$ as code symbols of a code ${\cal C}_{\text{LRC}}$ that encodes the actual message symbols of the overall code to produce the code symbols associated with left nodes of $B_0$. 
  \begin{figure}
		 		\begin{center}
		 			\includegraphics[width=\textwidth]{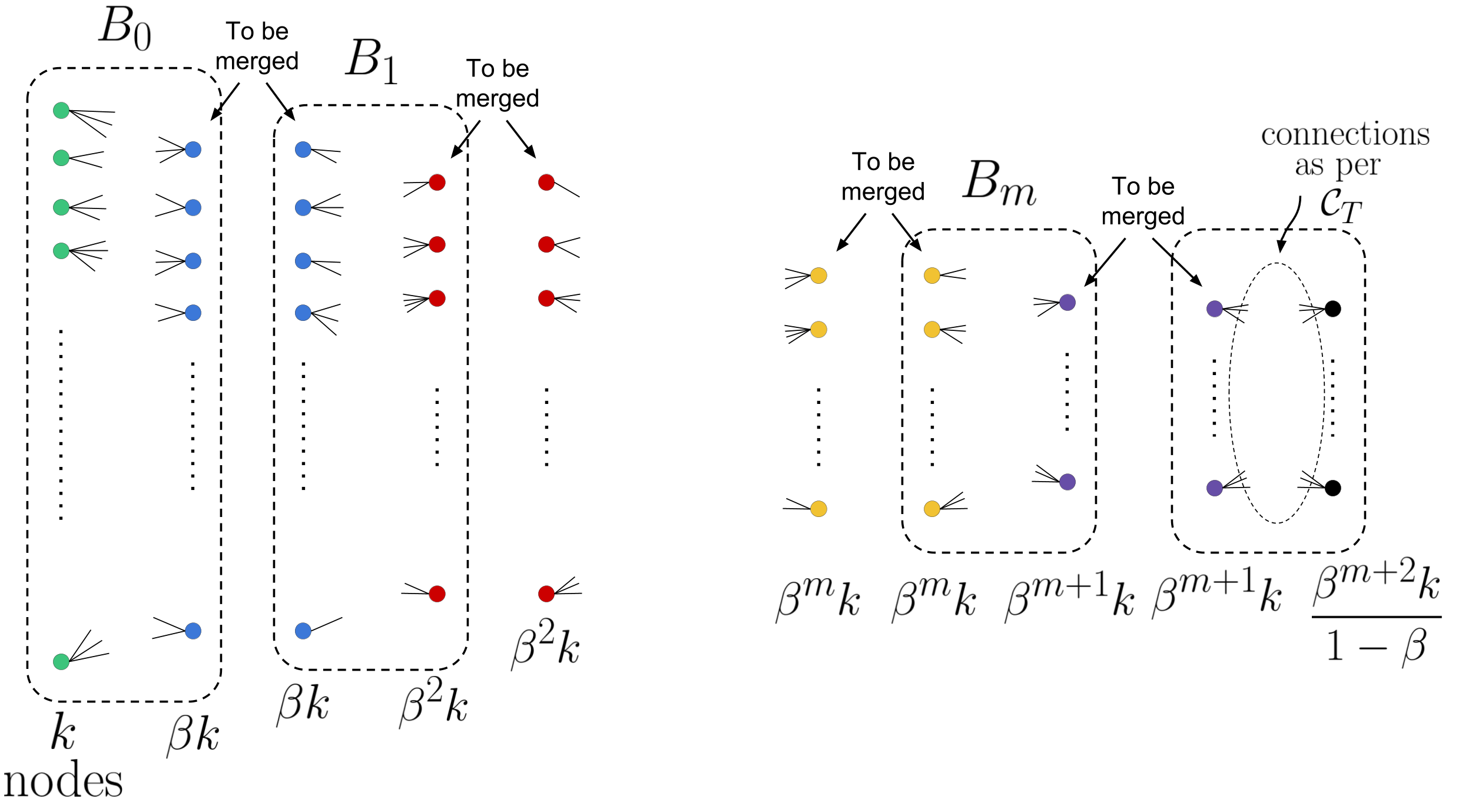}
		 			\caption{The $(m+1)$bipartite graphs $B_0,\ldots, B_m$, appearing in a Tornado code, along with the graph induced by the code ${\cal C}_{\text{T}}$. In the construction of a Tornado code, the ``right" nodes of $B_i$ and the ``left" nodes of $B_{i+1}$, for $0 \leq i \leq m-1$ represent the same set of $\beta^{i+1}k$ parity symbols, and will be merged in the final graph. Such sets of nodes are marked as ``to be merged" in the figure.}
		 			\label{fig:tornado1}
		 		\end{center}
		 	\end{figure}
			\begin{figure}
		 		\begin{center}
		 			\includegraphics[width=\textwidth]{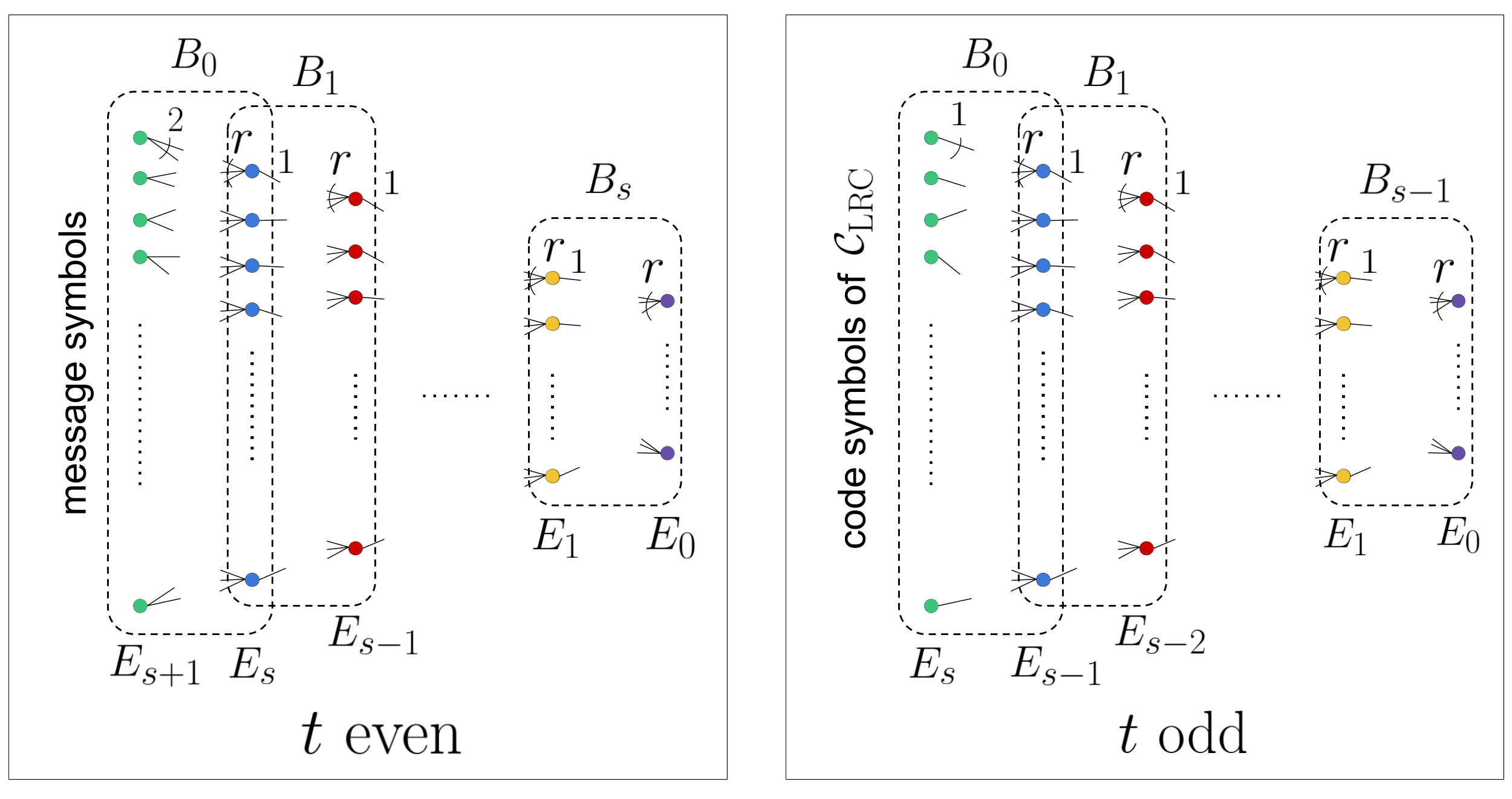}
		 			\caption{The rate-optimal seq-LRC presented as a sequence of bipartite graphs, so as to make the connection with Tornado codes.}
		 			\label{fig:tornado3}
		 		\end{center}
		 	\end{figure}

		 	\subsection{Some Differences Between Tornado Codes and Rate-Optimal seq-LRC}
		 	\ben
		 	\item {\bf Degree Distribution} In Tornado codes, the bipartite graphs $B_i$ in general, are designed to be either having a certain left and right degree distribution or as expander graphs (there could be other ways of designing $B_i$) whereas in the rate-optimal seq-LRC for $t$ even, the left most nodes in $B_0$ corresponding to message symbols have degree $2$ and left nodes of subsequent bipartite graphs has degree $1$ and all the right nodes have degree $r$.
		 	\item {\bf Presence of  an auxiliary code} 
			\bit \item In the case of a Tornado code, the left nodes of the bipartite graph $B_0$ are the message symbols.  The same is true of the bipartite graph $B_0$ in the case of a rate-optimal seq-LRC for $t$ even.  In the case of a rate-optimal seq-LRC for $t$ odd, however, the left most nodes in $B_0$ are code symbols of the code ${\cal C}_{\text{LRC}}$ which may be regarded as a precoder.		 	
			\item The right nodes of the bipartite graph $B_m$ in the case of a Tornado code are message symbols of the auxiliary code ${\cal C}_{\text{T}}$, whereas there is no such code in the case of a seq-LRC.  \eit 
		 	\item {\bf Probabilistic vs. Worst-Case Guarantees} Tornado codes offers probabilistic guarantees with regard to the recovery of $n\delta$ erasures (for some fraction $\delta>0$) and have been shown to achieve capacity on the erasure channel.  Rate-optimal seq-LRC offer deterministic guarantees on the recovery of a fixed number $t$ of erasures and achieve maximum possible rate for a seq-LRC for a given $r,t$.
		 	
		 	\een

	\appendices

	  \section{Proof of Theorem \ref{thm:rate_both}}\label{Appendix_rate_bounds}
	  	
		We begin with an overview of the proof which contains $3$ steps. 
	
	  	\bit
	  	\item Step 1: We first deduce the general form of the p-c matrix of the code.  The form is different for the $t$ even and $t$ odd cases. 
	  	\item Step 2: We obtain linear inequalities by counting in two different ways, along rows and columns, the number of non-zero entries in various sub-matrices of the p-c matrix.  
	  	\item Step 3: We algebraically manipulate these inequalities to derive the upper bound on rate of the code.
	  	\eit
	  	
	  	\subsection{\bf Notation}
		Given an $(m \times n)$ matrix $A$ and subsets $E_1=\{e_1,..,e_p\} \subseteq [m]$, $E_2=\{f_1,...,f_q\} \subseteq [n]$ of size $p$, $q$ respectively, with $e_{\ell} < e_{\ell+1}$, $\forall \ell \in [p-1]$ and $f_{\ell} < f_{\ell+1}$, $\forall \ell \in [q-1]$ we denote by $A|_{E_1,E_2}$, the $(p\times q)$ submatrix of $A$ whose rows are indexed by $E_1$ and columns by $E_2$. 
	  	
  \begin{proof}	  
	  We will handle separately the cases $t$ even and $t$ odd.  
\subsubsection{\bf Case (i) $t$ even}  We begin with the case $t$ even.  	
	  	\paragraph{\bf Step 1: Deducing the form of the p-c matrix}

		  	Recall that the p-c matrix $H$ is of the form $H={[{\underline{c}_1} \  {\underline{c}_2} \ldots {\underline{c}_m}]}^T$ where ${\underline{c}_1,\underline{c}_2,\ldots,\underline{c}_m}$ are $m$ linearly independent codewords with $w_H(\underline{c}_i) \leq r+1$, $\forall i \in [m]$.   It is not hard to see that by permuting rows and columns of $H$, the matrix can be brought into the form:
	  	\bea
	  	H = \left[
	  	\scalemath{1}{
	  		\begin{array}{c|c|c|c|c|c|c|c|c}
	  			D_0 & A_1 & 0 & 0 & \hdots & 0 & 0 & 0 & \\
	  			\cline{1-8}
	  			0 & D_1 & A_2 & 0 & \hdots & 0 & 0 & 0 & \\
	  			\cline{1-8}
	  			0 & 0 & D_2 & A_3 & \hdots & 0 & 0 & 0 & \\
	  			\cline{1-8}
	  			0 & 0 & 0 & D_3 & \hdots & 0 & 0 & 0 &  \\
	  			\cline{1-8}
	  			\vdots & \vdots & \vdots & \vdots & \hdots & \vdots & \vdots & \vdots & D\\
	  			\cline{1-8}
	  			0 & 0 & 0 & 0 & \hdots & A_{\frac{t}{2}-2} & 0 & 0 & \\
	  			\cline{1-8}
	  			0 & 0 & 0 & 0 & \hdots & D_{\frac{t}{2}-2} & A_{\frac{t}{2}-1} & 0 & \\
	  			\cline{1-8}
	  			0 & 0 & 0 & 0 & \hdots & 0 & D_{\frac{t}{2}-1} &  & \\
	  			\cline{1-7}
	  			0 & 0 & 0 & 0 & \hdots & 0 & 0 & C & \\
	  		\end{array} 
	  	}
	  	\right], \ \ \label{Hform}
	  	\eea
	  	%
	  	where the following conditions hold. 
	  	\begin{enumerate}[(i)]
	  		\item Rows are labeled by the integers $1,...,m$  with top most row of $H$ labelled $1$ and next row labeled $2$ and so on. Similarly columns are labeled by the integers $1,...,n$ with leftmost column of $H$ labeled $1$ and next column labeled $2$ and so on, 
	  		\item The matrix $D_0$ is a $(\rho_0 \times a_0)$ matrix for some integers $\rho_0,a_0$. In the matrix $D_0$, each column has weight $1$ and each row has weight at least $1$. The first $a_0$ columns of $H$ contains the columns of $D_0$. The set of first $a_0$ columns of $H$ is equal to the set of all those columns of $H$ which has weight $1$,
			\item $A_i$ is a $(\rho_{i-1} \times a_i)$ matrix for $1 \leq i \leq \frac{t}{2}-1$, $D_i$ is a $(\rho_{i} \times a_i)$ matrix for $0 \leq i \leq \frac{t}{2}-1$ for some integers $\{\rho_j\},\{a_j\}$.  An important point to note here is that we do not insist that all the integers $\{\rho_j\},\{a_j\}$ be non-zero.  This needs some explanation.  We interpret a matrix of size $(u \times v)$ where either $u$ or $v$ or both are zero as simply representing a vacuous or empty matrix, i.e., a matrix that is not present.   We present in the figure below an example to illustrate this point. 
 	\begin{figure}[ht!]
 		\centering
 		\begin{minipage}[c]{0.48\textwidth}
 			\centering
 				\bean
 				H = \left[
 				\begin{array}{c|c|c|c|c|c}
 					D_0 & A_1 & 0 & 0 & 0 & \\
 					\cline{1-5}
 					0 & D_1 & A_2 & 0 & 0 & \\
 					\cline{1-5}
 					0 & 0 & D_2 & A_3 & 0 & D \\
 					\cline{1-5}
 					0 & 0 & 0 & D_3 &  & \\
 				    \cline{1-4}
 				    0 & 0 & 0 & 0 & C & \\
 				\end{array} 
 				\right],  
 				\eean
 		\end{minipage}
 		\hspace{-0.1\textwidth}
 		\begin{minipage}[c]{0.48\textwidth}
 			\bean
 			H = \left[
 			\begin{array}{c|c|c|c|c}
 				D_0 & A_1  & 0 & 0 & \\
 				\cline{1-4}
 				0 & D_1  & 0 & 0 & \\
 				\cline{1-4}
 				0 & 0  & A_3 & 0 & D \\
 				\cline{1-4}
 				0 & 0  & D_3 &  & \\
 				\cline{1-3}
 				0 & 0 & 0 & C & \\
 			\end{array} 
 			\right].
 			\eean
 		\end{minipage}
 		\caption{Giving an example of form of p-c matrix given in \eqref{Hform} when there are empty matrices among $\{A_i\},\{D_i\}$: We illustrate it for the case $t=8$. In the left is the p-c matrix when all of $a_0,a_1,a_2,a_3,\rho_0,\rho_1,\rho_2,\rho_3$ are non-zero i.e., $D_0,A_1,D_1,A_2,D_2,A_3,D_3$ are non-empty matrices. In the right is the p-c matrix when $a_0,a_1,a_3,\rho_0,\rho_1,\rho_2,\rho_3$ are non-zero but $a_2=0$ i.e., $D_0,A_1,D_1,A_3,D_3$ are non-empty matrices but $A_2,D_2$ are empty matrices. In both left and right, we assume both $C,D$ are non-empty matrices.}
 		\label{fig:Genformpcmatrix}
 	\end{figure}
 	
	  		\item  For $1\leq i \leq \frac{t}{2}-1$, when the matrices $A_i$,$D_i$, are not empty, then they are such that, each column of the concatenated matrix $B_i \triangleq \left[\frac{A_i}{D_i}\right]$ has weight 2, each column of $A_i$ has weight at least $1$ and each row of $D_i$ has weight at least $1$ and each column of $D_i$ has weight at most $1$.  
	  		\item The matrix $C$ is a $((\rho_{\frac{t}{2}-1}+p) \times a_{\frac{t}{2}})$ matrix for some integers $p,a_{\frac{t}{2}}$. Here also we do not insist that both $p,a_{\frac{t}{2}}$ are non-zero. When the matrix $C$ is not empty i.e., when $a_{\frac{t}{2}} \neq 0$, it is a matrix with each column having weight $2$. The set of columns of the matrix $D$ is precisely the set of all columns of $H$ having weight $\geq 3$.
	  	\end{enumerate}
	  		  	\paragraph{\bf Handling the case when any of matrices $A_i,D_i,C$ is empty} We now consider cases when any of the above matrices $\{A_j\},\{D_j\},C$ is empty.  
	  		  	Define:
	  		  	\bean
	  		    \ell_1 &=& \min{\{\{j: 1 \leq j \leq \frac{t}{2}-1, A_j \text{ is an empty matrix}, D_j \text{ is not an empty matrix}\} \cup \{\frac{t}{2} \} \}}, \\
	  		  	\ell_2 &=& \min{\{\{j: 1 \leq j \leq \frac{t}{2}-1, A_j,D_j \text{ are empty matrices}\} \cup \{\frac{t}{2} \} \}}, \\
	  		  	\ell_3 &=& \min{\{\{j: 1 \leq j \leq \frac{t}{2}-1, A_j \text{ is not an empty matrix}, D_j \text{ is an empty matrix}\} \cup \{\frac{t}{2} \} \}}.
	  		  	\eean
	  		     We will now show that the case $A_j$ is an empty matrix and $D_j$ is not an empty matrix or the case $\ell_3 < \ell_2$ cannot occur. If the case $A_j$ is an empty matrix and $D_j$ is an empty matrix occur we will redefine our matrix $H$ by redefining the matrix $C$.
	  		  	\bit
	  		  	\item If $A_j$ is an empty matrix for some $j$ then $D_j$ is also an empty matrix because each column of $[\frac{A_j}{D_j}]$ has weight $2$ and $D_j$ has the constraint that each of its column has weight at most $1$. Hence the case $A_j$ an empty matrix and $D_j$ a non-empty matrix cannot occur for any $j$.
	  		  	\item If $A_j$ is not an empty matrix and $D_j$ is an empty matrix for some $j$ then $\rho_j=0$ and hence $A_{j+1}$ is also an empty matrix and by previous point, $D_{j+1}$ is also an empty matrix. Hence $\ell_2 \leq \ell_3+1$ (This inequality holds true obviously when $\ell_3=\frac{t}{2}-1$).
	  		  	\item We now handle the case when $A_j$ is not an empty matrix and $D_j$ is an empty matrix for some $j$ and $\ell_2 > \ell_3$ i.e., $\ell_2=\ell_3+1$. As $\ell_2 > \ell_3$, the proof of the Lemma \ref{claim1} mentioned following this discussion (since the proof of Lemma \ref{claim1} proceeds by induction starting with the proof of the lemma for $D_0$ first and then proceeding to $A_1$ and then $D_1$ and so on and hence we prove the lemma for $A_{\ell_3}$ first and then proceed to $D_{\ell_3}$ and since $D_0,A_i,D_i$ must be non-empty matrices $\forall 1 \leq i \leq \ell_3-1$) will imply that each column of $A_{\ell_3}$ has weight 1 which will imply that $D_{\ell_3}$ cannot be an empty matrix as each column of $[\frac{A_{\ell_3}}{D_{\ell_3}}]$ has weight $2$. Hence the case $A_{\ell_3}$, a non-empty matrix and $D_{\ell_3}$, an empty matrix cannot occur with $\ell_2 > \ell_3$.
	  		   \item We now handle the case when $A_j,D_j$ are empty matrices for some $j$ and $\ell_2 < \ell_3$ or the case when $D_0$ is an empty matrix. Note that we can assume $\ell_2 < \ell_3$ as we have already proved that $\ell_2 > \ell_3$ cannot occur. We set $\ell_2=0$, if $D_0$ is an empty matrix. We now redefine $H$ based on the value of $\ell_2$.
	  	We set $A_i$, $D_i$ to be empty matrices and set $a_i=0, \rho_i=0$, $\forall \ell_2 \leq i \leq \frac{t}{2}-1$. 
	  	Let $E_2  \subseteq \{\sum_{i=0}^{\ell_2-1}a_i+1,\ldots,n-1,n\}$ such that $E_2$ is the set of labels of all the 2-weight columns of $H$ apart from those 2-weight columns of $H$ containing the columns of $B_1,B_2,\ldots,B_{\ell_2-1}$. By a $2$-weight column, we refer to a column of weight $2$. Let $E_1 = \{\sum_{i=0}^{\ell_2-1}{\rho_i}+1,\ldots,m-1,m\}$. If $\ell_2=\frac{t}{2}$ then $H$ is defined by \eqref{Hform} with none of the matrices among $\{A_i\},\{D_i\}$ being empty. If $\ell_2 < \frac{t}{2}$, we redefine $C=H|_{E_1,E_2}$.   If $\ell_2 < \frac{t}{2}$, the matrix $H$ can be written in the form given in \eqref{Hform_111} and hence defined by \eqref{Hform_111}.
	  	\bea
	  	H = \left[
	  	\scalemath{1}{
	  		\begin{array}{c|c|c|c|c|c|c|c|c}
	  			D_0 & A_1 & 0 & 0 & \hdots & 0 & 0 & 0 & \\
	  			\cline{1-8}
	  			0 & D_1 & A_2 & 0 & \hdots & 0 & 0 & 0 & \\
	  			\cline{1-8}
	  			0 & 0 & D_2 & A_3 & \hdots & 0 & 0 & 0 & \\
	  			\cline{1-8}
	  			0 & 0 & 0 & D_3 & \hdots & 0 & 0 & 0 &  \\
	  			\cline{1-8}
	  			\vdots & \vdots & \vdots & \vdots & \hdots & \vdots & \vdots & \vdots & D\\
	  			\cline{1-8}
	  			0 & 0 & 0 & 0 & \hdots & A_{\ell_2-2} & 0 & 0 & \\
	  			\cline{1-8}
	  			0 & 0 & 0 & 0 & \hdots & D_{\ell_2-2} & A_{\ell_2-1} & 0 & \\
	  			\cline{1-8}
	  			0 & 0 & 0 & 0 & \hdots & 0 & D_{\ell_2-1} & 0 & \\
	  			\cline{1-8}
	  			0 & 0 & 0 & 0 & \hdots & 0 & 0 & C & \\
	  		\end{array} 
	  	}
	  	\right]. \ \ \label{Hform_111}
	  	\eea

	  	Note that none of the matrices $D_0,A_1,D_1,A_2,D_2,\ldots,A_{\ell_2-1},D_{\ell_2-1}$ are empty matrices. Irrespective of the value of $\ell_2$, let the number of columns in $C$ be denoted as $a_{\frac{t}{2}}$ and the number of rows of $C$ be denoted as $\rho_{\frac{t}{2}-1}+p$. Note that $\rho_{\frac{t}{2}-1}=0$, if $D_{\frac{t}{2}-1}$ is an empty matrix because if $A_{\frac{t}{2}-1}$ is also an empty matrix then clearly we have set $\rho_{\frac{t}{2}-1}=0$ and the case $A_{\frac{t}{2}-1}$ is not an empty matrix and $D_{\frac{t}{2}-1}$ is an empty matrix cannot occur at this point. If $C$ is an empty matrix then we can clearly set $a_{\frac{t}{2}}=0$.
	  	The entire derivation of upper bound on rate is correct and all the inequalities in the derivation will hold with $a_i=0$, $\rho_i=0$, $\forall \ell_2 \leq i \leq \frac{t}{2}-1$ and $a_{\frac{t}{2}}=0$ (if $C$ is an empty matrix).
	  	\eit
	  	Although we have to prove the following Lemma \ref{claim1} for $A_i,D_i$, $\forall 1 \leq i \leq \ell_2-1$, $D_0$, we assume all $D_0,$ $A_i,D_i$, $\forall 1 \leq i \leq \frac{t}{2}-1$ to be non-empty matrices and prove the lemma. Since the proof of the lemma is by induction, the induction can be made to stop after proving the lemma for $A_{\ell_2-1},D_{\ell_2-1}$ (The proof of the lemma is by induction and induction starts by proving the lemma for $D_0$ and proceeds to $A_1$ and to $D_1$ and so on as mentioned before) and the proof is unaffected by it.
	  		  	\paragraph{\bf Deducing the structure of p-c matrix further}
	  	\begin{lem}\label{claim1}
	  		For $1 \leq i \leq \frac{t}{2}-1$, $A_i$ is a matrix with each column having weight 1.
	  		For $0 \leq j \leq \frac{t}{2}-1$, $D_j$ is a matrix with each row having weight 1 and each column having weight 1. 
	  	\end{lem}
	  	%
	  	%
	  	%
	  	\begin{proof}
	  		Let $d_{\min}$ be the minimum distance of the code with parity check matrix $H$. We use the fact that $d_{\min} \geq t+1$ to prove the lemma i.e., we show that there is a set of $\leq t$ linearly dependent columns if the lemma fails to hold.\\
	  		It is enough to show that:
	  		
	  		\begin{itemize}
	  			\item For $1 \leq i \leq \frac{t}{2}-1$, $A_i$ is a matrix with each column having  weight 1.
	  			\item For $0 \leq i \leq \frac{t}{2}-1$, $D_i$ is a matrix with each row having weight 1. 
	  		\end{itemize}
	  		This is because the property that $A_i$ is a matrix with each column having weight 1 combined with the fact that each column of $B_i=\left[\frac{A_i}{D_i}\right]$ has weight 2  implies that $D_i$ is a matrix with each column having weight 1 and $D_0$ by definition is a matrix with each column having weight 1. \\
	  		Let us denote the column of the matrix $H$ with label $j$ by $\underline{h}_j$, $1 \leq j\leq n$. 
	  		
	  		Before proceeding to the proof, we want to note that the rows of $H$ with labels $\{\sum_{\ell=0}^{i-1}\rho_{\ell}+1,\ldots,\sum_{\ell=0}^{i-1}\rho_{\ell}+\rho_i\}$ correspond to the rows of $H$ containing $D_{i},A_{i+1}$ and the columns of $H$ with labels $\{\sum_{\ell=0}^{i-1}a_{\ell}+1,\ldots,\sum_{\ell=0}^{i-1}a_{\ell}+a_i\}$ correspond to the columns of $H$ containing $A_i,D_i$.  Also throughout the proof, support or support set of a vector refers to the set of indices corresponding to non-zero co-ordinates of the vector. For example the vector $[1, 0, 2, 0]$ (or $[1, 0, 2, 0]^T$) has support set equal to $\{1,3\}$ as it has non-zero first and third co-ordinate. A non-trivial linear combination of a set of vectors $\{\underline{v}_1,...,\underline{v}_j\}$ refers to a linear combination of the set of vectors $\sum a_i \underline{v}_i$ with at least one of the coefficients $a_1,a_2,\ldots,a_j$ taking a non-zero value. All linear combinations referred to in the proof are non-trivial. We use these facts throughout the proof. \\

	  		Let us prove the lemma by induction as follows:
	  		
	  		{\em Induction Hypothesis:} 
	  		\begin{itemize}
	  			\item We induct on a variable denoted by $i$.
	  			\item  Property $P_i$: any non-zero $m \times 1$ vector having weight at most 2 with support contained in $\{\sum_{\ell=0}^{i-1}\rho_{\ell}+1,\ldots,\sum_{\ell=0}^{i-1}\rho_{\ell}+\rho_i\}$  can be written as some linear combination of vectors  $\underline{h}_{p_1},\ldots,\underline{h}_{p_\psi}$ for some $\{p_1,p_2,\ldots,p_\psi\} \subseteq \{1,2,\ldots,\sum_{\ell=0}^{i}a_{\ell}\}$ and for some $0 < \psi \leq 2(i+1)$.
	  			\item {\em Let us assume as induction hypothesis that the property $P_i$ is true and the Lemma \ref{claim1} is true for $A_1,\ldots,A_i$, $D_0,\ldots,D_i$.}
	  		\end{itemize}
	  		{\em Initial step of induction corresponding to $i=0$ and $i=1$:}
	  		\begin{itemize}	
	  			\item {\em We now show that each row of $D_0$ has weight exactly $1$.}\\
	  			Suppose there exists a row in $D_0$ with weight more than $1$; let the support set of the row be $\{i_1,i_2,...\}$ in $D_0$. Then the columns $\underline{h}_{i_1},\underline{h}_{i_2}$ of $H$ can be linearly combined to give a zero column and hence $d_{\min} \leq 2$. This contradicts the fact that $d_{\min} \geq t+1, t \geq 2$ as $t$ is even.
	  			Hence, all rows of $D_0$ have weight exactly $1$.
	  			\item If $t = 2$, then the lemma is already proved. So let $t \geq 4$.		
	  			\item {\em We now show that each column of $A_1$ has weight exactly $1$.}\\
	  			Suppose $j^{th}$ column of $A_1$ $\text{for some } 1 \leq j \leq a_1$ has weight 2; let the support of the column be $\{j_1,j_2\}$ in $A_1$. Then the column $\underline{h}_{a_0 + j}$ in $H$ along with the 2 column vectors of $H$ say $\underline{h}_{p_1},\underline{h}_{p_2}$ where $\{p_1,p_2\} \subseteq \{1,\ldots,a_0\}$ where $\underline{h}_{p_1}$ has exactly only one non-zero element in  $j_1^{th}$ co-ordinate and $\underline{h}_{p_2}$ has exactly only one non-zero element $j_2^{th}$ co-ordinate, can be linearly combined to give a zero column again leading to a contradiction on minimum distance. Such columns with one column having only one non-zero element exactly in $j_1^{th}$ co-ordinate and another column having only one non-zero element exactly in $j_2^{th}$ co-ordinate with column labels in $\{1,\ldots,a_0\}$ exist due to the $1$-weight columns in the matrix $D_0$.
	  			\item The above argument also shows that any non-zero $m \times 1$ vector having weight at most 2 with support contained in $\{1,\ldots,\rho_0\}$ can be written as some linear combination of at most 2 column vectors of $H$ say $\underline{h}_{p_1},\ldots,\underline{h}_{p_\psi}$ for some $\{p_1,\ldots,p_\psi\} \subseteq \{1,\ldots,a_0\}$ ($\psi = 1 \text{ or }2$). Hence Property $P_0$ is true.		
	  			\item { \em We now show that each row of $D_1$ has weight exactly $1$.}
	  			Suppose $j^{th}$ row of $D_1$ for some $1 \leq j \leq \rho_1$ has weight more than 1; let the support set of the row be $\{i_1,i_2,\ldots\}$ in $D_1$. Now there is some linear combination of the columns $\underline{h}_{a_0+i_1}$ and $\underline{h}_{a_0+i_2}$ in $H$ that gives a zero in $(\rho_0+j)^{th}$ coordinate and thus this linear combination has support contained in $\{1,\ldots,\rho_0\}$ with weight at most $2$. Now applying Property $P_0$ on this linear combination implies that there is a non-empty set of at most $4$ linearly dependent columns in $H$ leading to a contradiction on minimum distance.		
	  			\item {\em Now we show that Property $P_{1}$ is true.} We have to prove that any non-zero $m \times 1$ vector with weight at most $2$ with support contained in $\{\rho_0+1,\ldots,\rho_0+\rho_1\}$ can be written as linear combination of at most $2(1+1)=4$ column vectors of $H$ say $\underline{h}_{p_1},\ldots,\underline{h}_{p_\psi}$ for some $\{p_1,\ldots,p_\psi\} \subseteq \{1,2,\ldots,\sum_{\ell=0}^{1}a_{\ell}\}$ and $0 < \psi \leq 4$. This can be easily seen using arguments similar to ones presented before. Let $1 \leq j_1,j_2 \leq \rho_{1}$.
	  			Let $\underline{v}$ be an $m \times 1$ vector having non-zero elements exactly in coordinates $\rho_0+j_1,\rho_0+j_2$ or $\rho_0+j_1$. Take $2$ columns $\underline{h}_{y_1},\underline{h}_{y_2}$ having a non-zero element in co-ordinates $\rho_0+j_1,\rho_0+j_2$ respectively with $\{y_1,y_2\} \subseteq \{ a_0+1,\ldots,a_0+a_{1} \}$. These columns $\underline{h}_{y_1},\underline{h}_{y_2}$ exist due to $D_1$. Note that $D_1$ is a matrix with each row and column  having weight exactly $1$. Then the vector $\underline{v}$ can be linearly combined with $\underline{h}_{y_1},\underline{h}_{y_2}$ where $\{y_1,y_2\} \subseteq \{ a_0+1,a_0+2,\ldots,a_0+a_{1} \}$ to form a $m \times 1$ vector with weight at most $2$ with support contained in $\{ 1,2,\ldots,\rho_0 \}$ which in turn can be written as linear combination of at most $2$ column vectors in $H$ say $\underline{h}_{z_1},\ldots,\underline{h}_{z_\theta}$ for some $\{z_1,\ldots,z_\theta\} \subseteq \{1,\ldots,a_0\}$ ($\theta = 1$ or $2$) by property $P_0$. Hence the given $m \times 1$ vector $\underline{v}$ is written as linear combination of at most $2(1+1)=4$ column vectors in $H$: $\underline{h}_{y_1},\underline{h}_{y_2},\underline{h}_{z_1}...,\underline{h}_{z_\theta}$ for some $\{y_1,y_2,z_1,\ldots,z_{\theta}\} \subseteq \{1,2,\ldots,\sum_{\ell=0}^{1}a_{\ell}\}$. Hence property $P_1$ is true.
                
	  		\end{itemize}
	  		{\em Induction step :}
	  		\begin{itemize}
	  			\item For $t=4$, the initial step of induction with $i \in \{0,1\}$ completes the proof of the Lemma \ref{claim1}. Hence we assume $t \geq 6$. Let us assume by induction hypothesis that Property $P_i$ is true and the Lemma \ref{claim1} is true for $A_1,\ldots,A_i$, $D_0,\ldots,D_i$ for some $i \leq \frac{t}{2}-2$ and prove the induction hypothesis for $i+1$. 
	  			\item {\em Now we show that each column of $A_{i+1}$ has weight exactly $1$}. Suppose $j^{th}$ column of $A_{i+1}$ for some $1 \leq j \leq a_{i+1}$ has weight $2$; let the support of the column be $\{j_1,j_2\}$ in $A_{i+1}$. It is clear that the corresponding column vector  $\underline{h}_{\sum\limits_{\ell=0}^{i}a_{\ell}+j}$ in $H$ is a vector with support  $\{\sum_{\ell=0}^{i-1}\rho_{\ell}+j_1,\sum_{\ell=0}^{i-1}\rho_{\ell}+j_2\} \subseteq \{\sum_{\ell=0}^{i-1}\rho_{\ell}+1,\ldots,\sum_{\ell=0}^{i-1}\rho_{\ell}+\rho_i\}$ and hence has weight $2$. Now applying Property $P_i$ on this column vector $\underline{h}_{\sum\limits_{\ell=0}^{i}a_{\ell}+j}$ implies that there is a non-empty set of at most $2(i+1)+1$ columns in $H$ which are linearly dependent; hence contradicts the minimum distance of the code as $2(i+1)+1 \leq t-1$. Hence each column of $A_{i+1}$ has weight exactly $1$.
	  			\item { \em Now we show that each row of $D_{i+1}$ has weight exactly $1$.} Suppose $j^{th}$ row of $D_{i+1}$ for some $1 \leq j \leq \rho_{i+1}$ has weight more than $1$; let the support set of the row be $\{\ell_1,\ell_2,\ldots\}$ in $D_{i+1}$. Now some linear combination of the columns $\underline{h}_{\sum\limits_{j=0}^{i}a_j+ \ell_1}$ and $\underline{h}_{\sum\limits_{j=0}^{i}a_j+ \ell_2}$ in $H$ will have a $0$ in $(\sum\limits_{\ell=0}^{i}\rho_{\ell}+j)^{th}$ coordinate and this linear combination also has weight at most 2 with support contained in $\{\sum_{\ell=0}^{i-1}\rho_{\ell}+1,\ldots,\sum_{\ell=0}^{i-1}\rho_{\ell}+\rho_i\}$ and hence applying Property $P_i$ on this linear combination implies that there is a non-empty set of at most $2(i+1)+2$ columns in $H$ which are linearly dependent; hence contradicts the minimum distance as $2(i+1)+2 \leq t$; thus proving that each row of $D_{i+1}$ has weight exactly 1.
	  			\item {\em Now we show that Property $P_{i+1}$ is true.} We have to prove that any non-zero $m \times 1$ vector with weight at most $2$ with support contained in $\{\sum_{\ell=0}^{i}\rho_{\ell}+1,\ldots,\sum_{\ell=0}^{i}\rho_{\ell}+\rho_{i+1}\}$ can be written as linear combination of at most $2(i+2)$ column vectors of $H$ say $\underline{h}_{p_1},\ldots,\underline{h}_{p_\psi}$ for some $\{p_1,\ldots,p_\psi\} \subseteq \{1,\ldots,\sum_{\ell=0}^{i+1}a_{\ell}\}$ and $0 < \psi \leq 2(i+2)$. This can be easily seen using arguments similar to ones presented before.
	  			
	  			Let $1 \leq j_1,j_2 \leq \rho_{i+1}$. Let $\underline{v}$ be an $m \times 1$ vector having non-zero elements exactly in coordinates $\sum_{\ell=0}^{i}\rho_{\ell}+j_1,\sum_{\ell=0}^{i}\rho_{\ell}+j_2$ or $\sum_{\ell=0}^{i}\rho_{\ell}+j_1$. Take $2$ columns $\underline{h}_{y_1},\underline{h}_{y_2}$ having a non-zero element in co-ordinates $\sum_{\ell=0}^{i}\rho_{\ell}+j_1,\sum_{\ell=0}^{i}\rho_{\ell}+j_2$ respectively with $\{y_1,y_2\} \subseteq \{ \sum_{\ell=0}^{i}a_{\ell}+1,\ldots,\sum_{\ell=0}^{i}a_{\ell}+a_{i+1} \}$. These columns $\underline{h}_{y_1},\underline{h}_{y_2}$ exist due to $D_{i+1}$. Note that $D_{i+1}$ is a matrix with each row and column  having weight exactly $1$. Then the vector $\underline{v}$ can be linearly combined with $\underline{h}_{y_1},\underline{h}_{y_2}$ where $\{y_1,y_2\} \subseteq \{ \sum_{\ell=0}^{i}a_{\ell}+1,\ldots,\sum_{\ell=0}^{i}a_{\ell}+a_{i+1} \}$ to form a $m \times 1$ vector with weight at most $2$ with support contained in $\{ \sum_{\ell=0}^{i-1}\rho_{\ell}+1,\sum_{\ell=0}^{i-1}\rho_{\ell} +\rho_i \}$ which in turn can be written as linear combination of at most $2(i+1)$ column vectors in $H$ say $\underline{h}_{z_1},\ldots,\underline{h}_{z_\theta}$ for some $\{z_1,\ldots,z_\theta\} \subseteq \{1,\ldots,\sum_{\ell=0}^{i}a_{\ell}\}$ and $0 < \theta \leq 2(i+1)$ by property $P_i$. Hence the given $m \times 1$ vector $\underline{v}$ is written as linear combination of at most $2(i+2)$ column vectors in $H$: $\underline{h}_{y_1},\underline{h}_{y_2},\underline{h}_{z_1},...,\underline{h}_{z_\theta}$ for some $\{y_1,y_2,z_1,\ldots,z_{\theta}\} \subseteq \{1,2,\ldots,\sum_{\ell=0}^{i+1}a_{\ell}\}$. Hence property $P_{i+1}$ is true.
	  		\end{itemize}
	  	\end{proof}
	  		  	\paragraph{\bf Step 2: Forming linear inequalities in variables $n,m,\{a_i\},\{\rho_i\}$}
	  	We now form linear inequalities in variables $n,m,\{a_i\},\{\rho_i\}$ based on counting the non-zero entries row-wise and column-wise in various matrices $\{A_i\},\{D_i\},C,H$ using the lemma \ref{claim1} in the process.
	  	By Lemma \ref{claim1}, after permutation of columns of $H$ (in \eqref{Hform} or \eqref{Hform_111} depending on the value of $\ell_2$) within the columns labeled by the set $\{\sum_{\ell=0}^{j-1}a_{\ell}+1,...\sum_{\ell=0}^{j-1}a_{\ell}+a_{j}\}$ for $0 \leq j \leq \ell_2-1$, the matrix $D_j,0 \leq j \leq \ell_2-1$ can be assumed to be a diagonal matrix with non-zero entries along the diagonal and hence $\rho_{i}=a_{i}$, $\forall 0 \leq i \leq \frac{t}{2}-1$ because $\rho_j=a_j=0$ for $j > \ell_2-1$. 
	  	
	  	Since the sum of the column weights of $A_i, 1 \leq i \leq \frac{t}{2}-1$ must equal the sum of the row weights and since each row of $A_i$ for $i \leq \ell_2-1$ can have weight atmost $r$ and not $r+1$ due to $1$-weight rows in $D_{i-1}$, and since for $\ell_2 \leq i \leq \frac{t}{2}-1$, $A_i$ is an empty matrix and we have set $a_i=0$, we obtain:
	  	\bea
	  	\text{For } 1 \leq i \leq \frac{t}{2}-1 : & & \notag \\
	  	\rho_{i-1} r & \geq & a_i, \notag \\
	  	a_{i-1} r & \geq & a_i. \label{Ineq1}
	  	\eea
	  	We also have that,
	  	\bea
	  	\sum_{i=0}^{\frac{t}{2}-1} \rho_i + p = \sum_{i=0}^{\frac{t}{2}-1} a_i + p =m.  \label{Ineq2}
	  	\eea 
	  	By equating sum of row weights of $C$, with sum of column weights of $C$, we obtain:
	  	\bea
	  	2 a_{\frac{t}{2}} & \leq & (\rho_{\frac{t}{2}-1}+p) (r+1)-\rho_{\frac{t}{2}-1} \notag \\
	  	2 a_{\frac{t}{2}} & \leq & (a_{\frac{t}{2}-1}+p) (r+1)-a_{\frac{t}{2}-1}. \label{Ineq3}	    
	  	\eea
	  	Note that if $C$ is an empty matrix then also the inequality \eqref{Ineq3} is true as we would have set $a_{\frac{t}{2}}=0$. If $\ell_2 < \frac{t}{2}$ and $C$ a non-empty matrix then the number of rows in $C$ is $p$ with each column of $C$ having weight 2, hence the inequality \eqref{Ineq3} is still true.\\
	  	Substituting \eqref{Ineq2} in \eqref{Ineq3} we get:
	  	\bea
	  	2 a_{\frac{t}{2}} & \leq & (m-\sum_{i=0}^{\frac{t}{2}-2} a_i) (r+1)-(m-\sum_{i=0}^{\frac{t}{2}-2} a_i-p), \notag \\
	  	2 a_{\frac{t}{2}} & \leq & (m-\sum_{i=0}^{\frac{t}{2}-2} a_i) r+p. \label{Ineq4}
	  	\eea
	  	
	  	By equating sum of row weights of $H$, with sum of column weights of $H$, we obtain:
	  	\bea
	  	m(r+1) & \geq & a_0 + 2(\sum_{i=1}^{\frac{t}{2}} a_i) + 3 (n-\sum_{i=0}^{\frac{t}{2}} a_i), \label{Ineq_star1} 
	  	\eea
	  	If $\ell_2 < \frac{t}{2}$ then $a_i=0$, $\forall \ell_2 \leq i \leq \frac{t}{2}-1$. If $C$ is an empty matrix then $a_{\frac{t}{2}}=0$. Hence the inequality \eqref{Ineq_star1} is true irrespective of whether $\ell_2=\frac{t}{2}$ or $\ell_2 < \frac{t}{2}$ (even if $C$ is an empty matrix).
	  	\bea
	  	\text{From \eqref{Ineq_star1} :} & &  \notag \\
	  	m(r+1) & \geq & 3n-2a_0 -(\sum_{i=1}^{\frac{t}{2}} a_i). \label{Ineq5} 
	  	\eea
	  	
	  	Our basic inequalities are \eqref{Ineq1},\eqref{Ineq2},\eqref{Ineq3},\eqref{Ineq_star1}. We manipulate these 4 inequalities to derive the bound on rate.
	  	\paragraph{\bf Step 3: Algebraic manipulation of the 4 linear inequalities to derive the bound on rate}

	  	Substituting \eqref{Ineq2} in \eqref{Ineq5} we get:
	  	\bea
	  	m(r+1) & \geq & 3n-a_0 -a_{\frac{t}{2}}-(m-p), \notag \\
	  	m(r+2) & \geq & 3n+p-a_0 -a_{\frac{t}{2}}. \label{Ineq6} 
	  	\eea
	  	Substituting \eqref{Ineq4} in \eqref{Ineq6}, we get:
	  	\bea
	  	m(r+2) & \geq & 3n+p-a_0 - \left( \frac{(m-\sum_{i=0}^{\frac{t}{2}-2} a_i) r+p}{2} \right), \notag \\
	  	m(r+2+\frac{r}{2}) & \geq & 3n+\frac{p}{2}-a_0 + \left(\sum_{i=0}^{\frac{t}{2}-2} a_i \right) \frac{r}{2}. \label{Ineq7}
	  	\eea
	  	
	  	From \eqref{Ineq2}, for any $0 \leq j \leq \frac{t}{2}-2$:
	  	\bea
	  	a_{\frac{t}{2}-2-j} = m-\sum_{i=0}^{\frac{t}{2}-2-j-1} a_i - \sum_{i=\frac{t}{2}-2-j+1}^{\frac{t}{2}-1} a_i - p. \label{Ineq10}
	  	\eea
	  	Subtituting \eqref{Ineq1} for $\frac{t}{2}-2-j+1 \leq i \leq \frac{t}{2}-1$ in \eqref{Ineq10}, we get: 
	  	\bea
	  	a_{\frac{t}{2}-2-j} & \geq & m-\sum_{i=0}^{\frac{t}{2}-2-j-1} a_i - \sum_{i=\frac{t}{2}-2-j+1}^{\frac{t}{2}-1} a_{\frac{t}{2}-2-j} r^{i-(\frac{t}{2}-2-j)} - p, \notag
	  	\eea
	  	\bea
	  	a_{\frac{t}{2}-2-j} & \geq & m-\sum_{i=0}^{\frac{t}{2}-2-j-1} a_i - \sum_{i=1}^{j+1} a_{\frac{t}{2}-2-j} r^{i} - p, \notag
	  	\eea
	  	\bea
	  	a_{\frac{t}{2}-2-j} & \geq & \frac{m-\sum_{i=0}^{\frac{t}{2}-2-j-1} a_i-p}{ 1+ \sum_{i=1}^{j+1} r^{i} }. \label{Ineq11}
	  	\eea
	  	Let,
	  	\bea
	  	\delta_0 & = & \frac{r}{2}, \label{Ineq12} \\
	  	\text{For $j_1 \geq 0$: }\delta_{j_1+1} & = & \delta_{j_1} - \frac{\delta_{j_1}}{1+\sum_{i=1}^{j_1+1}r^i}. \label{Ineq13}
	  	\eea
	  	Let us prove the following inequality by induction for $0 \leq J_1 \leq \frac{t}{2}-2$,
	  	\bea
	  	m(r+2+\delta_{J_1}) & \geq & 3n + p\left( \frac{1}{2}+\delta_{J_1} - \frac{r}{2}\right)-a_0  + \left(\sum_{i=0}^{\frac{t}{2}-2-J_1} a_i \right) \delta_{J_1}. \label{Ineq9}
	  	\eea
	  	\eqref{Ineq9} is true for $J_1=0$ by \eqref{Ineq7}. Hence \eqref{Ineq9} is proved for $t=4$ and the range of $J_1$ is vacuous for $t=2$. Hence assume $t>4$. Hence let us assume \eqref{Ineq9} is true for $ J_1$ such that $ \frac{t}{2}-3 \geq J_1 \geq 0$ and prove it for $J_1+1$. 
	  	Substituting \eqref{Ineq11} for $j=J_1$ in \eqref{Ineq9}, we get:
	  	\bea
	  	m(r+2+\delta_{J_1}) & \geq & 3n + p\left( \frac{1}{2}+\delta_{J_1} - \frac{r}{2}\right)-a_0 \\
	  	& +& \left(\sum_{i=0}^{\frac{t}{2}-2-J_1-1} a_i \right) \delta_{J_1} + \left( \frac{m-\sum_{i=0}^{\frac{t}{2}-2-J_1-1} a_i-p}{ 1+ \sum_{i=1}^{J_1+1} r^{i} } \right) \delta_{J_1}, \notag \\
	  	m \left (r+2+\delta_{J_1}-\frac{\delta_{J_1}}{1+ \sum_{i=1}^{J_1+1} r^{i} } \right) & \geq & 
	  	3n + p\left( \frac{1}{2}+\delta_{J_1} -\frac{\delta_{J_1}}{1+ \sum_{i=1}^{J_1+1} r^{i} }- \frac{r}{2}\right)-a_0  + \notag \\ 
	  	& & \left(\sum_{i=0}^{\frac{t}{2}-2-J_1-1} a_i \right) \left(\delta_{J_1} - \frac{\delta_{J_1}}{1+ \sum_{i=1}^{J_1+1} r^{i} }\right). \label{Ineq141}
	  	\eea
	  	Substituing \eqref{Ineq13} in \eqref{Ineq141}, we obtain
	  	\bea
	  	m(r+2+\delta_{J_1+1}) & \geq & 
	  	3n + p\left( \frac{1}{2}+\delta_{J_1+1}- \frac{r}{2}\right)-a_0  + \left(\sum_{i=0}^{\frac{t}{2}-2-J_1-1} a_i \right) \delta_{J_1+1}. \label{Ineq14}
	  	\eea
	  	Hence \eqref{Ineq9} is proved  for any $0 \leq J_1 \leq \frac{t}{2}-2$ for $t \geq 4$. Hence writing \eqref{Ineq9} for $J_1 = \frac{t}{2}-2$ for $t\geq 4$, we obtain:
	  	\bea
	  	m(r+2+\delta_{\frac{t}{2}-2}) & \geq & 3n + p\left( \frac{1}{2}+\delta_{\frac{t}{2}-2} - \frac{r}{2}\right)-a_0  + (a_0) \delta_{\frac{t}{2}-2}, \notag \\
	  	m(r+2+\delta_{\frac{t}{2}-2}) & \geq & 3n + p\left( \frac{1}{2}+\delta_{\frac{t}{2}-2} - \frac{r}{2}\right)+ a_0 (\delta_{\frac{t}{2}-2}-1). \label{Ineq15} 
	  	\eea
	  	It can be seen that $\delta_{j_1}$ for $r \geq 2$ has a product form as:
	  	\bea
	  	\delta_{j_1} = \frac{r}{2}\left(\frac{r^{j_1+1}-r^{j_1}}{r^{j_1+1}-1}\right). \label{Ineq16}
	  	\eea
	  	Hence for $r \geq 3$, $t \geq 4$:
	  	\bean
	  	\delta_{\frac{t}{2}-2}=\frac{r}{2}\left(\frac{r^{\frac{t}{2}-1}-r^{\frac{t}{2}-2}}{r^{\frac{t}{2}-1}-1}\right) > 1.
	  	\eean
	  	Hence we can substitute \eqref{Ineq11} for $j=\frac{t}{2}-2$ in \eqref{Ineq15} :
	  	\bea
	  	m(r+2+\delta_{\frac{t}{2}-2}) & \geq & 3n + p\left( \frac{1}{2}+\delta_{\frac{t}{2}-2} - \frac{r}{2}\right) \\ 
	  	&+& \left( \frac{m-p}{1+\sum_{i=1}^{\frac{t}{2}-1}r^i} \right) (\delta_{\frac{t}{2}-2}-1), \notag \\
	  	\scalemath{0.7}{m \left(r+2+\delta_{\frac{t}{2}-2}-\frac{\delta_{\frac{t}{2}-2}}{1+\sum_{i=1}^{\frac{t}{2}-1}r^i}+\frac{1}{1+\sum_{i=1}^{\frac{t}{2}-1}r^i} \right) }
	  	& \geq & \scalemath{0.7}{3n +  p\left( \frac{1}{2}+\delta_{\frac{t}{2}-2}-\frac{\delta_{\frac{t}{2}-2}}{1+\sum_{i=1}^{\frac{t}{2}-1}r^i} + \frac{1}{1+\sum_{i=1}^{\frac{t}{2}-1}r^i}- \frac{r}{2}\right).} \label{Ineq17}
	  	\eea
	  	Substituting \eqref{Ineq13} in \eqref{Ineq17}, we obtain:
	  	\bea
	  	m\left(r+2+\delta_{\frac{t}{2}-1}+ \frac{1}{1+\sum_{i=1}^{\frac{t}{2}-1}r^i}\right) & \geq & 3n + p\left( \frac{1}{2}+\delta_{\frac{t}{2}-1}+ \frac{1}{1+\sum_{i=1}^{\frac{t}{2}-1}r^i}- \frac{r}{2}\right). \label{Ineq18}
	  	\eea
	  	Using \eqref{Ineq16}, we obtain:
	  	\bean
	  	\left( \frac{1}{2}+\delta_{\frac{t}{2}-1}+ \frac{1}{1+\sum_{i=1}^{\frac{t}{2}-1}r^i}- \frac{r}{2}\right) > 0.
	  	\eean
	  	Hence \eqref{Ineq18} implies:
	  	\bea
	  	m\left(r+2+\delta_{\frac{t}{2}-1}+ \frac{1}{1+\sum_{i=1}^{\frac{t}{2}-1}r^i}\right) & \geq & 3n.  \label{Ineq19}
	  	\eea
	  	\eqref{Ineq19} after some algebraic manipulations gives the required upper bound on $1-\frac{m}{n}$ and hence gives the required upper bound on $\frac{k}{n}$ as stated in the theorem. Note that although the derivation is valid for $r\geq 3$, $t \geq 4$, the final bound given in the theorem is correct and tight for $t =2$. The upper bound on rate for $t=2$ can be derived specifically by substituting $a_0 \leq m$ in \eqref{Ineq7} and noting that $p \geq 0$. Note that the Appendix \ref{app:linprog} gives another approach for manipulating the inequalities to derive the same upper bound on code rate.

	  \paragraph{\bf Conditions for equality in \eqref{Thm1}}
	  	 Note that for achieving the upper bound on rate given in \eqref{Thm1}, a seq-LRC must have a parity check matrix $H$ (upto a permutation of columns) of the form given in \eqref{Hform} with parameters such that the inequalities given in \eqref{Ineq1},\eqref{Ineq2},\eqref{Ineq3},\eqref{Ineq_star1} become equalities with $p=0$ and $D$ must be an empty matrix i.e., no columns of weight $\geq 3$ (because once all these inequalities become equalities with $p=0$, the sub matrix of $H$ obtained by restricting $H$ to the columns with weights 1,2 will have each row of weight exactly $r+1$ and hence no non-zero entry can occur outside the columns having weights 1,2 for achieving the upper bound on rate). Hence it can be seen that an $(n,k,r,t)_{\text{seq}}$ code achieving the upper bound on rate \eqref{Thm1} must have a parity check matrix (upto a permutation of columns) of the form given in \eqref{eq:Hmatrixteven_ch3_App}.
          \bea
         H = \left[
         \begin{array}{c|c|c|c|c|c|c|c}
         	D_0 & A_1 & 0 & 0 & \hdots & 0 & 0 & 0  \\
         	\cline{1-8}
         	0 & D_1 & A_2 & 0 & \hdots & 0 & 0 & 0 \\
         	\cline{1-8}
         	0 & 0 & D_2 & A_3 & \hdots & 0 & 0 & 0  \\
         	\cline{1-8}
         	0 & 0 & 0 & D_3 & \hdots & 0 & 0 & 0 \\
         	\cline{1-8}
         	\vdots & \vdots & \vdots & \vdots & \hdots & \vdots & \vdots & \vdots \\
         	\cline{1-8}
         	0 & 0 & 0 & 0 & \hdots & A_{s-1} & 0 & 0  \\
         	\cline{1-8}
         	0 & 0 & 0 & 0 & \hdots & D_{s-1} & A_{s} & 0 \\
         	\cline{1-8}
         	0 & 0 & 0 & 0 & \hdots & 0 & D_{s} & C \\
         \end{array} 
         \right],  \label{eq:Hmatrixteven_ch3_App}
         \eea
	\subsubsection{\bf Case (ii) $t$ odd}  
	The following arguments are similar to $t$ even case. The proof for $t$ odd case deviates from $t$ even case starting from Lemma \ref{claim2} discussed in the following a little later. The difference in proof for $t$ odd case compared to $t$ even case is reflected in the subtle difference between the statements of Lemma \ref{claim1} for $t$ even case and  Lemma \ref{claim2} for $t$ odd case. Intuitively this happens because in p-c matrix we have the matrices $D_0,A_1,D_1,\ldots,A_s,D_s$ arranged in staircase form and we can have $2(s+1)$ linearly dependent vectors in p-c matrix in case of $t$ odd as $t=2s+1$ in case of $t$ odd and we cannot have $2(s+1)$ linearly dependent vectors for the case of $t$ even as $t=2s+2$ in case of $t$ even. Hence in $t$ odd case, if we follow the arguments same as in Lemma \ref{claim1}, we do not have the constraint that $D_s$ has each row of weight $1$. Hence in $t$ odd case, we have $D_s$ with constraint only on its column weights. If the reader is comfortable with the treatment of general form p-c matrix given in $t$ even case, the reader may directly start reading from Lemma \ref{claim2} in the following. For the sake of clarity, we repeat the arguments as follows until Lemma \ref{claim2}.
	Recall that the p-c matrix $H$ is of the form $H={[{\underline{c}_1} \  {\underline{c}_2} \ldots {\underline{c}_m}]}^T$ where ${\underline{c}_1,\underline{c}_2,\ldots,\underline{c}_m}$ are $m$ linearly independent codewords with $w_H(\underline{c}_i) \leq r+1$, $\forall i \in [m]$. 	
	  	\paragraph{\bf Step 1: Deducing the general form of p-c matrix}
	  	 Again by permuting rows and columns of $H$, the matrix can be brought into the form:
	  	\bea
	  	H = \left[
	  	\scalemath{1}{
	  		\begin{array}{c|c|c|c|c|c|c|c|c}
	  			D_0 & A_1 & 0 & 0 & \hdots & 0 & 0 & 0 & \\
	  			\cline{1-8}
	  			0 & D_1 & A_2 & 0 & \hdots & 0 & 0 & 0 & \\
	  			\cline{1-8}
	  			0 & 0 & D_2 & A_3 & \hdots & 0 & 0 & 0 & \\
	  			\cline{1-8}
	  			0 & 0 & 0 & D_3 & \hdots & 0 & 0 & 0 &  \\
	  			\cline{1-8}
	  			\vdots & \vdots & \vdots & \vdots & \hdots & \vdots & \vdots & \vdots & D\\
	  			\cline{1-8}
	  			0 & 0 & 0 & 0 & \hdots & A_{s-1} & 0 & 0 & \\
	  			\cline{1-8}
	  			0 & 0 & 0 & 0 & \hdots & D_{s-1} & A_{s} & 0 & \\
	  			\cline{1-8}
	  			0 & 0 & 0 & 0 & \hdots & 0 & D_{s} &  & \\
	  			\cline{1-7}
	  			0 & 0 & 0 & 0 & \hdots & 0 & 0 & C & \\
	  		\end{array} 
	  	}
	  	\right], \ \ \label{Hform2}
	  	\eea
	  	%
	  	where the following conditions hold. 
	  	\begin{enumerate}[(i)]
	  		\item The matrix $D_0$ is a $(\rho_0 \times a_0)$ matrix for some integers $\rho_0,a_0$. In the matrix $D_0$, each column has weight $1$ and each row has weight at least $1$. The first $a_0$ columns of $H$ contains the columns of $D_0$. The set of first $a_0$ columns of $H$ is equal to the set of all those columns of $H$ which has weight $1$,
	  		\item $A_i$ is a $(\rho_{i-1} \times a_i)$ matrix for $1 \leq i \leq s$, $D_i$ is a $(\rho_{i} \times a_i)$ matrix for $0 \leq i \leq s$ for some integers $\{\rho_j\},\{a_j\}$.  An important point to note here is that we do not insist that all the integers $\{\rho_j\},\{a_j\}$ be non-zero.  This needs some explanation.  We interpret a matrix of size $(u \times v)$ where either $u$ or $v$ or both are zero as simply representing a vacuous or empty matrix, i.e., a matrix that is not present.   We present in the figure below an example to illustrate this point. 
	  		\begin{figure}[ht!]
	  			\centering
	  			\begin{minipage}[c]{0.48\textwidth}
	  				\centering
	  				\bean
	  				H = \left[
	  				\begin{array}{c|c|c|c|c|c}
	  					D_0 & A_1 & 0 & 0 & 0 & \\
	  					\cline{1-5}
	  					0 & D_1 & A_2 & 0 & 0 & \\
	  					\cline{1-5}
	  					0 & 0 & D_2 & A_3 & 0 & D \\
	  					\cline{1-5}
	  					0 & 0 & 0 & D_3 &  & \\
	  					\cline{1-4}
	  					0 & 0 & 0 & 0 & C & \\
	  				\end{array} 
	  				\right],  
	  				\eean
	  			\end{minipage}
	  			\hspace{-0.1\textwidth}
	  			\begin{minipage}[c]{0.48\textwidth}
	  				\bean
	  				H = \left[
	  				\begin{array}{c|c|c|c|c}
	  					D_0 & A_1  & 0 & 0 & \\
	  					\cline{1-4}
	  					0 & D_1  & 0 & 0 & \\
	  					\cline{1-4}
	  					0 & 0  & A_3 & 0 & D \\
	  					\cline{1-4}
	  					0 & 0  & D_3 &  & \\
	  					\cline{1-3}
	  					0 & 0 & 0 & C & \\
	  				\end{array} 
	  				\right].
	  				\eean
	  			\end{minipage}
	  			\caption{Giving an example of form of p-c matrix given in \eqref{Hform2} when there are empty matrices among $\{A_i\},\{D_i\}$: We illustrate it for the case $t=7$. In the left is the p-c matrix when all of $a_0,a_1,a_2,a_3,\rho_0,\rho_1,\rho_2,\rho_3$ are non-zero i.e., $D_0,A_1,D_1,A_2,D_2,A_3,D_3$ are non-empty matrices. In the right is the p-c matrix when $a_0,a_1,a_3,\rho_0,\rho_1,\rho_2,\rho_3$ are non-zero but $a_2=0$ i.e., $D_0,A_1,D_1,A_3,D_3$ are non-empty matrices but $A_2,D_2$ are empty matrices. In both left and right, we assume both $C,D$ are non-empty matrices.}
	  			\label{fig:Genformpcmatrix1}
	  		\end{figure}
	  		
	  		\item  For $1\leq i \leq s$, when the matrices $A_i$,$D_i$, are not empty, then they are such that, each column of the concatenated matrix $B_i \triangleq \left[\frac{A_i}{D_i}\right]$ has weight 2, each column of $A_i$ has weight at least $1$ and each row of $D_i$ has weight at least $1$ and each column of $D_i$ has weight at most $1$.  
	  		\item The matrix $C$ is a $((\rho_{s}+p) \times a_{s+1})$ matrix for some integers $p,a_{s+1}$. Here also we do not insist that both $p,a_{s+1}$ are non-zero. When the matrix $C$ is not empty i.e., when $a_{s+1} \neq 0$, it is a matrix with each column having weight $2$. The set of columns of the matrix $D$ is precisely the set of all columns of $H$ having weight $\geq 3$.
	  	\end{enumerate}
	  	\paragraph{\bf Handling the case when any of matrices $A_i,D_i,C$ is empty} We now consider cases when any of the above matrices $\{A_j\},\{D_j\},C$ is empty.  
	  	Define:
	  	\bean
	  	\ell_1 &=& \min{\{\{j: 1 \leq j \leq s, A_j \text{ is an empty matrix}, D_j \text{ is not an empty matrix}\} \cup \{s+1 \} \}}, \\
	  	\ell_2 &=& \min{\{\{j: 1 \leq j \leq s, A_j,D_j \text{ are empty matrices}\} \cup \{s+1 \} \}}, \\
	  	\ell_3 &=& \min{\{\{j: 1 \leq j \leq s, A_j \text{ is not an empty matrix}, D_j \text{ is an empty matrix}\} \cup \{s+1 \} \}}.
	  	\eean
	  	We will now show that the case $A_j$ is an empty matrix and $D_j$ is not an empty matrix or the case $\ell_3 < \ell_2$ cannot occur. If the case $A_j$ is an empty matrix and $D_j$ is an empty matrix occur we will redefine our matrix $H$ by redefining the matrix $C$.
	  	\bit
	  	\item If $A_j$ is an empty matrix for some $j$ then $D_j$ is also an empty matrix because each column of $[\frac{A_j}{D_j}]$ has weight $2$ and $D_j$ has the constraint that each of its column has weight at most $1$. Hence the case $A_j$ an empty matrix and $D_j$ a non-empty matrix cannot occur for any $j$.
	  	\item If $A_j$ is not an empty matrix and $D_j$ is an empty matrix for some $j$ then $\rho_j=0$ and hence $A_{j+1}$ is also an empty matrix and by previous point, $D_{j+1}$ is also an empty matrix. Hence $\ell_2 \leq \ell_3+1$ (This inequality holds true obviously when $\ell_3=s$).
	  	\item We now handle the case when $A_j$ is not an empty matrix and $D_j$ is an empty matrix for some $j$ and $\ell_2 > \ell_3$ i.e., $\ell_2=\ell_3+1$. As $\ell_2 > \ell_3$, the proof of the Lemma \ref{claim2} mentioned following this discussion (since the proof of Lemma \ref{claim2} proceeds by induction starting with the proof of the lemma for $D_0$ first and then proceeding to $A_1$ and then $D_1$ and so on and hence we prove the lemma for $A_{\ell_3}$ first and then proceed to $D_{\ell_3}$ and since $D_0,A_i,D_i$ must be non-empty matrices $\forall 1 \leq i \leq \ell_3-1$) will imply that each column of $A_{\ell_3}$ has weight 1 which will imply that $D_{\ell_3}$ cannot be an empty matrix as each column of $[\frac{A_{\ell_3}}{D_{\ell_3}}]$ has weight $2$. Hence the case $A_{\ell_3}$, a non-empty matrix and $D_{\ell_3}$, an empty matrix cannot occur with $\ell_2 > \ell_3$.
	  	\item We now handle the case when $A_j,D_j$ are empty matrices for some $j$ and $\ell_2 < \ell_3$ or the case when $D_0$ is an empty matrix. Note that we can assume $\ell_2 < \ell_3$ as we have already proved that $\ell_2 > \ell_3$ cannot occur. We set $\ell_2=0$, if $D_0$ is an empty matrix. We now redefine $H$ based on the value of $\ell_2$.
	  	We set $A_i$, $D_i$ to be empty matrices and set $a_i=0, \rho_i=0$, $\forall \ell_2 \leq i \leq s$. 
	  	Let $E_2  \subseteq \{\sum_{i=0}^{\ell_2-1}a_i+1,\ldots,n-1,n\}$ such that $E_2$ is the set of labels of all the 2-weight columns of $H$ apart from those 2-weight columns of $H$ containing the columns of $B_1,B_2,\ldots,B_{\ell_2-1}$. By a $2$-weight column, we refer to a column of weight $2$. Let $E_1 = \{\sum_{i=0}^{\ell_2-1}{\rho_i}+1,\ldots,m-1,m\}$. If $\ell_2=s+1$ then $H$ is defined by \eqref{Hform2} with none of the matrices among $\{A_i\},\{D_i\}$ being empty. If $\ell_2 < s+1$, we redefine $C=H|_{E_1,E_2}$.   If $\ell_2 < s+1$, the matrix $H$ can be written in the form given in \eqref{Hform_11} and hence defined by \eqref{Hform_11}.
	  	\bea
	  	H = \left[
	  	\scalemath{1}{
	  		\begin{array}{c|c|c|c|c|c|c|c|c}
	  			D_0 & A_1 & 0 & 0 & \hdots & 0 & 0 & 0 & \\
	  			\cline{1-8}
	  			0 & D_1 & A_2 & 0 & \hdots & 0 & 0 & 0 & \\
	  			\cline{1-8}
	  			0 & 0 & D_2 & A_3 & \hdots & 0 & 0 & 0 & \\
	  			\cline{1-8}
	  			0 & 0 & 0 & D_3 & \hdots & 0 & 0 & 0 &  \\
	  			\cline{1-8}
	  			\vdots & \vdots & \vdots & \vdots & \hdots & \vdots & \vdots & \vdots & D\\
	  			\cline{1-8}
	  			0 & 0 & 0 & 0 & \hdots & A_{\ell_2-2} & 0 & 0 & \\
	  			\cline{1-8}
	  			0 & 0 & 0 & 0 & \hdots & D_{\ell_2-2} & A_{\ell_2-1} & 0 & \\
	  			\cline{1-8}
	  			0 & 0 & 0 & 0 & \hdots & 0 & D_{\ell_2-1} & 0 & \\
	  			\cline{1-8}
	  			0 & 0 & 0 & 0 & \hdots & 0 & 0 & C & \\
	  		\end{array} 
	  	}
	  	\right]. \ \ \label{Hform_11}
	  	\eea
	  	
	  	Note that none of the matrices $D_0,A_1,D_1,A_2,D_2,\ldots,A_{\ell_2-1},D_{\ell_2-1}$ are empty matrices. Irrespective of the value of $\ell_2$, let the number of columns in $C$ be denoted as $a_{s+1}$ and the number of rows of $C$ be denoted as $\rho_{s}+p$. Note that $\rho_{s}=0$, if $D_{s}$ is an empty matrix because if $A_{s}$ is also an empty matrix then clearly we have set $\rho_{s}=0$ and the case $A_{s}$ is not an empty matrix and $D_{s}$ is an empty matrix cannot occur at this point. If $C$ is an empty matrix then we can clearly set $a_{s+1}=0$.
	  	The entire derivation of upper bound on rate is correct and all the inequalities in the derivation will hold with $a_i=0$, $\rho_i=0$, $\forall \ell_2 \leq i \leq s$ and $a_{s+1}=0$ (if $C$ is an empty matrix).
	  	\eit
	  	Although we have to prove the following Lemma \ref{claim1} for $A_i,D_i$, $\forall 1 \leq i \leq \ell_2-1$, $D_0$, we assume all $D_0,$ $A_i,D_i$, $\forall 1 \leq i \leq s$ to be non-empty matrices and prove the lemma. Since the proof of the lemma is by induction, the induction can be made to stop after proving the lemma for $A_{\ell_2-1},D_{\ell_2-1}$ (The proof of the lemma is by induction and induction starts by proving the lemma for $D_0$ and proceeds to $A_1$ and to $D_1$ and so on as mentioned before) and the proof is unaffected by it.
	  	\paragraph{\bf Deducing the structure of p-c matrix further}
	  	\begin{lem}\label{claim2}
	  		
	  		For $1\leq i \leq s$, $A_i$ is a matrix with each column having weight 1.
	  		For $0 \leq i \leq s-1$, $D_i$ is a matrix with each row and each column having weight 1. $D_{s}$ is a matrix with each column having weight 1.

	  	\end{lem}
	  	\begin{proof}
	  		Proof is exactly same as the proof of Lemma \ref{claim1} and proceeds by induction. We use exactly same arguments as the proof of Lemma \ref{claim1} until proving the Lemma \ref{claim2} by induction for $D_0$ and then for $A_1$ and then for $D_1$ and so on until we prove the lemma for $A_s$. Note that the property that $A_s$ has each column of weight $1$ implies $D_s$ has each column of weight $1$ and we stop the proof at this point as the proof of Lemma \ref{claim2} is done. Since it is exact repitition of arguments of Lemma \ref{claim1} (except that we do not prove that $D_s$ has each row of weight $1$ at the end of the proof. The intuitive reason why we cannot do this is explained at the beginning of the proof for the current $t$ odd case), we skip the proof.
	  	\end{proof}
	  	Note that here in $t$ odd case, $D_s$ need not have the property that each of its rows are of weight $1$ but $D_s$ had the property that each of its rows are of weight $1$ in $t$ even case. This is the point where the proof deviates from the $t$ even case.
	  	
	  		\paragraph{\bf Step 2: Forming linear inequalities in variables $n,m,\{a_i\},\{\rho_i\}$}
	  		  	We now form linear inequalities in variables $n,m,\{a_i\},\{\rho_i\}$ based on counting the non-zero entries row-wise and column-wise in various matrices $\{A_i\},\{D_i\},C,H$ using the lemma \ref{claim2} in the process. By Lemma \ref{claim2}, after permutation of columns of $H$ (in \eqref{Hform2} or \eqref{Hform_11} depending on the value of $\ell_2$)  within the columns labeled by the set $\{\sum_{\ell=0}^{j-1}a_{\ell}+1,\ldots,\sum_{\ell=0}^{j-1}a_{\ell}+a_{j}\}$ for $0 \leq j \leq \min(\ell_2-1,s-1)$, the matrix $D_j,0 \leq j \leq \min(\ell_2-1,s-1)$ can be assumed to be a diagonal matrix with non-zero entries along the diagonal and hence $\rho_{i}=a_{i}$, for $0 \leq i \leq s-1$ because $\rho_{i}=a_{i}=0$ for $i > \min(\ell_2-1,s-1)$. 
	  	
	  	Since the sum of the column weights of $A_i, 1 \leq i \leq s$ must equal the sum of the row weights and since each row of $A_i$ for $i \leq \ell_2-1$ can have weight atmost $r$ and not $r+1$ due to weight one rows in $D_{i-1}$, and since for $\ell_2 \leq i \leq s$, $A_i$ is an empty matrix and we have set $a_i=0$, we obtain:
	  	\bea
	  	\text{For } 1 \leq i \leq s : & & \notag \\
	  	\rho_{i-1} r & \geq & a_i, \notag \\
	  	a_{i-1} r & \geq & a_i. \label{Ineq1O}
	  	\eea
	  	We also have that,
	  	\bea
	  	\rho_{s}+\sum_{i=0}^{s-1} \rho_i + p = \rho_{s}+\sum_{i=0}^{s-1} a_i + p =m.  \label{Ineq2O}
	  	\eea 
	  	By equating sum of row weights of $[\frac{D_{s}}{0} | C]$, with sum of column weights of $[\frac{D_{s}}{0} | C]$, we obtain:
	  	\bea
	  	2 a_{s+1}+a_{s} & \leq & (\rho_{s}+p) (r+1). \label{Ineq3O}	    
	  	\eea
	  	If $\ell_2 \leq s$ then the number of rows in $C$ is $p$ with each column of $C$ having weight 2 and $a_{s}=\rho_{s}=0$ (and $a_{s+1}=0$ if $C$ is also an empty matrix), hence the inequality \eqref{Ineq3O} is true. If $\ell_2=s+1$ and $C$ is an empty matrix then also the inequality \eqref{Ineq3O} is true as we would have set $a_{s+1}=0$.
	  	
	  	Substituting \eqref{Ineq2O} in \eqref{Ineq3O}:
	  	\bea
	  	2 a_{s+1} & \leq & (m-\sum_{i=0}^{s-1} a_i) (r+1)-a_{s}. \label{Ineq4O}
	  	\eea
	  	By equating sum of row weights of $D_{s}$, with sum of column weights of $D_{s}$, we obtain (Note that if $D_{s}$ is an empty matrix then also the following inequality is true as we would have set $a_{s}=0$):
	  	\bea
	  	a_{s} \leq \rho_{s}(r+1). \label{OddSpe}
	  	\eea
	  	By equating sum of row weights of $H$, with sum of column weights of $H$, we obtain
	  	\bea
	  	m(r+1) & \geq & a_0 + 2(\sum_{i=1}^{s+1} a_i) + 3 (n-\sum_{i=0}^{s+1} a_i), \label{Ineq_star2}
	  	\eea 
	  	If $\ell_2 \leq s$ then $a_i=0$ $\forall \ell_2 \leq i \leq s$. If $C$ is an empty matrix then $a_{s+1}=0$. Hence the inequality \eqref{Ineq_star2} is true irrespective of whether $\ell_2=s+1$ or $\ell_2 \leq s$ (even if $C$ is an empty matrix).
	  	\bea
	  	m(r+1) & \geq & 3n-2a_0 -(\sum_{i=1}^{s+1} a_i). \label{Ineq5O} 
	  	\eea
	  	
	  	Our basic inequalities are \eqref{Ineq1O},\eqref{Ineq2O},\eqref{Ineq3O},\eqref{OddSpe},\eqref{Ineq_star2}. We manipulate these 5 inequalities to derive the bound on rate.
	  		\paragraph{\bf Step 3: Algebraic manipulation of the 5 linear inequalities to derive the bound on rate}
	  	Substituting \eqref{Ineq4O} in \eqref{Ineq5O}:
	  	\bea
	  	m(r+1) & \geq & 3n-2a_0 -(\sum_{i=1}^{s} a_i) - \left(\frac{(m-\sum_{i=0}^{s-1} a_i) (r+1)-a_{s}}{2} \right). \label{Ineq6O}
	  	\eea
	  	For $s=0$, \eqref{Ineq6O} becomes:
	  	\bea
	  	m(r+1) & \geq & 3n-2a_0 - \left(\frac{m (r+1)-a_{0}}{2} \right), \notag \\        m\frac{3(r+1)}{2} & \geq & 3n-\frac{3}{2}a_0. \label{Ineq7O} 
	  	\eea
	  	Substituting  \eqref{OddSpe} in \eqref{Ineq7O}:
	  	\bea
	  	m\frac{3(r+1)}{2} & \geq & 3n-\frac{3}{2}\rho_0 (r+1), \label{Ineq8O} 
	  	\eea
	  	Substituting  \eqref{Ineq2O} in \eqref{Ineq8O}:
	  	\bea
	  	m\frac{3(r+1)}{2} & \geq & 3n-\frac{3}{2}(m-p) (r+1), \notag \\
	  	\text{Since $p \geq 0$, } 3m(r+1) & \geq & 3n. \label{Ineq9O}
	  	\eea
	  	\eqref{Ineq9O} implies,
	  	\bea
	  	\frac{k}{n} \leq \frac{r}{r+1}. \label{Ineq10O}
	  	\eea
	  	\eqref{Ineq10O} proves the bound \eqref{Thm2} for $s=0$. Hence from now on we assume $s \geq 1$.\\
	  	For $s \geq 1$, \eqref{Ineq6O} implies:
	  	\bea
	  	m \frac{3(r+1)}{2} & \geq & 3n+ a_0 \left( \frac{r+1}{2}-2 \right) +(\sum_{i=1}^{s-1} a_i) \left( \frac{r+1}{2}-1 \right) - \frac{a_{s}}{2}. \label{Ineq11O} 
	  	\eea
	  	Substituting  \eqref{Ineq1O} in \eqref{Ineq11O} and since $r \geq 3$:    
	  	\bea
	  	m \frac{3(r+1)}{2} & \geq & 3n+ \frac{a_{s}}{r^{s}} \left( \frac{r+1}{2}-2 \right) +(\sum_{i=1}^{s-1} \frac{a_{s}}{r^{s-i}}) \left( \frac{r+1}{2}-1 \right) - \frac{a_{s}}{2},  \notag \\
	  	m \frac{3(r+1)}{2} & \geq & 3n+ a_{s} \left( \left(\sum_{i=1}^{s} \frac{1}{r^i} \right) \left(\frac{r+1}{2}-1\right) -\frac{1}{r^{s}}-\frac{1}{2} \right),  \notag \\
	  	m \frac{3(r+1)}{2} & \geq & 3n- a_{s} \left(\frac{3}{2 r^{s}}\right). \label{Ineq12O}
	  	\eea
	  	Rewriting \eqref{Ineq2O}:
	  	\bea
	  	\rho_{s}+\sum_{i=0}^{s-1} a_i + p =m.  \label{Ineq13O}
	  	\eea
	  	Substituting \eqref{OddSpe},\eqref{Ineq1O} in \eqref{Ineq13O}:
	  	\bea
	  	\rho_{s}+\sum_{i=0}^{s-1} a_i + p &=& m, \notag \\
	  	\frac{a_{s}}{r+1}+\sum_{i=0}^{s-1} \frac{a_{s}}{r^{s-i}} & \leq & m-p, \notag \\
	  	a_{s} \leq \frac{m-p}{\frac{1}{r+1}+\sum_{i=1}^{s} \frac{1}{r^i}}, \notag \\
	  	a_{s} \leq \frac{(m-p)(r+1)}{1+\frac{(r^{s}-1)(r+1)}{(r^{s})(r-1)} }. \label{Ineq14O}
	  	\eea
	  	Substituting \eqref{Ineq14O} in \eqref{Ineq12O}:
	  	\bea
	  	m \frac{3(r+1)}{2} & \geq & 3n- \frac{(m-p)(r+1)}{1+\frac{(r^{s}-1)(r+1)}{(r^{s})(r-1)} } \left(\frac{3}{2 r^{s}}\right), \notag \\
	  	\text{Since $p \geq 0$, } m \frac{3(r+1)}{2} \left( 1+\frac{1}{r^{s}+\frac{(r^{s}-1)(r+1)}{(r-1)}} \right) & \geq & 3n. \label{Ineq15O} 
	  	\eea
	  	\eqref{Ineq15O} after some algebraic manipulations gives the required upper bound on $1-\frac{m}{n}$ and hence gives the required upper bound on $\frac{k}{n}$ as stated in the theorem.
	  	\paragraph{\bf Conditions for equality in \eqref{Thm2}}
	  	Note that for achieving the upper bound on rate given in \eqref{Thm2}, a seq-LRC must have a parity check matrix $H$ (upto a permutation of columns) of the form given in \eqref{Hform2} with parameters such that the inequalities given in \eqref{Ineq1O},\eqref{Ineq2O},\eqref{Ineq3O},\eqref{OddSpe},\eqref{Ineq_star2} become equalities with $p=0$ and $D$ must be an empty matrix i.e., no columns of weight $\geq 3$ (because once all these inequalities become equalities with $p=0$, the sub matrix of $H$ obtained by restricting $H$ to the columns with weights $1,2$ will have each row of weight exactly $r+1$ and hence no non-zero entry can occur outside the columns having weights $1,2$ for achieving the upper bound on rate). Note that for achieving the upper bound on rate, $C$ must also be an empty matrix. This is because inequality \eqref{OddSpe} must become an equality which implies that $D_s$ is a matrix with each row of weight $r+1$ and we also saw that $p=0$. Hence $C$ must be a zero matrix which implies $C$ is an empty matrix. Hence it can be seen that an $(n,k,r,t)_{\text{seq}}$ code achieving the upper bound on rate \eqref{Thm2} must have a parity check matrix (upto a permutation of columns) of the form given in \eqref{eq:Hmatrixtodd_ch3_App}.
	  	
	  	\bea
	  	H & = & \left[
	  	\begin{array}{c|c|c|c|c|c|c}
	  		D_0 & A_1 & 0 & 0 & \hdots & 0 & 0   \\
	  		\cline{1-7}
	  		0 & D_1 & A_2 & 0 & \hdots & 0 & 0  \\
	  		\cline{1-7}
	  		0 & 0 & D_2 & A_3 & \hdots & 0 & 0   \\
	  		\cline{1-7}
	  		0 & 0 & 0 & D_3 & \hdots & 0 & 0   \\
	  		\cline{1-7}
	  		\vdots & \vdots & \vdots & \vdots & \ddots & \vdots & \vdots  \\
	  		\cline{1-7}
	  		0 & 0 & 0 & 0 & \hdots & A_{s-1} & 0   \\
	  		\cline{1-7}
	  		0 & 0 & 0 & 0 & \hdots & D_{s-1} & A_{s}   \\
	  		\cline{1-7}
	  		0 & 0 & 0 & 0 & \hdots & 0 & D_{s} 
	  	\end{array} \right] \label{eq:Hmatrixtodd_ch3_App},
	  	\eea
	  	
		\section{An Alternative, Linear-Programming-Based Derivation of Theorem \ref{thm:rate_both}} \label{app:linprog}

	  An alternative proof of Theorem \ref{thm:rate_both} using linear programming is given in this appendix. This approach based on linear programming uses the inequalities on $a_0, a_1,\ldots, p$ (given in Appendix \ref{Appendix_rate_bounds}) to calculate an upper bound on dimension of a seq-LRC for a given $n,r,t$. Hence for the cases when the upper bound on dimension is achievable, solving the linear programming problem under integer variable constraint gives the parameters $a_0, a_1,\ldots, p$ which can be used to construct dimension-optimal codes.
	  \subsubsection{Case (i) $t$ even}
	  Our basic inequalities are \eqref{Ineq1},\eqref{Ineq2},\eqref{Ineq3},\eqref{Ineq5}. The inequalities \eqref{Ineq1},\eqref{Ineq3} and \eqref{Ineq5}, are linear inequalities and are written in matrix form, after substituting \eqref{Ineq2} in \eqref{Ineq5}, as\footnote{Suppose $\mathbf{\underline{x}} = [x_1  x_2 \ldots x_n]^T $ and $\mathbf{\underline{y}} = [y_1  y_2 \ldots  y_n]^T $, then $\mathbf{\underline{x}} \geq \mathbf{\underline{y}}$ denotes that $x_i \geq y_i$, $\forall$ $1 \leq i \leq n$}:
	  \begin{equation*}
	  A\mathbf{\underline{x}} \geq \mathbf{\underline{b}}
	  \end{equation*}
	  where
	  \bea
	  A & = & \left[
	  \begin{array}{c c c c c c c c}
	  	r & -1 & 0 & \hdots & 0 & 0 & 0 & 0 \\
	  	0 & r & -1 & \hdots & 0 & 0 & 0 & 0 \\
	  	\vdots & \vdots & \vdots & \ddots & \vdots & \vdots & \vdots & \vdots \\
	  	0 & 0 & 0 & \hdots & r & -1 & 0 & 0 \\
	  	0 & 0 & 0 & \hdots & 0 & r & -2 & (r+1) \\
	  	(r+3) & (r+2) & (r+2) & \hdots & (r+2) & (r+2) & 1 & (r+1)
	  \end{array} \right] \label{Amat}
	  \eea
	  which is a $(\frac{t}{2} + 1) \times (\frac{t}{2} + 2) $ matrix and
	  \bea 
	  \mathbf{\underline{x}} = \left[ \begin{array}{c c c c c}
	  	a_0 & a_1 & \hdots & a_{\frac{t}{2}} & p
	  \end{array} \right]^T, \mathbf{\underline{b}} = \left[ \begin{array}{c c c c c}
	  	0 &	0 &	\hdots & 0 & 3n
	  \end{array}\right]^T
	  \eea
	  where $\mathbf{\underline{x}}$ is a $(\frac{t}{2} + 2) \times 1$ matrix and $\mathbf{\underline{b}}$ is a $(\frac{t}{2} + 1) \times 1$ matrix.
	  The problem of finding an upper bound on rate of the code now becomes one of minimizing $m = \mathbf{\underline{c}}^T\mathbf{\underline{x}}$, which is a linear objective function where $\mathbf{\underline{c}} = \left[ \begin{array}{c c c c c c}
	  1 & 1 & \hdots & 1 & 0 & 1
	  \end{array}\right]^T$ is a $(\frac{t}{2} + 2) \times 1$ matrix. Also by definition of $\mathbf{\underline{x}}$, $\mathbf{\underline{x}} \geq \mathbf{\underline{0}}$ (the all-zero vector). 	
	  Hereon we follow the terminology in the book \cite{ChongZak}. We restate certain relevant terms and results from the book here:\\
	  
	  Consider a linear programming problem of the form:
	  \begin{align*}
	  \text{minimize } & \mathbf{\underline{c}}^T\mathbf{\underline{x}}\\
	  \text{s.t. } A\mathbf{\underline{x}} & = \mathbf{\underline{b}}, \mathbf{\underline{x}} \geq \mathbf{\underline{0}}
	  \end{align*}
	  \textbf{Basic solution:} Consider the equalities ${A\underline{\mathbf{x}}=\underline{\mathbf{b}}}, A\in \mathbb{R}^{m_0\times n_0}$. Let $m_0$ be the rank of $A$. Let $B$ be a non-singular submatrix of $A$ such that it has the same number of rows as $A$. The system $B\underline{\mathbf{x}_B}=\underline{\mathbf{b}}$ can be solved as $\underline{\mathbf{x}_B}=B^{-1}\underline{\mathbf{b}}$. Let $\underline{\mathbf{x}}=[{\underline{\mathbf{x}_B}}^T \underline{\mathbf{0}}^T]^T$, then $\underline{\mathbf{x}}$ is a solution to ${A\underline{\mathbf{x}}=\underline{\mathbf{b}}}$. Then the vector $[{\underline{\mathbf{x}_B}}^T \underline{\mathbf{0}}^T]^T$ is called a \textit{basic solution} to $A\underline{\mathbf{x}}=\underline{\mathbf{b}}$ with respect to the basis $B$. The components of the vector $\underline{\mathbf{x}_B}$ are called \textit{basic variables} and the columns of $B$ are called \textit{basic columns}.\\
	  \textbf{Feasible solution:} A vector $\underline{\mathbf{x}}$ satisfying $A\underline{\mathbf{x}}=\underline{\mathbf{b}}$ and $\underline{\mathbf{x}}\geq \underline{\mathbf{0}}$ is called a \textit{feasible solution}.\\
	  \textbf{Basic feasible solution:} A feasible solution that is also basic is called a \textit{basic feasible solution}.\\
	  \textbf{Optimal feasible solution:} Any vector $\underline{\mathbf{x}}$ that yields the minimum value of the objective function $\mathbf{\underline{c}}^T\mathbf{\underline{x}}$ over the set of vectors satisfying the constraints $A\mathbf{\underline{x}} = \mathbf{\underline{b}}$ and  $\mathbf{\underline{x}} \geq \mathbf{\underline{0}}$ is said to be an \textit{optimal feasible solution}.\\
	  \textbf{Optimal basic feasible solution:} An optimal feasible solution that is basic is called an \textit{optimal basic feasible solution}.\\
	  \begin{thm}{\cite{ChongZak}}{\textbf{Fundamental theorem of linear programming}\\}
	  	Consider a linear program in the form
	  	\begin{align*}
	  	\text{minimize } & \mathbf{\underline{c}}^T\mathbf{\underline{x}}\\
	  	\text{s.t. } A\mathbf{\underline{x}} & = \mathbf{\underline{b}}, \mathbf{\underline{x}} \geq \mathbf{\underline{0}}
	  	\end{align*}
	  	\begin{enumerate}
	  		\item If there exists a feasible solution, then there exists a basic feasible solution.
	  		\item If there exists an optimal feasible solution, then there exists an optimal basic feasible solution.
	  	\end{enumerate}
  	  \label{thm:fundamental thm linprog}
	  \end{thm}
	  Due to above theorem, the task of solving a linear programming problem is reduced to searching over a finite set of basic solutions.\\
	  \begin{thm}{\cite{ChongZak}}
	  	A basic feasible solution is optimal if and only if the corresponding reduced cost coefficients are all non-negative.
	  	\label{thm:optimal basic feasible solution}
	  \end{thm}
	  The quantity \textit{reduced cost coefficient} will be defined later.\\
	  \textbf{Primal problem and Dual problem:} Consider a linear programming problem of the form
	  \begin{align*}
	  \text{minimize } & \mathbf{\underline{c}}^T\mathbf{\underline{x}}\\
	  \text{s.t. } A\mathbf{\underline{x}} & \geq \mathbf{\underline{b}}, \mathbf{\underline{x}} \geq \mathbf{\underline{0}}
	  \end{align*}
	  The above problem is referred to as the \textit{primal problem}. The \textit{dual problem} is defined as below:
	  \begin{align*}
	  \text{maximize } & \mathbf{\underline{b}}^T\mathbf{\underline{\lambda}}\\
	  \text{s.t. } A^T\mathbf{\underline{\lambda}} & \leq \mathbf{\underline{c}}, \mathbf{\underline{\lambda}} \geq \mathbf{\underline{0}}
	  \end{align*}
	  \begin{thm}{\cite{ChongZak}}{\textbf{Strong Duality}\\}
	  	If the primal problem has an optimal solution, then so does the dual, and the optimal values of their respective objective functions are equal.
	  	\label{thm:strong duality}
	  \end{thm}
  	  \begin{cor}
  	  	The primal problem has an optimal solution is and only if the dual problem has an optimal solution
  	  \end{cor}
      \begin{proof}
      	It can be verified that the dual of the dual problem is the primal problem itself. Then the proof follows from Theorem \ref{thm:strong duality}
      \end{proof}
	  The above fact is made use of in solving the rate-bound as a linear program.\\
	  The problem at hand is now in a standard form of a linear program formulation as:
	  \begin{align*}
	  \text{minimize } & \mathbf{\underline{c}}^T\mathbf{\underline{x}}\\
	  \text{s.t. } A\mathbf{\underline{x}} & \geq \mathbf{\underline{b}}, \mathbf{\underline{x}} \geq \mathbf{\underline{0}}
	  \end{align*}
	  where
	  \bea
	  A & = & \left[
	  \begin{array}{c c c c c c c c}
	  	r & -1 & 0 & \hdots & 0 & 0 & 0 & 0 \\
	  	0 & r & -1 & \hdots & 0 & 0 & 0 & 0 \\
	  	\vdots & \vdots & \vdots & \ddots & \vdots & \vdots & \vdots & \vdots \\
	  	0 & 0 & 0 & \hdots & r & -1 & 0 & 0 \\
	  	0 & 0 & 0 & \hdots & 0 & r & -2 & (r+1) \\
	  	(r+3) & (r+2) & (r+2) & \hdots & (r+2) & (r+2) & 1 & (r+1)
	  \end{array} \right] \label{Amat}
	  \eea
	  which is a $(\frac{t}{2} + 1) \times (\frac{t}{2} + 2) $ matrix and
	  \bea 
	  \mathbf{\underline{x}} = \left[ \begin{array}{c c c c c}
	  	a_0 & a_1 & \hdots & a_{\frac{t}{2}} & p
	  \end{array} \right]^T, \mathbf{\underline{b}} = \left[ \begin{array}{c c c c c}
	  	0 &	0 &	\hdots & 0 & 3n
	  \end{array}\right]^T
	  \eea
	  where $\mathbf{\underline{x}}$ is a $(\frac{t}{2} + 2) \times 1$ matrix and $\mathbf{\underline{b}}$ is a $(\frac{t}{2} + 1) \times 1$ matrix, and $\mathbf{\underline{c}} = \left[ \begin{array}{c c c c c c}
	  1 & 1 & \hdots & 1 & 0 & 1
	  \end{array}\right]^T$ is a $(\frac{t}{2} + 2) \times 1$ matrix.
	  
	  The dual problem of the above is
	  \begin{align*}
	  \text{maximize } & \mathbf{\underline{b}}^T\mathbf{\underline{\lambda}}\\
	  \text{s.t. } A^T\mathbf{\underline{\lambda}} & \leq \mathbf{\underline{c}}, \mathbf{\underline{\lambda}} \geq \mathbf{\underline{0}}
	  \end{align*}
	  We will solve the dual problem by writing it in standard ``$\text{minimize}$ $-\mathbf{\underline{b}}^T\underline{\lambda}$" form. We solve the dual problem since the steps involved in deriving a closed form solution were simpler for the dual compared to the primal problem. For example, notice that, the vector $\underline{\mathbf{b}}$ in the objective function of the dual problem has only one non-zero entry. Let us introduce \textit{slack variables} $s_1,\ldots,s_{\frac{t}{2}+2}$ and re-write the constraints as
	  \begin{align*}
	  B\mathbf{\underline{v}} = \mathbf{\underline{c}}, \text{ }\mathbf{\underline{v}} \geq \mathbf{\underline{0}},
	  \end{align*}
	  where
	  \bea
	  B & = & \left[
	  \begin{array}{c c c c c c c c c c c c c}
	  	r & 0 & 0 & \hdots & 0 & (r+3) & 0 & 1 & 0 & 0 & \hdots & 0 & 0 \\
	  	-1 & r & 0 & \hdots & 0 & (r+2) & 0 & 0 & 1 & 0 &\hdots & 0 & 0\\
	  	0 & -1 & r & \hdots & 0 & (r+2) & 0 & 0 & 0 & 1 &\hdots & 0 & 0\\
	  	\vdots & \vdots & \ddots & \ddots & \vdots & \vdots & \vdots & \vdots & \vdots & \vdots & \ddots & \vdots & \vdots\\
	  	0 & 0 & \hdots & -1 & r & (r+2) & 0 & 0 & 0 & 0 & \hdots & 1 & 0 \\
	  	0 & 0 & 0 & \hdots & -2 & 1 & 0 & 0 & 0 & 0 & \hdots & 0 & 1 \\
	  	0 & 0 & 0 & \hdots & (r+1) & (r+1) & 1 & 0 & 0 & 0 & \hdots & 0 & 0 \\ 
	  \end{array} \right] \label{Amat},
	  \eea
	  and 
	  \bea 
	  \mathbf{\underline{v}} = \left[ \begin{array}{c c c c c c c}
	  	\lambda_1 & \lambda_2 & \hdots & \lambda_{\frac{t}{2}+1} & s_1 & \hdots & s_{\frac{t}{2}+2}
	  \end{array} \right]^T.
	  \eea
	  With this, the objective function now is $\mathbf{\underline{d}}^T\mathbf{\underline{v}}$, where $\mathbf{\underline{d}} = \left[\begin{array}{c c c c c}
	  -\mathbf{\underline{b}}^T & 0 & 0 & \hdots & 0
	  \end{array} \right]^T$ which is a $(t+3) \times 1$ matrix.
	  Define $\mathbf{\underline{v}}_{BV}=[\lambda_1,...,\lambda_{\frac{t}{2}+1},s_1]^T$.
	  
	  We pick the variables $\beta_1=\lambda_1$,...,$\beta_{\frac{t}{2}+1}=\lambda_{\frac{t}{2}+1}$, $\beta_{\frac{t}{2}+2}=s_1$ as ``basic variables" and the rest, called ``non-basic variables" $\alpha_1=s_2,...,\alpha_{\frac{t}{2}+1}=s_{\frac{t}{2}+2}$ will be set to $0$. The set of basic variables is chosen such that the matrix formed by the columns of $B$ corresponding to these basic variables is a full-rank square matrix $B_{BV}$. This fact for the chosen basic variables can be verified. The remaining columns of $B$ will give a matrix $B_{NBV}$. The system of equations is now in the following form:
	  \begin{align*}
	  \left[\begin{array}{c c}
	  B_{BV} & B_{NBV}
	  \end{array}\right]
	  \left[\begin{array}{c}
	  \mathbf{\underline{v}}_{BV} \\
	  \hline
	  \mathbf{\underline{0}}
	  \end{array}
	  \right] = \mathbf{\underline{c}}
	  \end{align*}
	  Therefore we will equivalently solve
	  \begin{align*}
	  B_{BV}\mathbf{\underline{v}}_{BV} = \mathbf{\underline{c}}
	  \end{align*}
	  The above system of equations can be solved in closed form to get the following:
	  \begin{align}
	  \lambda_{\frac{t}{2}+1} & = \frac{2\sum_{i = 0}^{\frac{t}{2}-1}r^i}{3(r^{\frac{t}{2}}+2\sum_{i = 0}^{\frac{t}{2}-1}r^i)}, \label{LP1} \\
	  s_1 & = \frac{1}{r^{\frac{t}{2}}+2\sum_{i = 0}^{\frac{t}{2}-1}r^i}, \label{LP2}\\
	  \lambda_{j+1} & = \frac{r^{\frac{t}{2}} - 3r^{\frac{t}{2}-(j+1)} + 2}{3(r-1)(r^{\frac{t}{2}}+2\sum_{i = 0}^{\frac{t}{2}-1}r^i)},\text{   for }0 \leq j \leq \frac{t}{2}-1 \label{LP3}
	  \end{align}
	  which are non-negative if $r \geq 3$. Hence the solution given by \eqref{LP1},\eqref{LP2},\eqref{LP3} is a basic (with respect to the basis formed by the columns of matrix $B$ corresponding to the basic variables) feasible solution. Let the elements of the vector $\mathbf{\underline{d}}$ be indexed by the elements of the vector $\mathbf{\underline{v}}$ i.e., $i^{th}$ component of $\mathbf{\underline{d}}$ is indexed by $i^{th}$ component of $\mathbf{\underline{v}}$. To check for optimality of the basic solution, as per Theorem \ref{thm:optimal basic feasible solution}, we check if the ``reduced cost coefficients" $r_{\alpha_i} = d_{\alpha_i} - z_{\alpha_i}$ are non-negative, for every non-basic variable $\alpha_i,1\leq i \leq \frac{t}{2}+1$. We note that for the above made choice of basic and non-basic variables, in the vector $\mathbf{\underline{d}}$ only $d_{\beta_{\frac{t}{2}+1}} = -3n$ is non-zero. The quantity $z_{\alpha_i}$ is defined as follows:
	  \begin{equation*}
	  z_{\alpha_i} = \sum_{j = 1}^{{\frac{t}{2}+2}}d_{\beta_j}y_{(j,\alpha_i)} = d_{\beta_{\frac{t}{2}+1}}y_{(\frac{t}{2}+1,\alpha_i)} = -3ny_{(\frac{t}{2}+1,\alpha_i)}
	  \end{equation*}
	  where $y_{(\frac{t}{2}+1,\alpha_i)}$ are as shown in the row reduced echelon form of matrix B below:
	  \begin{align*}
	  B_{rref} = \left[\scalemath{0.8}{\begin{array}{c c c c c c c c c c c c c}
	  	1 & 0 & 0 & \hdots & 0 & 0 & 0 & y_{(1,\alpha_1)} & y_{(1,\alpha_2)} & y_{(1,\alpha_3)} & \hdots & y_{(1,\alpha_{\frac{t}{2}})} & y_{(1,\alpha_{\frac{t}{2}+1})} \\
	  	0 & 1 & 0 & \hdots & 0 & 0 & 0 & y_{(2,\alpha_1)} & y_{(2,\alpha_2)} & y_{(2,\alpha_3)} & \hdots & y_{(2,\alpha_{\frac{t}{2}})} & y_{(2,\alpha_{\frac{t}{2}+1})}\\
	  	0 & 0 & 1 & \hdots & 0 & 0 & 0 & y_{(3,\alpha_1)} & y_{(3,\alpha_2)} & y_{(3,\alpha_3)} & \hdots & y_{(3,\alpha_{\frac{t}{2}})} & y_{(3,\alpha_{\frac{t}{2}+1})}\\
	  	\vdots & \vdots & \vdots & \ddots & \vdots & \vdots & \vdots & \vdots & \vdots & \vdots & \ddots & \vdots & \vdots\\
	  	0 & 0 & 0 & \hdots & 1 & 0 & 0 & y_{(\frac{t}{2},\alpha_1)} & y_{(\frac{t}{2},\alpha_2)} & y_{(\frac{t}{2},\alpha_3)} & \hdots & y_{(\frac{t}{2},\alpha_{\frac{t}{2}})} & y_{(\frac{t}{2},\alpha_{\frac{t}{2}+1})}\\
	  	0 & 0 & 0 & \hdots & 0 & 1 & 0 & y_{(\frac{t}{2}+1,\alpha_1)} & y_{(\frac{t}{2}+1,\alpha_2)} & y_{(\frac{t}{2}+1,\alpha_3)} & \hdots & y_{(\frac{t}{2}+1,\alpha_{\frac{t}{2}})} & y_{(\frac{t}{2}+1,\alpha_{\frac{t}{2}+1})}\\
	  	0 & 0 & 0 & \hdots & 0 & 0 & 1 & y_{(\frac{t}{2}+2,\alpha_1)} & y_{(\frac{t}{2}+2,\alpha_2)} & y_{(\frac{t}{2}+2,\alpha_3)} & \hdots & y_{(\frac{t}{2}+2,\alpha_{\frac{t}{2}})} & y_{(\frac{t}{2}+2,\alpha_{\frac{t}{2}+1})}\\ 
	  	\end{array}}
	  \right]
	  \end{align*}
	  It can be observed that, in going from $B$ to $B_{rref}$, to row-$(\frac{t}{2}+1)$ of $B$, only non-negative linear combinations of the rows above it in $B$ are added, entries of which are either $0$ or $1$. Hence $y_{(\frac{t}{2}+1,\alpha_1)},...,y_{(\frac{t}{2}+1,\alpha_{\frac{t}{2}+1})} \geq 0$. Therefore $r_{\alpha_i} \geq 0$ for $\alpha_i$ all non-basic variables. Hence, by Theorem \ref{thm:optimal basic feasible solution} the basic solution given by \eqref{LP1}, \eqref{LP2} and \eqref{LP3} is an ``optimal basic feasible" solution.\\
	  By the theorem of strong duality (Theorem \ref{thm:strong duality}) the optimal values of objective functions of the primal problem and the dual problem are equal. Therefore the minimum value of $m$ is 
	  \begin{equation}
	  m \geq 3n\lambda_{\frac{t}{2}+1} = \frac{n2\sum_{i = 0}^{\frac{t}{2}-1}r^i}{(r^{\frac{t}{2}}+2\sum_{i = 0}^{\frac{t}{2}-1}r^i)}\label{m_min_value}
	  \end{equation}.
	  Hence we get the upper bound on the rate:
	  \begin{equation*}
	  \frac{k}{n} = 1 - \frac{m}{n} \leq \frac{r^{\frac{t}{2}}}{r^{\frac{t}{2}} + 2 \sum\limits_{i=0}^{\frac{t}{2}-1} r^i}
	  \end{equation*}
	  We now pick a solution for the primal problem and show that it is feasible and gives the optimal objective function value.
	  \begin{align}
	  a_i & = \frac{2nr^i}{r^{\frac{t}{2}}+2\sum_{i = 0}^{\frac{t}{2}-1}r^i},\text{ for }0 \leq i \leq \frac{t}{2}-1\label{optimalai}\\
	  a_{\frac{t}{2}} & = \frac{nr^{\frac{t}{2}}}{r^{\frac{t}{2}}+2\sum_{i = 0}^{\frac{t}{2}-1}r^i},\text{   }
	  p = 0.\label{optimalaip}
	  \end{align}
	  It is easy to check that this solution satisfies the constraints of the primal problem with equality. Therefore the chosen solution is a feasible solution. It is also easy to check that the solution gives the optimal value of the objective function. Hence it is an optimal feasible solution. We thus conclude that a code having the above chosen values will have the optimal rate.
\end{proof}
\begin{note}
A proof along similar lines exists for the case of odd $t$. It is skipped here. 
\end{note}

	  		\section{Coloring \gbase\ with $(r+1)$ Colours} \label{Gbasecolouring_ch4}
	  		
	  		In this section of the appendix, we show how to construct the base graph \gbase\ in such a way that its edges can be colored using $(r+1)$ colors. In our construction, for $t$ odd, we set $a_0=(r+1)$ which is the smallest possible value of $a_0$ by Theorem \ref{thm:t_odd_azero}.  For the case of $t$ even, we set $a_0=4$.   We begin with a key ingredient that we make use of in the construction, namely, that the edges of a regular bipartite graph of degree $d$ can be colored with $d$ colors. 
	  		
	  		\begin{thm} \label{bipartiteColouring}
	  			The edges of a regular bipartite graph $G$ of degree $d$ can be coloured with exactly $d$ colours such that each edge is associated with a colour and adjacent edges of the graph does not have the same colour. 
	  			Let ${\cal E}_i$ be the set of edges of $G$ of color $i$, then ${\cal E}_1,\ldots,{\cal E}_d$ is a collection of $d$ pairwise disjoint set of edges each of which is a perfect matchings of the graph $G$.
	  		\end{thm}
	  		\begin{proof}
	  			By Hall's Theorem (1935) \cite{dieRich}, there exists a perfect matching in the $d$-regular bipartite graph $G$. We first identify this perfect matching and then remove the edges corresponding to this perfect matching from $G$ to form a $(d-1)$-regular bipartite graph to which we can once again apply Hall's Theorem and so on until we have partitioned the edges of the bipartite graph $G$ into the disjoint union of $d$ perfect matchings.  
	  			To conclude the proof, we simply choose $d$ different colors say $\{1,2,\ldots,d\}$ and color each edge in $i$th perfect matching with color $i$, $\forall 1 \leq i \leq d$. 
	  		\end{proof}
	  		
	  		\subsection{The Construction of \gbase\ for $t$ Odd}
	  		
	  		The aim here is to show that we can construct the base graph \gbase\ (recall that \gbase\ has the same structure as \gzero) and color the edges of it using exactly $(r+1)$ colors such that adjacent edges does not have the same color.  
	  		
	  		In the case of $t$ odd, we set $a_0=r+1$.  From Remark~\ref{note:bipartite_count}, it follows that if we set $a_0=(r+1)$, then 
	  		\bean
	  		|V_{s-1}|r & = & |V_{s}|(r+1),
	  		\eean
	  		and hence from \eqref{eq:Vsminus1_ch4}, \eqref{eq:Vs_ch4}, it is possible to connect nodes in $V_{s-1}$ to nodes in $V_s$ so that \gsmone\ is a bipartite graph with the two sets of nodes in \gsmone\  being equal to $V_{s-1},V_s$ respectively where each node in $V_{s-1}$ is of degree $r$ and each node in $V_{s}$ is of degree $(r+1)$. The argument for the construction of such \gsmone\ is straight forward and we skip the description. Since the graph \gbase\ is completely specified once \gsmone\ and $a_0$ are specified as mentioned in Section \ref{tOdd_Graph_Desc}, \gbase\ can be constructed with $a_0=r+1$. Let us recall that if we add a node $V_{\infty}$ to \gbase\ and connect $V_{\infty}$ to all the nodes in $V_0$, we will recover the graph \ginf.  It is easily seen that \ginf\ is an $(r+1)$-regular graph.  By grouping together nodes in alternate layers in \ginf\ i.e., by letting $U_1 = \{V_{\infty}\} \cup V_1 \cup V_3 \cup \ldots$ and $U_2 = V_{0} \cup V_2 \cup V_4 \cup \ldots$, it can be verified that \ginf\ is in fact, an $(r+1)$-regular {\em bipartite} graph with node-set $U_1$ on the left and node-set $U_2$ to the right.  Hence by Theorem \ref{bipartiteColouring}, the edges of the graph \ginf\ and hence the edges of the graph \gbase\ can be colored with exactly $r+1$ colors.
	  		
	  		\subsection{The Construction of \gbase\ for $t$ Even}
	  		Pick an arbitrary graph \gbase\ (recall that \gbase\ has the same structure as \gzero) with $a_0=4$. Note that $s \geq 2$, $r \geq 3$. Since the only freedom lies in edge set $E_{s+1}$, which connects nodes in $V_s$ according to a regular graph of degree $r$, it can be seen that a graph  \gbase\ with $a_0=4$ can be constructed because we can construct a regular graph with degree $r$ with $|V_s|$ nodes as $|V_s|$ is even and $|V_s| \geq r+1$. Since we do not care about girth at this point, the argument for the construction of such regular graph is straightforward and we skip the description. Now we will reconstruct the edge set $E_{s+1}$ so that \gbase\ can be colored with $r+1$ colors. But before that we color the edges other than $E_{s+1}$.
	  		Let $T_v$ be the tree (a subgraph of \gbase) with root node $v \in V_0$ formed by all the paths of length at most $s$ starting from $v$ i.e., $T_v$ is the subgraph of \gbase\ induced by the vertices which are at distance atmost $s$ from $v$ where we remove in this induced subgraph all the edges which are in $E_{s+1}$. 
	  		Let $V(T_v)$ be the vertex set of $T_v$. The nodes $V(T_v) \cap V_i$ are at depth $i$ in $T_v$ with the root node $v$ at depth 0.  We now color the edges of the tree $T_v$ with the $r+1$ colors $\{1,\ldots,r+1\}$. It is clear that such a coloring of edges of $T_v$ can be done since $T_v$ is an $r$-ary tree. Since $T_v$ is an $r$-ary tree, the color of any edge of $T_v$ is frozen to a fixed color once the colors of edges incident on $v$ is frozen to some fixed colors. There are $r$ edges $\{e_1,\ldots,e_r\}$ incident on $v$ in $T_v$. Let the color of $e_i$ be $i$, $\forall i \in [r]$. Hence there is no edge of color $r+1$ incident on $v$. 
	  		\ben
	  		\item Let $j$ be one of the $r$ colors, $1 \leq j \leq r$. Let $X^j_i$ be the largest subset of $V(T_v) \cap V_i$ (nodes at depth $i$) such that each node in $X^j_i$ is connected by an edge of color $j$ to a distinct node at depth $(i-1)$ i.e., a node in $V(T_v) \cap V_{i-1}$. Let $|X^j_i| = x^j_i$. Let $Y_i$ be the largest subset of $V(T_v) \cap V_i$ (nodes at depth $i$) such that each node in $Y_i$ is connected by an edge of color $(r+1)$ to a distinct node in $V(T_v) \cap V_{i-1}$ (nodes at depth $i-1$). Let $|Y_i|=y_i$. It is clear that $x^1_i = x^2_i = \ldots = x^r_i$. We set $x^j_i = x_i$. It can be verified that:
	  		\bean
	  		y_i & = & \left\{ \begin{array}{rl} x_i-1 & \text{if $i$ is odd}, \\  x_i+1 & \text{if $i$ is even} \end{array} \right. ,
	  		\eean
	  		because it can be seen that $x_{i+1}=(r-1)x_i+y_i$ and $y_{i+1} = rx_i$.
	  		Hence $|V(T_v) \cap V_i| = r^i = x_i r + y_i$. It follows that 
	  		\bean
	  		x_i & = & \left\{ \begin{array}{rl} \frac{r^i+1}{r+1} & \text{if $i$ is odd}, \\  \frac{r^i-1}{r+1} & \text{if $i$ is even} \end{array} \right. .
	  		\eean
	  		Since the set $X^j_s$ depends on $v$, from now on we denote it as $X^j_{s,v}$. Since the set $Y_s$ depends on $v$, from now on we denote it as $Y_{s,v}$.    
	  		\item Next, since $a_0=4$, let $V_0=\{v_1,v_2,v_3,v_4\}$. Hence there are $4$ such trees $T_{v_1},T_{v_2},T_{v_3},T_{v_4}$. Color the edges of the tree $T_{v_{\ell}}$ with $r+1$ colors as before such that there is no edge of color $r+1$ incident on $v_{\ell}$, $\forall \ell \in [4]$. So far we have colored all edges except edges in $E_{s+1}$. We now remove all the edges in $E_{s+1}$ and construct another edge set $E_{s+1}$ so that it can be colored such that \gbase\ can be colored with $r+1$ colors. An edge in $E_{s+1}$ which is incident on any node in $X^j_{s,v_{\ell}} (\subseteq V_s)$ cannot be colored with color $j$ but every color from $[r+1] \setminus \{ j\}$ can be used to color it, $\forall \ell \in [4]$, $\forall j \in [r]$.  Similarly an edge in $E_{s+1}$ which is incident on any node in $Y_{s,v_{\ell}} (\subseteq V_s)$ cannot be colored with color $(r+1)$, $\forall \ell \in [4]$. We can connect the set of nodes $\cup_{\ell=1}^{4}X^j_{s,v_{\ell}}$ of size $4x_s$ (or the set of nodes $\cup_{\ell=1}^{4}Y_{s,v_{\ell}}$ nodes of size $4y_s$) to form bipartite graph of degree $r$.  This is possible because we can first construct a regular graph of degree $r$ with $2x_s$ or $2y_s$ nodes (such a regular graph can be constructed as $\min\{2x_s,2y_s\} \geq (r+1)$ because $s \geq 2$, $r \geq 3$. Since we do not care about girth at this point, the argument for the construction of such regular graph is straightforward and we skip the description.), and then create two copies of the vertex set of this regular graph to create the left and right nodes of a bipartite graph $G$.  The edges of the bipartite graph can then be formed by connecting vertices in accordance with the regular graph, i.e., if nodes $w_1$ and $w_2$ were connected in the regular graph, then node $w_1$ on the left is connected to node $w_2$ on the right in the bipartite graph and vice versa. Hence a bipartite graph $G$ of degree $r$ with $4x_s$ or $4y_s$ nodes can be constructed. Now connect the $4x_s$ nodes in $\cup_{\ell=1}^{4}X^j_{s,v_{\ell}}$ in accordance with this bipartite graph $G$ which also has $4x_s$ nodes and color the edges of this graph by colors from $[r+1]\setminus \{j\}$, $\forall j \in [r]$. This is possible by Theorem \ref{bipartiteColouring}. Similarly connect the $4y_s$ nodes in $\cup_{\ell=1}^{4}Y_{s,v_{\ell}}$ in accordance with this bipartite graph $G$ which also has $4y_s$ nodes and color the edges of this graph by colors from $[r+1]\setminus \{r+1\}$. This is again possible by Theorem \ref{bipartiteColouring}. The edges of these $(r+1)$ bipartite graphs form the edge set $E_{s+1}$. The construction is complete as we have constructed \gbase\ and connected the nodes in $V_s = (\cup_{j=1}^{r} \cup_{\ell=1}^{4}X^j_{s,v_{\ell}}) \cup (\cup_{\ell=1}^{4}Y_{s,v_{\ell}})$ according to an edge set $E_{s+1}$ and colored the edges of \gbase\ using $(r+1)$ colors.      
	  		\een

	\bibliographystyle{IEEEtran}
	\bibliography{bib_file}

\end{document}